\documentclass[AMA,Times1COL]{WileyNJDv5}
\usepackage{amsmath}
\usepackage{graphicx,psfrag,epsf}
\usepackage{enumerate}
\usepackage{indentfirst}
\usepackage{amsthm,amsmath,amssymb}
\usepackage{mathrsfs}
\usepackage{threeparttable}
\usepackage{float}
\usepackage{booktabs}
\usepackage{subfigure}
\usepackage[graphicx]{realboxes}
\usepackage{url} 
\usepackage{comment}
\usepackage{cleveref}

\articletype{ORIGINAL ARTICLE}%

\received{Date Month Year}
\revised{Date Month Year}
\accepted{Date Month Year}
\journal{Journal}
\volume{00}
\copyyear{2025}
\startpage{1}

\newcommand\liu[1]{\textcolor{red}{#1}}

\def\var{\textnormal{var}}

\def\ccov{\textnormal{ccov}}
\def\imp{\textnormal{imp}}
\def\mim{\textnormal{mim}}

\def\adj{\textnormal{F}}
\def\int{\textnormal{L}}
\def\tom{\textnormal{T}}
\def\sadj{\textnormal{sF}}
\def\sint{\textnormal{sL}}
\def\stom{\textnormal{sT}}
\def\qk{q_{[k]}}
\def\pk{p_{[k]}}
 
\def\ss{\textnormal{ss}}
\def\sumi{\sum_{i=1}^{n}}
\def\sumk{\sum_{k=1}^{K}}
\def\nk{n_{[k]}}
\def\nkt{n_{[k]1}}
\def\nkc{n_{[k]0}}
\def\var{\textnormal{var}}
\def\cov{\textnormal{cov}}

\begin{document}

\title{Regression adjustment in covariate-adaptive randomized experiments with missing covariates}

\author[1]{Wanjia Fu}

\author[2]{Yingying Ma}

\author[3]{Hanzhong Liu}

\authormark{Fu \textsc{et al.}}
\titlemark{Regression adjustment in covariate-adaptive randomized experiments with missing covariates}

\address[1]{\orgdiv{Department of Statistics and Data Science}, \orgname{Tsinghua University}, \orgaddress{\state{Beijing}, \country{China}}}

\address[2]{\orgdiv{School of Economics and Management}, \orgname{Beihang University}, \orgaddress{\state{Beijing}, \country{China}}}

\address[3]{\orgdiv{Department of Statistics and Data Science}, \orgname{Tsinghua University}, \orgaddress{\state{Beijing}, \country{China}}}

\corres{Hanzhong Liu, Department of Statistics and Data Science, Tsinghua University, China. \email{lhz2016@tsinghua.edu.cn}}


\abstract[Abstract]{Covariate-adaptive randomization is widely used in clinical trials to balance prognostic factors, and regression adjustments are often adopted to further enhance the estimation and inference efficiency. In practice, the covariates may contain missing values. Various methods have been proposed to handle the covariate missing problem under simple randomization. However, the statistical properties of the resulting average treatment effect estimators under stratified randomization, or more generally, covariate-adaptive randomization, remain unclear. To address this issue, we investigate the asymptotic properties of several average treatment effect estimators obtained by combining commonly used missingness processing procedures and regression adjustment methods. Moreover, we derive consistent variance estimators to enable valid inferences. Finally, we conduct a numerical study to evaluate the finite-sample performance of the considered estimators under various sample sizes and numbers of covariates and provide recommendations accordingly. Our analysis is model-free, meaning that the conclusions remain asymptotically valid even in cases of misspecification of the regression model. 
}


\jnlcitation{\cname{%
\author{W. Fu},
\author{Y. Ma}, and
\author{H. Liu}}.
\ctitle{Regression adjustment in covariate-adaptive randomized experiments with missing covariates.} \cjournal{\it Statistics in Medicine.} \cvol{2025;00(00):1--15}.}

\maketitle

\renewcommand\thefootnote{}
\footnotetext{\textbf{Abbreviations:} ATE, average treatment effect; CC, complete case analysis; CCOV, complete covariate analysis; IMP, single imputation; MIM, missing indicator method; MPM, missing pattern method.}

\renewcommand\thefootnote{\fnsymbol{footnote}}
\setcounter{footnote}{1}

\section{Introduction}
\label{sec1}

Originating from the seminal work of \cite{Fisher:1926}, the advocacy for randomization in experimental design emphasized its pivotal role in estimating the treatment effects. Simple randomization, the primary randomization method, assigns treatment levels independently to units based on predetermined probabilities. However, this method may lead to unbalanced treatment assignments within important covariates that are predictive of the outcomes. For example, imagine a scenario where all men are assigned to the treatment group while all women are to the control--an obvious imbalance with potential ramifications. 

To mitigate such risks, various randomization methods have been proposed to balance treatment allocations among covariates, such as biased coin design \citep{Efron:1971}, permutation block randomization \citep{Zelen:1974}, Pocock--Simon's minimization \citep{Pocock:1975}, collectively termed covariate-adaptive randomization. These randomization schemes have been widely used across disciplines like clinical trials, biomedical, and social sciences. \cite{Lin:2015} and \cite{Ciolino:2019} underscored their prevalence, with approximately 80\% of papers in leading medical journals utilizing covariate-adaptive randomization.

In addition to stratification variables, experimenters often collect additional covariates \citep{Liu:2020,Ma:2020,Ye:2020}. To address remaining imbalances, regression adjustment has become a stalwart technique, enhancing estimation and inference efficiency \citep[see, e.g.,][]{Lin:2013,Liu:2020,Shi:2022,Zhao2022reconciling,Zhao2022fact,Wang:2023,Zhao:2023,Liu:2024}. The rigorous establishment of asymptotic properties of treatment effect estimators under covariate-adaptive randomization was pioneered by \cite{Bugni:2018}, followed by significant advancements in the theory of regression adjustment and robust inference \citep{Bugni:2019,Ma:2020,Ye:2020,LiuH:2022, Bai:2022, Bai:2023, Bai:2024}. Robustness refers to the validity of the resulting estimates and inference procedures under model misspecification.

However, practical scenarios often present challenges, such as missing covariate values, complicating the implementation of regression adjustment. To tackle this issue, one approach is to employ parametric methods such as maximum likelihood estimation, Bayesian techniques, or multiple imputations. These methods leverage available information to enhance efficiency \citep{Rubin:1987, Little:1992}. Nevertheless, the effectiveness of such approaches depends heavily on the validity of underlying assumptions, and deviations from these assumptions may result in efficiency loss \citep{White:2010}. Alternatively, nonparametric methods provide a strategy that does not rely on assumptions regarding the missingness mechanism. \cite{Zhao:2022} conducted a comprehensive analysis of missingness processing methods in conjunction with Fisher's additive \citep{Fisher:1935} and Lin's \citep{Lin:2013} with-interaction regression adjustments within the finite population framework; while our paper considers a super-population framework. Under covariate-adaptive randomization, while \cite{Wang:2023} explored the missing outcomes problem, they assumed full observance of covariates. 

There's a dearth of comprehensive investigations into missingness processing methods combined with Fisher's and Lin's regression adjustments under covariate-adaptive randomization within the super-population framework that allows for model-misspecification. Moreover, under complete randomization, Fisher's regression lacks assurances of efficiency gains \citep{Freedman:2008a}, and the incorporation of treatment-by-covariate interactions in Lin's regression significantly increases the degrees of freedom and influences the robustness of resulting treatment effect estimators \citep{Lu:2022}. Here, robustness refers to the estimator's insensitivity to minor changes in the sample \citep{huber2004robust}, an expanding dimension of the covariates, and an unbalanced design (the treated proportion $\ne$ 0.5). This issue becomes more pronounced under covariate-adaptive randomization, especially when considering the treatment-by-stratum interactions with a large number of strata.

To bridge these gaps, our study evaluates eighteen regression-adjusted average treatment effect (ATE) estimators under covariate-adaptive randomization with missing covariates. These estimators are derived from the combination of three commonly used nonparametric missingness processing procedures, along with six regression adjusting methods, encompassing both stratum-common and stratum-specific versions \citep{Ma:2020} of Fisher's, Lin's, and tyranny-of-the-minority (ToM) regressions. It is important to note that the analysis conducted by \cite{Zhao:2022} does not incorporate ToM regression, which constitutes one of the main differences between our work and theirs. ToM regression allocates greater weights to units within smaller sample size groups. While it is asymptotically equivalent to Lin's regression, ToM regression demonstrates enhanced robustness in finite samples, attributed in part to its resolution of nearly half of the degrees of freedom. These properties were scrutinized by \cite{Lu:2022} under complete randomization and stratified randomization without missing covariates within a finite population framework. We extend this theory to encompass more general covariate-adaptive randomization within the super-population framework, accommodating scenarios with missing covariates. Moreover,  we propose consistent variance estimators to facilitate valid inferences. Based on both theoretical insights and finite sample performance, we identify the most efficient estimator and offer practical recommendations regarding the utilization of missingness processing and regression adjustment methods under covariate-adaptive randomization with missing covariates. Importantly, our analysis is model-free, that is, the validity of our results does not require modeling assumptions on the missingness mechanism and allows the linear regression model to be misspecified. Moreover, the recommended estimator can be readily implemented using standard statistical software packages.

The remainder of the paper is organized as follows. In Section \ref{sec2}, we introduce the framework and notation. In Section \ref{sec3}, we introduce the missingness processing procedures. In Section \ref{sec4}, we apply the missingness processing procedures introduced in Section \ref{sec3} to six different regression adjustment methods and present the corresponding theoretical results. In Section \ref{sec5}, we conduct a numerical study to explore the finite sample performance of the considered missingness processing and regression adjusting methods. We conclude the paper in Section \ref{sec7}. Proofs are relegated to the Supplementary Material.

\section{Framework and notation}
\label{sec2}


	Consider a randomized experiment involving $n$ units. For each unit $i$, $i=1,\ldots,n$, denote $X_i = (X_{i1},\ldots,X_{ip})^\top \in \mathbb R^p$ and $Z_i = (Z_{i1},\ldots,Z_{ip_1})^\top \in \mathbb R^{p_1}$ the baseline covariates with $Z_i$ being used for stratification. Let $B$: $\textnormal{support}(Z)\rightarrow\{1,\ldots,K\}$ be the stratum label, which divides all units into $K$ strata according to the values of $Z$. Denote $B_i = B(Z_i)$. We assume that $p$ and $K$ are fixed, and that the probability $p_{[k]}= P(B_i=k)$ of units belonging to stratum $k$ is greater than zero for $k=1,\ldots,K$. In each stratum, units are assigned to the treatment and control groups by a randomization method. Let $A_i$ be the treatment indicator for unit $i$, where $A_i = 1$ represents treatment and $A_i = 0$ represents control. 
Denote $n_{[k]} = \sum_{i \in [k]} 1$, $n_{[k]1} = \sum_{i \in [k]} A_i$ and $n_{[k]0} = \sum_{i \in [k]}(1-A_i)$ the number of units, the number of treated units, and the number of control units in stratum $k$, respectively, where $[k]$ is the index set of units that belong to the $k$th stratum. Let $p_{n[k]} ={n_{[k]}}/{n}$ and $\pi_{n[k]} = n_{[k]1}/n_{[k]}$ be the proportion of stratum-size and the proportion of treated units in stratum $k$. Let $n_1 = \sumi A_i$ and $n_0 = \sumi (1 - A_i)$ denote the total numbers of treated and control units, respectively. We use the Neyman--Rubin potential outcome model \citep{Neyman:1923, Rubin:1974} to define the treatment effect. Let $Y_i (1)$ and $Y_i(0)$ be the potential outcomes of unit $i$ under treatment and control, respectively, and $Y_i=A_i Y_i (1)+(1-A_i)Y_i (0)$ be the observed outcome. We are interested in estimating and inferring the average treatment effect (ATE), $\tau = E\{Y_i(1) - Y_i(0)\}$.

The covariate $X_{ij}$ ($i=1,\ldots,n$, $j=1,\ldots,p$) may have missing values. Let $M_{ij} = A_i M_i(1) + (1 - A_i) M_i(0)$ be the missing indicator of whether $X_{ij}$ is missing or not, where $M_i(1)$ and $M_i(0)$ represent the potential missing indicators under treatment and control, respectively. The observed data is $\{(Y_i,X_i, B_i, A_i,M_i )\}_{i=1}^n$ with $X_i$ being observed if $M_i=0$ and missing otherwise. Denote $W_i = (Y_i(1),Y_i(0),X_i, B_i, M_i(1), M_i(0))^\top$, $A^{(n)} = (A_1,\ldots,A_n)$, $B^{(n)} = (B_1,\ldots,B_n)$, and $W^{(n)} = (W_1,\ldots,W_n)$. Let $D_{n[k]} = \sum_{i=1}^n (A_i-\pi)I_{i \in [k]},$ $k = 1,\ldots,K$, measure the covariate imbalance \citep{Bugni:2018},  where $\pi \in (0,1)$ is the target probability of assigning units to the treatment group and $I$ is the indicator function. Denote $\mathcal{R}_2  =\{V: \max_{k = 1, \ldots, K} \var (V \vert B = k) > 0\}$ the set of random variables with at least one non-zero stratum-specific variance. We make the following assumptions on the data-generating process and treatment assignment mechanism and missingness.


\begin{assumption}
\label{as1}
$\{W_i\}_{i=1}^{n}$ are independent and identically distributed (i.i.d.), $E\{X_{ij}^2\} < \infty$, $E\{ Y_i^2(a) \} < \infty$, and $Y_i(a) \in \mathcal{R}_2$ for $i=1,\ldots,n$, $j=1,\ldots,p$, and $a=0,1$.

\end{assumption}

\begin{assumption}
\label{as2}

$W^{(n)} \perp \!\!\! \perp A^{(n)} \vert B^{(n)}$.

\end{assumption}

\begin{assumption}
\label{as4}

$D_{n[k]} = O_p (\sqrt{n}), ~ k=1,\ldots,K$.

\end{assumption}

\begin{assumption}
\label{as5}

$M_i(1) = M_i(0) =  M_i$.

\end{assumption}

Assumptions \ref{as1}--\ref{as4} are parallel to Assumptions 1--2 and 4 in \cite{Ma:2020}, with the only difference being that the potential missing indicators $M_i(1)$ and $M_i(0)$ are included in $W_i$. 
Assumption~\ref{as2} requires that conditional on the stratum indicators, the treatment assignments are conditionally independent of the potential outcomes, additional covariates $X_i$, and the potential missing indicators. Assumption~\ref{as4} does not necessitate (asymptotic) independence between the treatment assignments across different strata. This allows for the consideration of covariate-adaptive randomization methods with intricate dependence structures across strata, such as minimization \citep{Pocock:1975} and the designs proposed by \cite{Hu2012}. It's worth highlighting that Assumption~\ref{as4} is quite general and is met by most encountered covariate-adaptive randomization methods in practice. 
Assumption \ref{as5} holds automatically if baseline covariates are collected before the treatment assignment, such that they are not affected by treatment. Within a finite-population framework, Assumption \ref{as5} is indispensable for deriving design-based asymptotic properties of ATE estimators obtained by commonly used missingness processing and regression adjustment methods \citep{Zhao:2022}. 
This assumption is also required for model-free analysis of covariate-adaptive randomized experiments under the super-population framework.
 

{\bf Notation.} Let random variable $r_i(a)$ ($i=1,\ldots,n$, $a=0,1$) be a transformed outcome, such as potential outcome $Y_i(a)$, covariates $X_i$ ($X_i(1) = X_i(0) = X_i$ as baseline covariates are not affected by treatment), or their linear transformations. Let $r_i = A_i r_i(1) + (1 - A_i) r_i(0)$. Denote $\bar{r}_1 = ({1}/{n_1})\sum_{i=1}^n A_{i}r_i$ and $\bar{r}_0 = ({1}/{n_0})\sum_{i=1}^n (1-A_{i})r_i$ as the sample means of $r_i(a)$ in the treatment and control groups, respectively, and the corresponding stratum-specific sample means are denoted by $\bar{r}_{[k]1} = ({1}/{n_{[k]1}}) \sum_{i \in [k]} A_{i}r_i$ and $\bar{r}_{[k]0} = ({1}/{n_{[k]0}}) \sum_{i \in [k]}(1-A_{i})r_i$. Let $\bar{r} = (1/n) \sum_{i=1}^{n} r_i$ and $\bar{r}_{[k]} = (1/n_{[k]}) \sum_{i \in [k]} r_i$. For a random variable or vector $V_i$, let $\tilde{V}_i = V_i - E(V_i \mid B_i)$.
In our theoretical results, the asymptotic variance of ATE estimators depends on the following two quantities:
$$
\varsigma_{r}^2(\pi) = \pi^{-1} \var\big[ r_i(1) - E\{ r_i(1) \mid B_i \} \big] +  (1 - \pi)^{-1} \var\big[ r_i(0) - E\{ r_i(0) \mid B_i \} \big],
$$
$$
\varsigma_{Hr}^2 =E \big([ E\{ r_i(1) \mid B_i \} - E\{ r_i(1) \} ] - [ E\{ r_i(0) \mid B_i \} - E\{ r_i(0) \} ] \big)^2.
$$
They can be estimated by the corresponding sample version quantities:
\begin{eqnarray*}
\hat{\varsigma}_{r}^2(\pi) &=& \frac{1}{\pi} \sum_{k=1}^K \Bigl\{\frac{p_{n[k]}}{n_{[k]1}} \sum_{i \in [k]}A_i{(r_i-\bar{r}_{[k]1})}^2\Bigr\} +  \frac{1}{1- \pi} 
 \sum_{k=1}^K \Bigl\{\frac{p_{n[k]}}{n_{[k]0}} \sum_{i \in [k]}(1-A_i)  {(r_i-\bar{r}_{[k]0})}^2\Bigr\},    
\end{eqnarray*}
\begin{eqnarray*}
\hat{\varsigma}_{Hr}^2 = \sum_{k=1}^K p_{n[k]}{\bigl\{(\bar{r}_{[k]1}-\bar{r}_{1})-(\bar{r}_{[k]0}-\bar{r}_{0})\bigr\}}^2. 
\end{eqnarray*}


\section{Missingness Processing}
\label{sec3}

	There are five commonly used missingness processing methods: complete case analysis (CC), complete covariate analysis (CCOV), single imputation (IMP), missing indicator method (MIM), and missing pattern method (MPM). Specifically,

\begin{itemize}
\item CC is the default missingness processing method for various software packages. It uses observations that do not include missing values to estimate the ATE;

\item CCOV uses all covariates that do not include missing values. Recall that $M_{ij}$ is the missing indicator of whether $X_{ij}$ is missing or not, which is equal to one if $X_{ij}$ is missing and zero otherwise. Denote $\mathscr{J} = \{j:M_{ij}=0, ~ i=1,\ldots,n\}$ and $X_{i,\ccov} = {(X_{ij})}_{j \in \mathscr{J}}$ as covariates without missing values. CCOV uses $X_{i,\ccov}$ instead of $X_i$ in further analysis;

\item IMP imputes missing values using observed data. It interpolates the same value $\hat{c}_j$ for the missing items in the $j$th covariate, where $\hat{c}_j$ is a constant or a random variable, such as the mean of the observed covariate values. Denote $\hat{c} = (\hat{c}_1,\ldots,\hat{c}_p)^\top$, and $X_{i,\imp}(\hat{c}) = (X_{i1,\imp} (\hat{c}_1),\ldots,X_{ip,\imp} (\hat{c}_p))^\top$, where $X_{ij,\imp} (\hat{c}_j)  = (1-M_{ij}) X_{ij} + M_{ij} \hat{c}_j$. IMP uses $X_{i,\imp}(\hat{c})$ instead of $X_i$ in further analysis;

\item MIM uses $X_{i,\imp}(\hat{c})$ and all missing indicators $M_{ij}$ as covariates. Here, perfect colinearity is adjusted, and all-zero vectors are deleted. Denote $M_i = (M_{i1}, \ldots, M_{ip})$. MIM uses $(X_{i,\imp}(\hat{c})^\top, M_i^\top)^\top$ instead of $X_i$ in further analysis;

\item MPM is a generalization of MIM. It first conducts a separate regression within each \emph{missing pattern} to obtain an ATE estimator. Then the final ATE estimator is a weighted average of the ATE estimators within each \emph{missing pattern}, with the weights of their proportions. Here, a \emph{missing pattern} is a collection of units with the same pattern of covariate missingness.

\end{itemize}



The CC method may introduce bias into the regression-adjusted ATE estimator, which will be discussed in Section A in the Supplementary Material, while MPM requires the sample size in each missing pattern to be large enough \citep{Zhao:2022}, which is often unrealistic in practice. Thus, we will mainly focus on the remaining three missingness processing methods in this paper.

\section{Regression adjustment}
\label{sec4}

Let $Y_i \sim \alpha + \tau A_i$ denote the ordinary least squares (OLS) fit of regressing $Y_i$ on $A_i$ with an intercept, and $Y_i \overset{w_i}{\sim} \alpha + A_i \tau $ the weighted least squares (WLS) with weight $w_i$. We use ``$+$" or ``$\sum$" if we add more covariates into the regression. Notably, we use the software-type formula to derive point and variance estimators and will evaluate their asymptotic properties in a model-assisted manner, without assuming linear model assumptions.

Denote $\hat \tau_{\adj}$ and $\hat \tau_{\int}$ as the OLS estimators of $\tau$ in the regressions $Y_i \sim \alpha + A_i \tau + \sum_{k=1}^{K-1}\alpha_k  I_{i \in [k]}$ and $Y_i \sim \alpha + A_i \tau + \sum_{k=1}^{K-1}\alpha_k  I_{i \in [k]} + \sum_{k=1}^{K-1} \nu_{k} A_i (I_{i \in [k]} - p_{n[k]} )$, respectively. \cite{Ma:2020} showed that $\hat \tau_{\int}$ is optimal (achieving the smallest asymptotic variance) among the estimators adjusting only stratification variables. Moreover, $\hat \tau_{\adj}$ has the same asymptotic distribution as $\hat \tau_{\int}$ for the case of equal allocation ($\pi=0.5$). The authors then suggested using $\hat \tau_{\adj}$ for equal allocation and $\hat \tau_{\int}$ for unequal allocations when additional covariates are not available. Thus, we consider $\hat \tau_{\adj}$ and $\hat \tau_{\int}$ as the benchmark estimators for equal and unequal allocations, respectively.

Denote $U_{i,\ccov}(c) = X_{i,\ccov}$, $U_{i,\imp}(c) = X_{i,\imp}(c)$, $U_{i,\mim}(c) = (\{X_{i,\imp}(c)\}^\top, M_i)^\top$. If the imputation value is a random variable $\hat{c}$, the above notations have the same form with $c$ replaced by $\hat{c}$. Let $\Sigma_{PQ} = E\bigl[\{P-E(P)\}\{Q-E(Q)\}^\top\bigr]$ denote the covariance between two random variables or vectors $P$ and $Q$. Throughout the paper, we assume that {$U_{i,\dag}(c)$} has a finite second moment and {$\Sigma_{\tilde U_{\dag}(c) \tilde U_{\dag}(c)}$} is strictly positive-definite. Define $\beta_{U_{\dag}(c)}(a) = \Sigma_{\tilde U_{\dag}(c) \tilde U_{\dag}(c)}^{-1} \Sigma_{\tilde U_{\dag}(c) \tilde Y (a)}$ and $r_{i, \dag}(a) = Y_i(a) - U_{i, \dag}(c)^\top \beta_{U_{\dag}(c)}$, where $\beta_{U_{\dag}(c)} = (1 - \pi) \beta_{U_{\dag}(c)}(1) + \pi \beta_{U_{\dag}(c)}(0)$ and $\tilde U_{i,\dag}(c) = U_{i,\dag}(c) - E\{U_{i,\dag}(c) \mid B_i\}$. Let 
$\sigma_{\adj}^2 = \sigma_{\int}^2 = \varsigma_{Y}^2 (\pi) + \varsigma_{HY}^2$ and $\sigma_{\dag}^2 = \varsigma_{r_{\dag}}^2 (\pi) + \varsigma_{Hr_{\dag}}^2$.


\subsection{Fisher's regression for equal allocation}

When additional covariates are available, \cite{Ma:2020} suggested using Fisher's regression under equal allocation ($\pi = 1/2$). In this subsection, we integrate Fisher's regression with three methods for handling missingness data and determine the most efficient approach. The regression formula is given by
$$
Y_i \sim \alpha + A_i \tau + \sum_{k=1}^{K-1}\alpha_k  I_{i \in [k]} + {U_{i,\dag}(\hat c)^\top} \beta, \quad {\dag} \in \{\ccov,\imp,\mim\}.
$$

Let $\hat{\tau}_{\adj,\dag}$ denote the ATE estimators of $\tau$ under three missingness processing methods, where ${\dag} \in \{\ccov,\imp,\mim\}$. Correspondingly, denote the OLS variance estimators (scaled by $n$) as $\hat{\sigma}^2_{\adj, \dag}$.

\begin{proposition}
\label{th1}
Suppose that Assumptions \ref{as1}--\ref{as5} hold, $ r_{i, \dag}(a) \in \mathcal{R}_2$ for $\dag \in \{\ccov,\imp,\mim\}$ and $a=0,1$, $\hat{c} \xrightarrow {P} c$, and $\pi = 1/2$. Then, we have
$\sqrt{n}(\hat{\tau}_{\adj, \dag}-\tau)\xrightarrow{d} N(0,\sigma_{\dag}^2)
$, $\hat{\sigma}_{\adj,\dag}^{2} \xrightarrow{P} \sigma_{\dag}^2$, and $\sigma_{\mim}^2  \leq \sigma_{\imp}^2  \leq  
 \sigma_{\ccov}^2 \leq \sigma_{\adj}^2$.


\end{proposition}

\Cref{th1} indicates that Fisher's regression adjusted estimators combined with CCOV, IMP, and MIM are all consistent and asymptotically normal under the condition $\hat{c} \xrightarrow {P} c$. It is noteworthy that $\sigma_{\adj}^2$ serves as the asymptotic variance of $\hat{\tau}_{\adj}$ (and $\hat{\tau}_{\int}$) when $\pi=1/2$ \citep{Ma:2020}. Therefore, under equal allocation, the asymptotic variances of $\hat{\tau}_{\adj, \dag}$ are all no greater than that of $\hat{\tau}_{\adj}$, demonstrating the advantageous properties of Fisher's regression with or without imputation methods. Among the missingness processing methods, MIM yields the most efficient estimator, followed by IMP and CCOV. Furthermore, since $U_{i,\ccov}(\hat c) \subseteq U_{i,\imp}(\hat{c}) \subseteq U_{i,\mim}(\hat{c})$, augmenting the regression with additional covariates often enhances, or at least does not hurt, the asymptotic efficiency of the regression-adjusted estimators. Additionally, the OLS variance estimators are consistent and can be utilized to construct asymptotically valid confidence intervals or tests.

\begin{remark}
The consistency and asymptotic normality remain valid even when $\pi \neq 1/2$. However, in such cases, Fisher's regression-adjusted estimators may exhibit a larger asymptotic variance than $\hat \tau_{\adj}$, hence they are not recommended.
\end{remark}

\subsection{Lin's regression for unequal allocation}


For scenarios involving unequal allocation, \cite{Ma:2020} advocated for Lin's regression. In this section, we integrate Lin's regression with three methods for handling missing data and determine the most effective approach. For ${\dag} \in \{\ccov,\imp,\mim\}$, Lin's regression is given by
$$
Y_i \sim \alpha + A_i \tau + \sum_{k=1}^{K-1}\alpha_k  I_{i \in [k]} + \sum_{k=1}^{K-1} \nu_{k} A_i (I_{i \in [k]} - p_{n[k]} )+ {U_{i,\dag}(\hat c)^\top} \beta + A_i {\{U_{i,\dag}(\hat c)-\bar{U}_{\dag}(\hat c)\}^\top} \xi.
$$

Let $\hat{\tau}_{\int,\dag}$ denote the OLS estimators of $\tau$ in Lin's regression. $\hat{\tau}_{\int,\dag}$ can be obtained by separately running regressions in the treatment and control groups:
$$
Y_i \sim \alpha(1) + \sum_{k=1}^{K-1}\alpha_k(1)  I_{i \in [k]} + U_{i,\dag}(\hat c)^\top \beta(1), \text{ over } A_i = 1,
$$
$$
Y_i \sim \alpha(0) + \sum_{k=1}^{K-1}\alpha_k(0)  I_{i \in [k]} + U_{i,\dag}(\hat c)^\top \beta(0), \text{ over } A_i = 0.
$$
Let $\hat \beta_{U_{\dag}(\hat c)}(1)$ and $\hat \beta_{U_{\dag}(\hat c)}(0)$ denote the OLS coefficients of $\beta(1)$ and $\beta(0)$ in the respective regressions. Then,
$$
\hat{\tau}_{\int, \dag} = \sum_{k=1}^K p_{n[k]} \bigg\{\bigl[\bar{Y}_{[k]1} - \{\bar{U}_{[k]1, \dag}(\hat c) - \bar{U}_{[k], \dag}(\hat c)\}^\top \hat \beta_{U_{\dag}(\hat c)}(1)\bigr] - \bigl[\bar{Y}_{[k]0} - \{\bar{U}_{[k]0, \dag}(\hat c) - \bar{U}_{[k], \dag}(\hat c)\}^\top \hat{\beta}_{U_{\dag}(\hat c)}(0)\bigr]\biggr\}.
$$
The OLS variance estimators are no longer consistent for unequal allocation \citep{Ma:2020}, necessitating the use of nonparametric plug-in variance estimators. Define $\hat r_{i,\int,\dag}(a) = Y_i(a) - U_{i,\dag}(\hat c) ^\top \hat \beta_{\int, \dag }$, $a=0,1$, $\hat \beta_{\int, \dag } = (1 - \pi) \hat \beta_{U_{\dag}(\hat c)}(1) + \pi \hat \beta_{U_{\dag}(\hat c)}(0) $, and $\hat r_{i,\int,\dag} = A_i \hat r_{i,\int,\dag}(1) + ( 1 - A_i) \hat r_{i,\int,\dag}(0) $. The plug-in variance estimator is $ \hat \sigma^2_{\int,\dag} = \hat{\varsigma}_{\hat{r}_{\int,\dag}}^2 (\pi) + \hat{\varsigma}_{H{ \hat r_{\int,\dag}}}^2$.

\begin{proposition}
\label{th2}

Suppose that Assumptions \ref{as1}--\ref{as5} hold, $ r_{i, \dag}(a) \in \mathcal{R}_2$ for $\dag \in \{\ccov,\imp,\mim\}$ and $a=0,1$, and $\hat{c} \xrightarrow {P} c$. Then, we have
$\sqrt{n}(\hat{\tau}_{\int, \dag}-\tau)\xrightarrow{d} N(0,\sigma_{\dag}^2)
$, $\hat \sigma^2_{\int,\dag} \xrightarrow{P} \sigma_{\dag}^2$, and $\sigma_{\mim}^2  \leq \sigma_{\imp}^2  \leq  
 \sigma_{\ccov}^2 \leq \sigma_{\int}^2$.
\end{proposition}

\Cref{th2} asserts that Lin's regression adjusted estimators $\hat{\tau}_{\int, \dag}$, utilizing three missingness processing methods, are consistent and asymptotically normal. Their asymptotic variances are no greater than that of the benchmark method $\hat{\tau}_{\int}$. Moreover, Lin's regression integrated with MIM attains the smallest asymptotic variance, making it optimal among the considered estimators. Furthermore, the nonparametric variance estimators are consistent. Importantly, these conclusions hold true regardless of whether allocations are equal or unequal.

\begin{remark}
    In finite samples, the variance estimator $\hat{\varsigma}_{\hat{r}_{\int,\dag}}^2 (\pi) + \hat{\varsigma}_{H{\hat{r}_{\int,\dag}}}^2$ may underestimate the asymptotic variance when the dimension of $U_{i,\dag}(\hat c)$ is relatively large compared to the sample size. This limitation can be partially addressed by adjusting the degrees of freedom. Specifically, we adjust $\hat{\varsigma}_{\hat{r}_{\int,\dag}}^2 (\pi)$ by replacing $n_{[k]a},~a=0, 1$, by $n_{[k]a} - p_{[k]a}$, where $p_{[k]a}$ is the dimension of regressor within treatment $a$ and stratum $[k]$. 
\end{remark}

\begin{remark}
    In scenarios with unequal allocation, Lin's regression-adjusted estimators $\hat{\tau}_{\int, \dag}$ typically exhibit greater efficiency compared to the corresponding Fisher's regression-adjusted estimators $\hat{\tau}_{\adj, \dag}$. While under equal allocation, these two regression-adjusted estimators are asymptotically equivalent.
\end{remark}

\subsection{ToM regression for unequal allocation}

Despite Lin's regression generally being more efficient than Fisher's regression in the asymptotic sense, it involves doubling the degrees of freedom used for covariate coefficients and can lead to issues such as extreme calibration weights, large leverage scores, and unstable sample influence curves \citep{Lu:2022}. These characteristics make Lin's regression sensitive to small changes in the sample, increasing dimensions of covariate vectors, and unbalanced designs (i.e., treated proportion $\neq$ 0.5). This problem becomes more serious when using IMP and MIM to address the issue of covariate missingness, as it increases the dimensionality of the covariates. \cite{Lu:2022} investigated a more robust regression method named \textbf{t}yranny-\textbf{o}f-the-\textbf{m}inority (ToM), which assigns greater weights to units in the minority group. ToM regression conserves half of the degrees of freedom and is asymptotically equivalent to Lin's regression in completely randomized experiments without missing covariates. In this section, we investigate the asymptotic properties of ToM regression with three missingness processing methods under covariate-adaptive randomization.

ToM regression is a weighted regression defined as follows:
$$
Y_i \overset{w_i}{\sim} \alpha + A_i \tau + \sum_{k=1}^{K-1}\alpha_k  (I_{i \in [k]} - p_{n[k]} ) + \sum_{k=1}^{K-1} \nu_{k} A_i (I_{i \in [k]} - p_{n[k]} )+ U_{i,\dag}(\hat c) ^\top \beta, \quad \dag \in \{\ccov,\imp,\mim\},
$$
where the weights are given by $w_i = A_i / \pi^{2} + (1-A_i)/(1-\pi)^{2}$. 
When $\pi = 0.5$, ToM regression reduces to Fisher's regression. Let $\hat{\tau}_{\tom,\dag}$ and $\hat{\beta}_{\tom,U_{\dag}(\hat c)}$, $\dag \in \{\ccov,\imp,\mim\}$, denote the WLS estimators of $\tau$ and $\beta$ in ToM regression, respectively. Define $\hat r_{i,\tom,\dag}(a) = Y_i(a) - U_{i,\dag}(\hat c) ^\top \hat{\beta}_{\tom,U_{\dag}(\hat c)}$ for $a=0,1$, and $\hat r_{i,\tom,\dag} = A_i \hat r_{i,\tom,\dag}(1) + ( 1 - A_i) \hat r_{i,\tom,\dag}(0) $. The asymptotic variance can be consistently estimated by $ \hat \sigma^2_{\tom,\dag} = \hat{\varsigma}_{\hat{r}_{\tom,\dag}}^2 (\pi) + \hat{\varsigma}_{H{ \hat r_{\tom,\dag}}}^2$.

\begin{theorem}
\label{th3}
Suppose that Assumptions \ref{as1}--\ref{as5} hold, $ r_{i, \dag}(a) \in \mathcal{R}_2$ for $\dag \in \{\ccov,\imp,\mim\}$ and $a=0,1$, and $\hat{c} \xrightarrow {P} c$. Then, we have
$\sqrt{n}(\hat{\tau}_{\tom, \dag}-\tau)\xrightarrow{d} N(0,\sigma_{\dag}^2)
$, $\hat \sigma^2_{\tom,\dag} \xrightarrow{P} \sigma_{\dag}^2$, and $\sigma_{\mim}^2  \leq \sigma_{\imp}^2  \leq  
 \sigma_{\ccov}^2 \leq \sigma_{\int}^2$.
\end{theorem}

Theorem~\ref{th3} suggests that ToM regression-adjusted estimators share the same asymptotic distributions as Lin's regression-adjusted estimators, thereby extending the conclusions applicable to Lin's regression to ToM regression. Importantly, ToM regression demonstrates superior finite-sample performance compared to Lin's regression, particularly when the dimension of $U_{i,\dag}(\hat c)$ is relatively large compared to the sample size; see our simulations for detailed insights. Therefore, we recommend Fisher's regression for cases of equal allocation ($\pi=1/2$) and ToM regression for cases of unequal allocation ($\pi \neq 1/2$). These results extend the ToM regression theory from a finite-population framework to a super-population framework, and from stratified randomization to more complex, covariate-adaptive randomization methods.

\begin{remark}
    Similar to Lin's regression, we can adjust the degrees of freedom to achieve improved finite-sample performance of the variance estimator.
\end{remark}

\subsection{Stratum-specific regression}

Above, we discussed several regression-adjusted estimators that employ stratum-common adjusted coefficients, where the estimated coefficients of covariates are the same across strata. However, as highlighted by \cite{LiuH:2022}, utilizing stratum-specific regression adjustment can further enhance asymptotic efficiency under covariate-adaptive randomization. Therefore, in this section, we investigate the asymptotic properties of stratum-specific Fisher's, Lin's, and ToM regression-adjusted estimators in the presence of missing covariates. They are defined as follows:

{\bf Stratum-specific Fisher's regression}. We perform Fisher's regression within each stratum: for $k=1,\ldots,K$ and $\dag \in \{\ccov,\imp,\mim\}$,
$$
Y_i \sim \alpha_{[k]} + A_i \tau_{[k]} + U_{i,\dag}(\hat c) ^\top \beta_{[k]}, \quad \textnormal{over } i \in [k].
$$
Let $\hat \tau_{\adj [k],\dag}$ and $\hat \sigma^2_{\adj [k], \dag}$ denote the OLS point estimator and variance estimator  of $\tau_{[k]}$, respectively. The stratum-specific Fisher's regression-adjusted ATE estimator is defined as $\hat \tau_{\adj,\dag, \ss} = \sum_{k=1}^{K} p_{n[k]} \hat \tau_{\adj [k],\dag}$, with an asymptotic variance estimator $ \hat \sigma^2_{\adj,\dag,\ss} =  \sum_{k=1}^{K} p_{n[k]} \hat \sigma^2_{\adj [k], \dag}$.

{\bf Stratum-specific Lin's regression}. We perform Lin's regression within each stratum: for $k=1,\ldots,K$ and $\dag \in \{\ccov,\imp,\mim\}$,
$$
Y_i \sim \alpha_{[k]}(1) + U_{i,\dag}(\hat c) ^\top \beta_{[k]}(1), \text{ over } A_i = 1, \ i \in [k],
$$
$$
Y_i \sim \alpha_{[k]}(0) + U_{i,\dag}(\hat c) ^\top \beta_{[k]}(0), \text{ over } A_i = 0, \ i \in [k].
$$
    Let $\hat \beta_{U_{\dag}(\hat c) [k]}(1)$ and $\hat \beta_{U_{\dag}(\hat c) [k]}(0)$ denote the OLS estimators of $\beta_{[k]}(1)$ and $\beta_{[k]}(0)$, respectively. The stratum-specific Lin's regression-adjusted ATE estimator is defined as

\begin{align*}
\hat{\tau}_{\int,\dag,\ss} = \sum_{k=1}^K p_{n[k]} & \bigg\{\bigl[\bar{Y}_{[k]1} - \{\bar{U}_{[k]1, \dag}(\hat c) - \bar{U}_{[k], \dag}(\hat c)\}^\top \hat \beta_{U_{\dag}(\hat c) [k]}(1)\bigr] \\
&  - \bigl[\bar{Y}_{[k]0} - \{\bar{U}_{[k]0, \dag}(\hat c) - \bar{U}_{[k], \dag}(\hat c)\}^\top \hat \beta_{U_{\dag}(\hat c) [k]}(0)\bigr]\biggr\}.
\end{align*}

Define $\hat r_{i,\int,\dag, \ss}(a) = Y_i(a) - U_{i,\dag}(\hat c)^\top \hat \beta_{U_{\dag}(\hat c) [B_i]}$ for $a=0,1$, where $\hat \beta_{U_{\dag}(\hat c) [B_i]} = (1 - \pi) \hat \beta_{U_{\dag}(\hat c) [B_i]}(1) + \pi \hat \beta_{U_{\dag}(\hat c) [B_i]}(0) $, and $\hat r_{i,\int,\dag, \ss} = A_i \hat r_{i,\int,\dag, \ss}(1) + ( 1 - A_i) \hat r_{i,\int,\dag, \ss}(0) $. The asymptotic variance estimator is defined as $\hat{\varsigma}_{\hat{r}_{\int,\dag, \ss}}^2 (\pi) + \hat{\varsigma}_{H{ \hat r_{\int,\dag, \ss}}}^2 $.

{\bf Stratum-specific ToM regression}. We perform ToM regression within each stratum: for $k=1,\ldots,K$ and $\dag \in \{\ccov,\imp,\mim\}$,
$$
Y_i \overset{w_i}{\sim} \alpha_{[k]} + A_i \tau_{[k]} + U_{i,\dag}(\hat c)^\top \beta_{[k]},  \text{ over } \ i \in [k],
$$
where the weights are given by $w_i = A_i / \pi^{2} + (1-A_i)/(1-\pi)^{2}$. Let $\hat{\tau}_{\tom [k], \dag} $ and $\hat{\beta}_{\tom, U_{\dag}(\hat c) [k]}$ denote the WLS estimators of $\tau_{[k]}$ and $\beta_{[k]}$, respectively. The stratum-specific ToM regression-adjusted ATE estimator is 
$
\hat{\tau}_{\tom,\dag,\ss} = \sum_{k=1}^{K} p_{n[k]} \hat{\tau}_{\tom [k], \dag}.
$
Define the residual $\hat r_{i,\tom,\dag, \ss}(a) = Y_i(a) - U_{i,\dag}(\hat c)^\top \hat{\beta}_{\tom, U_{\dag}(\hat c) [B_i]} $, $a=0,1$, and $\hat r_{i,\tom,\dag, \ss} = A_i \hat r_{i,\tom,\dag, \ss}(1) + ( 1 - A_i) \hat r_{i,\tom,\dag, \ss}(0) $. Define $ \hat \sigma^2_{\tom,\dag, \ss} = \hat{\varsigma}_{\hat{r}_{\tom,\dag, \ss}}^2 (\pi) + \hat{\varsigma}_{H{ \hat r_{\tom,\dag, \ss}}}^2$.

To examine the asymptotic properties of stratum-specific regression-adjusted estimators, we introduce several notations. Let $\Sigma_{PQ[k]} = E\bigl[\{P-E(P\vert B_i = k)\}\{Q-E(Q\vert B_i = k)\}^\top\vert B_i = k\bigr]$ denote the stratum-specific covariance between two random variables or vectors $P$ and $Q$ in stratum $k$. Assume that $ \Sigma_{U_{\dag}(c) U_{\dag}(c)[k]}$ is strictly positive-definite. Define the population regression coefficient $\beta_{U_{\dag}(c)[k]}(a) = \Sigma_{U_{\dag}(c) U_{\dag}(c)[k]}^{-1} \Sigma_{U_{\dag}(c) Y (a)[k]}$ and the residuals $r_{i, \dag, \ss}(a) = Y_i(a) - U_{i, \dag}(c)^\top \beta_{U_{\dag}(c)[k]}$, $i \in [k]$, where $\beta_{U_{\dag}(c)[k]} = (1 - \pi) \beta_{U_{\dag}(c)[k]}(1) + \pi \beta_{U_{\dag}(c)[k]}(0)$. Finally, define $\sigma_{\dag,\ss}^2 = \varsigma_{r_{\dag,\ss}}^2 (\pi) + \varsigma_{Hr_{\dag,\ss}}^2$.

\begin{theorem}
\label{th4}

Suppose that Assumptions \ref{as1}--\ref{as5} hold, $ r_{i, \dag, \ss}(a) \in \mathcal{R}_2$ for $\dag \in \{\ccov,\imp,\mim\}$ and $a=0,1$, and $\hat{c} \xrightarrow {P} c$. For $\dag \in \{\ccov,\imp,\mim\}$, we have
\begin{itemize}
\item[(i)] If $\pi = 1/2$, then $\sqrt{n}(\hat{\tau}_{\adj, \dag, \ss}-\tau)\xrightarrow{d} N(0,\sigma_{\dag,\ss}^2)
$ and $\hat \sigma^2_{\adj,\dag, \ss} \xrightarrow{P} \sigma_{\dag,\ss}^2$;
\item[(ii)] $\sqrt{n}(\hat{\tau}_{\int, \dag, \ss}-\tau)\xrightarrow{d} N(0,\sigma_{\dag,\ss}^2)
$ and $\hat \sigma^2_{\int,\dag, \ss} \xrightarrow{P} \sigma_{\dag,\ss}^2$;
\item[(iii)] $\sqrt{n}(\hat{\tau}_{\tom, \dag, \ss}-\tau)\xrightarrow{d} N(0,\sigma_{\dag,\ss}^2)
$ and $\hat \sigma^2_{\tom,\dag, \ss} \xrightarrow{P} \sigma_{\dag,\ss}^2$;
\item[(iv)] $\sigma_{\mim, \ss}^2  \leq \sigma_{\imp, \ss}^2  \leq  
 \sigma_{\ccov, \ss}^2 \leq \sigma_{\int}^2$ and $\sigma_{\dag, \ss}^2  \leq \sigma_{\dag}^2$.
\end{itemize}
\end{theorem}

\begin{remark}
Similar to scenarios without covariate missing \citep{Ma:2020}, the OLS variance estimators for Lin's regression can be anti-conservative and therefore are not recommended. This conclusion also applies to ToM regression. In cases of unequal allocation, although Fisher's regression-adjusted estimators remain asymptotically normal, their asymptotic variances can be larger than those of the unadjusted estimators.
\end{remark}

Theorem~\ref{th4}(i)--(iii) parallels \Cref{th1}, \Cref{th2}, and \Cref{th3}, demonstrating that all regression-adjusted ATE estimators using three missingness processing methods are consistent and asymptotically normal, with proposed variance estimators showing consistency. Under equal allocation, all three regression types are asymptotically equivalent. Given the simplicity of Fisher's regression and its readily available variance estimator, we recommend its use in this scenario. However, for unequal allocation, only Lin's regression and ToM regression exhibit asymptotic variances that are no greater than those of Fisher's regression and $\hat \tau_{\int}$. Since ToM regression uses approximately half degrees of freedom and offers greater robustness compared to Lin's regression \citep{Lu:2022}, we recommend its use in cases of unequal allocation. 
Theorem~\ref{th4}(iv) indicates that stratum-specific regression-adjusted estimators generally exhibit greater efficiency than their stratum-common counterparts in the asymptotic sense. Furthermore, regression adjustment typically enhances asymptotic efficiency, especially with the inclusion of more additional covariates. Therefore, MIM yields the most efficient estimators. Additionally, regression with MIM remains invariant to the choice of $\hat{c}$, as demonstrated by Proposition~\ref{ps1} below. For these reasons, we recommend using MIM to address covariate missing problems under covariate-adaptive randomization. It should be noted, however, that MIM might underperform compared to IMP in finite samples, particularly when the number of covariates is large relative to the sample size. In such cases, IMP would be the preferable option.

Let $\hat{\tau}_{\diamond,\mim}(\hat{c})$ and  $\hat{\tau}_{\diamond,\mim,\ss}(\hat{c})$ ($\diamond \in \{\adj,\int,\tom \}$) denote the regression-adjusted ATE estimators integrated with MIM, where missing values are imputed using $\hat{c}$.



\begin{proposition}
\label{ps1}

For $\diamond \in \{\adj,\int,\tom \}$, $\hat{\tau}_{\diamond,\mim}(\hat{c})$ and $\hat{\tau}_{\diamond,\mim,\ss}(\hat{c})$ are invariant to the choice of $\hat{c}$, regardless of whether $\hat{c}$ is a constant or a random variable.

\end{proposition}

Proposition~\ref{ps1} extends the findings of \cite{Zhao:2022} from a finite-population framework to covariate-adaptive randomization within a super-population framework. This proposition guarantees that we can impute missing covariate values using either the observed sample mean or simply a constant of $0$ for the MIM method.


\section{Numerical study}
\label{sec5}

\subsection{Simulated data}
\label{sec5:1}
	In this section, we conduct simulation studies to evaluate the finite-sample properties of the proposed estimators. The outcomes are generated as follows: for $i=1,\ldots,n$,
\[
	Y_i(1) = \mu_1 + 5\xi_i + g_1(X_i) + \sigma_1\epsilon_{1, i}, \quad
	Y_i(0) = \mu_0 + g_0(X_i) + \sigma_0\epsilon_{0, i},
\]
where $(X_i, \epsilon_{1, i}, \epsilon_{0, i})_{i=1}^n$ are i.i.d. and $\epsilon_{1, i}$, $\epsilon_{0, i}$, and $X_i$ are mutually independent. We set $\mu_1 = 2$, $\mu_0 = 1$, and  $n=200$. Three models are used to generate the covariates $X_i$ and $g_a(X_i)$, which are detailed below. The first two covariates are used for stratification, and the first three are fully observed, while the others include missing values. The binary random variable $\xi_i \sim B(1, 0.2)$ represents the propensity for missingness, where $B(n, p)$ denotes the binomial distribution with parameters $n$ and $p$. Specifically, we generate missing indicators as $M_{ij} \sim B(1, 0.5\xi_i + 0.05)$ for $j=4,\ldots,7$, where $M_{ij} = 1$ indicates that $X_{ij}$ is missing. The noise terms are generated as $\epsilon_{a, i} \sim N(0, \sigma_a^2)$ with $\sigma_a$ chosen to achieve signal-to-noise ratios (SNRs) of $\mathrm{SNR}_1 = 3$ in the treatment group and $\mathrm{SNR}_0 = 1$ in the control group.

Model 1: Stratum-specific model:
\[
	g_1(X_i) = \beta_{11} X_{i1} + \beta_{12} X_{i2} + \beta_{13} X_{i3} + \beta_{14} X_{i1}X_{i4} + \beta_{15} X_{i2}X_{i5},
\]
\[
	g_0(X_i) = \beta_{01} X_{i1} + \beta_{02} X_{i2} + \beta_{04} X_{i1}X_{i4} + \beta_{05} X_{i2}X_{i5},
\]
where $X_{i1}$ take values uniformly from $\{-2, -1, 1, 2\}$, $X_{i2}$ take values from $\{1, 2, 3\}$ with probabilities $\{0.2, 0.3, 0.5\}$, respectively, $X_{i3} \sim \mathcal{E}(\xi_i + 1)$ where $\mathcal{E}(\lambda)$ denotes the exponential distribution with rate parameter $\lambda$, $(X_{i4}, X_{i5})^\top \sim N( (1+\xi_i, \xi_i)^\top, (4 \ 1;   1 \ 1) )$,
 $\beta_{0j} \sim U(-2, 2)$, $\beta_{1j} \sim \beta_{0j} + U(0, 1)$ for $j = 1, 2, 4, 5$, and $\beta_{13} \sim U(0, 10)$. Here, the coefficients are generated only once. All observations are stratified into four strata based on whether $X_{i1} \in \{1, 2\}$ and whether $X_{i2} \in \{1, 2\}$.

	Model 2: Stratum-common linear model:
\[
	g_0(X_i) = g_1(X_i) = \sum_{j = 1}^{5} \beta_j X_{ij},
\]
where $\beta_j \sim U(-2, 2), ~j = 1, \ldots, 5$, with other settings identical to those in Model 1.

	Model 3: Non-linear model:
\[
	g_1(X_i) = \alpha_1 X_{i1} + \alpha_2 X_{i2}^2X_{i3} + X_{i2} \log(\alpha_3 X_{i3} \log(X_{i4}+1)+1) + \alpha_4 e^{X_{i5}},
\]
\[
	g_0(X_i) = \alpha_1 X_{i1}^2 + \alpha_2 X_{i2}X_{i3} + X_{i2}\log(\alpha_3 X_{i3}\log(X_{i4}+1)+1),
\]
where $(X_{i4}, X_{i5})$ are generated as follows: first, generate $(X_{i40}, X_{i5})^\top \sim N( (1, \xi_i)^\top, (4 \ 1; 1 \ 1) )$ and then set $X_{i4} = \max(X_{i40}, 0) + \xi_i$. The coefficients are generated as $\alpha_1, \alpha_2 \sim U(-2, 2)$, $\alpha_3 \sim U(0, 4)$, $\alpha_4 \sim U(-0.4, 0.4)$. Other settings remain identical to those in Model 1.

We evaluate the finite-sample performance of regression adjustment methods integrated with CCOV, IMP, and MIM under simple randomization, stratified block randomization, and minimization. The target treated proportion $\pi$ is set to $1/2$ (equal allocation) and $2/3$ (unequal allocation). The number of covariates used in regression adjustment, denoted as $p$, is set to $5$ and $7$. When $p = 5$, all covariates generating the outcome data are included. For $p = 7$, two additional i.i.d. covariates $X_{i6} \sim N(0, 1)$ and $X_{i7} \sim N(1, 4)$, which are independent of the outcomes, are added to the regression. We repeat the simulation $10,000$ times to compute the bias, standard deviation (SD), standard error estimator (SE), root mean squared error (RMSE), and coverage probability (CP) of $95\%$ confidence intervals.


Tables \ref{tab1}--\ref{tab4} present the results under stratified block randomization with $n = 200$. Due to space constraints, results for other randomization methods are provided in the Supplementary Material. Our findings are summarized as follows: First, all regression-adjusted estimators show a small finite-sample bias, which is negligible compared to the SD. Second, IMP exhibits smaller SD and RMSE than CCOV, with reductions of approximately $7\%$ to $7.9\%$ relative to the benchmark estimators $\hat{\tau}_{\adj}$ for equal allocation and $\hat{\tau}_{\int}$ for unequal allocation. Third, MIM performs the best when the number of covariates is small ($p=5$), but occasionally shows slightly worse performance than IMP with larger numbers of covariates ($p=7$). This is due to the potential instability of coefficient estimators when using many uncorrelated covariates in regression. Fourth, in Model 1 where true regression coefficients vary across strata, stratum-specific regression-adjusted estimators can enhance precision compared to stratum-common estimators. However, in Model 2 where regression coefficients are stratum-common, the efficiency gain is minimal. Notably, with $p=7$, stratum-specific estimators can even perform slightly worse than their stratum-common counterparts, for similar reasons as seen in the IMP versus MIM comparison. Fifth, for unequal allocation with $p=7$, ToM regression improves precision relative to Lin's regression. This finding also holds for stratum-specific regressions, indicating that ToM regression is more robust with increasing covariate dimensions. Sixth, in Model 3, these conclusions generally hold, although the strong non-linearity reduces the percentage of interpretability of SD by linear models. Therefore, efficiency gains from MIM to IMP and from stratum-specific to stratum-common estimators are less pronounced relative to overall SD. Exploring non-linear covariate adjustment methods would be an interesting future research direction. Finally, the OLS variance estimator for Fisher's regression under $\pi = 1/2$, and the plug-in variance estimator for Lin's regression and ToM regression, closely approximate the true variance, resulting in empirical coverage probabilities close to the confidence level.

\subsection{Application-based analysis}
\label{sec6}

The NIDA-CNT-0027 dataset \citep{Saxon:2013} was collected from a stratified randomized trial that compared changes in liver enzymes and other health indicators associated with treatment using buprenorphine/naloxone (BUP/NX) versus methadone (MET) in an outpatient setting over 24 weeks. The participants, all with opioid dependence, were assessed for changes in the Clinical Opiate Withdrawal Scale (COWS), which measures signs and symptoms of opiate withdrawal. These assessments were conducted weekly from baseline (week 0) through the end of the study (week 24). We focus on the COWS at week 24 as the outcome, using the COWS at week 0, high blood pressure, medical and psychiatric history, the total score of the risk behavior survey, and body temperature at week 0 as the baseline covariates. In the experiment, a total of 1297 patients were stratified into two strata based on their liver enzyme levels.

To evaluate the finite-sample performance of different methods, we impute the unobserved potential outcomes using a random forest fit of the observed data. We then take a random sample of size $n = 500$ with replacement to align with our super-population setting and assign treatments using stratified block randomization with $\pi = 2/3$. The missing values for the covariates—temperature and COWS at week 0—are generated using the method described in Section~\ref{sec5:1}. This process is repeated 1000 times. The results are presented in Table~\ref{TAB0}. The conclusions are similar to those in Section~\ref{sec5:1}. In particular, IMP improves the precision of estimation by 8.5\% compared to the baseline estimator. It is worth noting that MIM performs slightly worse than IMP, and the stratum-specific estimators are also marginally less accurate than their stratum-common counterparts in this specific problem.


\section{Discussions}
\label{sec7}

In this paper, we examined the complexities and advancements of regression adjustment methods under covariate-adaptive randomization, with a particular focus on scenarios involving missing covariate data. Building upon foundational works advocating for randomization by Fisher and subsequent developments in covariate-balanced designs, we reviewed various strategies to enhance the precision and robustness of treatment effect estimation. Our evaluation encompassed eighteen regression-adjusted estimators, combining three nonparametric missingness processing techniques with stratum-common and stratum-specific Fisher's, Lin's, and ToM regressions. This comprehensive approach allowed us to address critical gaps in the literature, particularly concerning the integration of missingness processing methods with ToM regression adjustments under covariate-adaptive randomization within the super-population framework. The key findings underscore the efficacy of IMP in reducing the standard deviation and the root mean squared error relative to CCOV. Additionally, MIM demonstrates superior performance with fewer covariates, but shows increased variability in more complex models, underscoring the sensitivity of coefficient estimators to model complexity. Stratum-specific regression adjustments provide notable efficiency gains in scenarios with varied regression coefficients across strata when the sample size is large. Interestingly, ToM regression emerges as a robust alternative to Lin's regression, particularly in handling larger covariate dimensions and unbalanced designs. Based on theoretical analysis and simulation results, we make the following recommendations: For equal allocation, we generally recommend using Fisher's regression integrated with MIM. In cases of unequal allocation, we advise employing ToM regression. For large sample sizes, we advocate using stratum-specific estimators and MIM for handling missing data. Conversely, for small sample sizes, we recommend using IMP with stratum-common estimators.


Several avenues warrant further investigation. Firstly, this paper primarily focuses on three methods for handling missing covariates (CCOV, IMP, and MIM). Additionally, MPM, another method, shows potential for enhancing efficiency within a finite-population framework \citep{Zhao:2022}. However, the dimensionality of covariates post-imputation via MPM may become excessively large, potentially compromising the finite-sample performance of regression adjustments. This issue is particularly pronounced in covariate-adaptive randomization with many strata. Exploring techniques such as Lasso \citep{Tibshirani:1996} to address high-dimensional challenges could be fruitful. Secondly, while this paper focuses on treatment-control studies, extending these findings to scenarios involving multiple treatments \citep{Bugni:2019, Gu2024} is an intriguing prospect. Thirdly, investigating non-linear covariate adjustment methods and their implications for treatment effect estimation remains a promising area for future research. Lastly, this paper only addresses missing covariate data; however, in practice, outcomes may also be affected by missing values \citep{Zhao-Ding-biometrika}. Exploring strategies to handle missing outcome data under covariate-adaptive randomization presents a compelling direction for future research.


\bibliography{main}

\bmsection*{Supporting information}

Additional supporting information can be found in the
online version of the article at the publisher’s website.

\newpage

\begin{table}[H]
\tiny
\centering
\caption{Simulation results under stratified block randomization, $p = 5$, $\pi = 1/2$}
\label{tab1}
\begin{threeparttable}
\begin{tabular}{cccccccccccccccc}
\toprule
\multicolumn{3}{c}{~} & \multicolumn{4}{c}{Model 1} & \multicolumn{4}{c}{Model 2} & \multicolumn{4}{c}{Model 3} \\ \hline
estimator	&	bias	&	SD	&	SE	&	RMSE	&	CP	&	bias	&	SD	&	SE	&	RMSE	&	CP	&	bias	&	SD	&	SE	&	RMSE	&	CP	\\ \midrule
$\hat{\tau}_{\adj}$	&	0.01	&	1.27	&	1.27	&	1.27	&	0.95	&	0.00&	0.83	&	0.83	&	0.83	&	0.95	&	0.00	&	0.44	&	0.44	&	0.44	&	0.95	\\
$\hat{\tau}_{\adj,\ccov}$	&	-0.01	&	1.27	&	1.26	&	1.27	&	0.95	&	0.00	&	0.79	&	0.78	&	0.79	&	0.95	&	-0.01	&	0.40&	0.40&	0.40&	0.95	\\
$\hat{\tau}_{\int, \ccov}$	&	0.01	&	1.27	&	1.29	&	1.27	&	0.95	&	0.00&	0.79	&	0.80&	0.79	&	0.95	&	0.00&	0.40&	0.41	&	0.40&	0.95	\\
$\hat{\tau}_{\tom, \ccov}$	&	-0.01	&	1.27	&	1.29	&	1.27	&	0.95	&	0.00&	0.79	&	0.80&	0.79	&	0.95	&	-0.01	&	0.40&	0.41	&	0.40&	0.95	\\
$\hat{\tau}_{\adj, \ccov, \ss}$	&	-0.08	&	1.28	&	1.27	&	1.28	&	0.95	&	0.01	&	0.79	&	0.79	&	0.79	&	0.95	&	-0.03	&	0.39	&	0.38	&	0.39	&	0.94	\\
$\hat{\tau}_{\int, \ccov, \ss}$	&	0.00&	1.28	&	1.28	&	1.28	&	0.95	&	-0.01	&	0.79	&	0.80&	0.79	&	0.95	&	-0.01	&	0.39	&	0.39	&	0.39	&	0.95	\\
$\hat{\tau}_{\tom, \ccov, \ss}$	&	-0.08	&	1.28	&	1.28	&	1.28	&	0.95	&	0.01	&	0.79	&	0.79	&	0.79	&	0.95	&	-0.03	&	0.39	&	0.39	&	0.39	&	0.95	\\
$\hat{\tau}_{\adj, \imp}$	&	-0.01	&	1.17	&	1.17	&	1.17	&	0.95	&	0.01	&	0.60&	0.60&	0.60&	0.95	&	-0.01	&	0.39	&	0.38	&	0.39	&	0.95	\\
$\hat{\tau}_{\int, \imp}$	&	0.01	&	1.17	&	1.24	&	1.17	&	0.96	&	0.01	&	0.60&	0.64	&	0.60&	0.96	&	0.00&	0.39	&	0.41	&	0.39	&	0.96	\\
$\hat{\tau}_{\tom, \imp}$	&	-0.01	&	1.17	&	1.24	&	1.17	&	0.96	&	0.01	&	0.60&	0.64	&	0.60&	0.96	&	-0.01	&	0.39	&	0.41	&	0.39	&	0.96	\\
$\hat{\tau}_{\adj, \imp, \ss}$	&	-0.07	&	1.00&	0.99	&	1.00&	0.95	&	0.02	&	0.62	&	0.62	&	0.62	&	0.95	&	-0.04	&	0.38	&	0.37	&	0.38	&	0.94	\\
$\hat{\tau}_{\int, \imp, \ss}$	&	0.02	&	1.00&	1.01	&	1.00&	0.95	&	0.01	&	0.62	&	0.63	&	0.62	&	0.95	&	0.00&	0.39	&	0.39	&	0.39	&	0.95	\\
$\hat{\tau}_{\tom, \imp, \ss}$	&	-0.07	&	1.00&	1.00&	1.00&	0.95	&	0.02	&	0.62	&	0.62	&	0.62	&	0.95	&	-0.04	&	0.38	&	0.38	&	0.38	&	0.95	\\
$\hat{\tau}_{\adj, \mim}$	&	-0.03	&	1.16	&	1.16	&	1.16	&	0.95	&	-0.01	&	0.58	&	0.58	&	0.58	&	0.95	&	-0.02	&	0.37	&	0.37	&	0.38	&	0.94	\\
$\hat{\tau}_{\int, \mim}$	&	0.01	&	1.17	&	1.28	&	1.17	&	0.97	&	0.00&	0.58	&	0.64	&	0.58	&	0.97	&	0.00&	0.37	&	0.41	&	0.37	&	0.97	\\
$\hat{\tau}_{\tom, \mim}$	&	-0.03	&	1.16	&	1.27	&	1.16	&	0.97	&	-0.01	&	0.58	&	0.64	&	0.58	&	0.97	&	-0.02	&	0.37	&	0.41	&	0.38	&	0.97	\\
$\hat{\tau}_{\adj, \mim, \ss}$	&	-0.15	&	0.99	&	0.98	&	1.00&	0.94	&	-0.06	&	0.60&	0.60&	0.61	&	0.95	&	-0.12	&	0.38	&	0.36	&	0.40&	0.92	\\
$\hat{\tau}_{\int, \mim, \ss}$	&	0.00&	1.01	&	1.01	&	1.01	&	0.95	&	-0.01	&	0.62	&	0.62	&	0.62	&	0.95	&	-0.02	&	0.38	&	0.38	&	0.38	&	0.95	\\
$\hat{\tau}_{\tom, \mim, \ss}$	&	-0.15	&	0.99	&	0.99	&	1.00&	0.95	&	-0.06	&	0.60&	0.61	&	0.61	&	0.95	&	-0.12	&	0.38	&	0.37	&	0.40&	0.93	\\ \bottomrule
\end{tabular}

\end{threeparttable}
\end{table}

\begin{table}[H]
\tiny
\centering
\caption{Simulation results under stratified block randomization, $p = 5$, $\pi = 2/3$}
\label{tab2}
\begin{threeparttable}
\begin{tabular}{cccccccccccccccc}
\toprule
\multicolumn{3}{c}{~} & \multicolumn{4}{c}{Model 1} & \multicolumn{4}{c}{Model 2} & \multicolumn{4}{c}{Model 3} \\ \hline
estimator	&	bias	&	SD	&	SE	&	RMSE	&	CP	&	bias	&	SD	&	SE	&	RMSE	&	CP	&	bias	&	SD	&	SE	&	RMSE	&	CP	\\ \midrule
$\hat{\tau}_{\int}$	&	-0.03	&	1.29	&	1.27	&	1.29	&	0.95	&	-0.01	&	0.84	&	0.83	&	0.84	&	0.95	&	-0.01	&	0.44	&	0.45	&	0.44	&	0.95	\\
$\hat{\tau}_{\adj,\ccov}$	&	-0.05	&	1.28	&	1.27	&	1.29	&	0.95	&	-0.01	&	0.79	&	0.79	&	0.79	&	0.95	&	-0.02	&	0.40&	0.40&	0.40&	0.95	\\
$\hat{\tau}_{\int, \ccov}$	&	-0.03	&	1.28	&	1.29	&	1.28	&	0.95	&	-0.01	&	0.79	&	0.80&	0.79	&	0.95	&	-0.01	&	0.40&	0.41	&	0.40&	0.95	\\
$\hat{\tau}_{\tom, \ccov}$	&	-0.05	&	1.28	&	1.29	&	1.29	&	0.95	&	-0.01	&	0.79	&	0.80&	0.79	&	0.95	&	-0.02	&	0.40&	0.41	&	0.40&	0.95	\\
$\hat{\tau}_{\adj, \ccov, \ss}$	&	-0.12	&	1.29	&	1.28	&	1.30&	0.95	&	0.00&	0.80&	0.80&	0.80&	0.95	&	-0.04	&	0.39	&	0.39	&	0.40&	0.94	\\
$\hat{\tau}_{\int, \ccov, \ss}$	&	-0.04	&	1.30&	1.29	&	1.30&	0.95	&	-0.02	&	0.80&	0.80&	0.80&	0.95	&	-0.02	&	0.39	&	0.39	&	0.39	&	0.95	\\
$\hat{\tau}_{\tom, \ccov, \ss}$	&	-0.12	&	1.29	&	1.28	&	1.30&	0.95	&	0.00&	0.80&	0.80&	0.80&	0.95	&	-0.04	&	0.39	&	0.39	&	0.40&	0.94	\\
$\hat{\tau}_{\adj, \imp}$	&	-0.04	&	1.19	&	1.18	&	1.19	&	0.95	&	-0.01	&	0.61	&	0.61	&	0.61	&	0.95	&	-0.02	&	0.39	&	0.39	&	0.39	&	0.95	\\
$\hat{\tau}_{\int, \imp}$	&	-0.02	&	1.20&	1.24	&	1.20&	0.95	&	-0.01	&	0.61	&	0.64	&	0.61	&	0.96	&	-0.01	&	0.39	&	0.41	&	0.39	&	0.96	\\
$\hat{\tau}_{\tom, \imp}$	&	-0.04	&	1.20&	1.24	&	1.20&	0.95	&	-0.01	&	0.61	&	0.64	&	0.61	&	0.96	&	-0.02	&	0.39	&	0.41	&	0.39	&	0.96	\\
$\hat{\tau}_{\adj, \imp, \ss}$	&	-0.10&	1.01	&	1.00&	1.01	&	0.94	&	0.01	&	0.62	&	0.62	&	0.62	&	0.95	&	-0.05	&	0.39	&	0.38	&	0.39	&	0.94	\\
$\hat{\tau}_{\int, \imp, \ss}$	&	-0.01	&	1.02	&	1.01	&	1.02	&	0.95	&	-0.01	&	0.63	&	0.63	&	0.63	&	0.95	&	-0.01	&	0.39	&	0.39	&	0.39	&	0.95	\\
$\hat{\tau}_{\tom, \imp, \ss}$	&	-0.10&	1.01	&	1.00&	1.01	&	0.94	&	0.01	&	0.62	&	0.62	&	0.63	&	0.95	&	-0.05	&	0.39	&	0.38	&	0.39	&	0.94	\\
$\hat{\tau}_{\adj, \mim}$	&	-0.06	&	1.18	&	1.17	&	1.18	&	0.94	&	-0.02	&	0.59	&	0.59	&	0.59	&	0.95	&	-0.04	&	0.38	&	0.37	&	0.38	&	0.94	\\
$\hat{\tau}_{\int, \mim}$	&	-0.02	&	1.19	&	1.28	&	1.19	&	0.96	&	-0.01	&	0.59	&	0.64	&	0.59	&	0.97	&	-0.01	&	0.38	&	0.41	&	0.38	&	0.96	\\
$\hat{\tau}_{\tom, \mim}$	&	-0.06	&	1.19	&	1.28	&	1.19	&	0.96	&	-0.02	&	0.59	&	0.64	&	0.59	&	0.97	&	-0.04	&	0.38	&	0.41	&	0.38	&	0.96	\\
$\hat{\tau}_{\adj, \mim, \ss}$	&	-0.18	&	1.00&	0.99	&	1.02	&	0.94	&	-0.08	&	0.61	&	0.61	&	0.62	&	0.95	&	-0.14	&	0.38	&	0.37	&	0.41	&	0.92	\\
$\hat{\tau}_{\int, \mim, \ss}$	&	-0.04	&	1.04	&	1.02	&	1.04	&	0.94	&	-0.03	&	0.63	&	0.62	&	0.63	&	0.95	&	-0.04	&	0.39	&	0.39	&	0.39	&	0.95	\\
$\hat{\tau}_{\tom, \mim, \ss}$	&	-0.18	&	1.00&	1.00&	1.02	&	0.94	&	-0.08	&	0.61	&	0.61	&	0.62	&	0.95	&	-0.14	&	0.38	&	0.37	&	0.41	&	0.92	\\ \bottomrule
\end{tabular}

\end{threeparttable}
\end{table}

\begin{table}[H]
\tiny
\centering
\caption{Simulation results under stratified block randomization, $p = 7$, $\pi = 1/2$}
\label{tab3}
\begin{threeparttable}
\begin{tabular}{cccccccccccccccc}
\toprule
\multicolumn{3}{c}{~} & \multicolumn{4}{c}{Model 1} & \multicolumn{4}{c}{Model 2} & \multicolumn{4}{c}{Model 3} \\ \hline
estimator	&	bias	&	SD	&	SE	&	RMSE	&	CP	&	bias	&	SD	&	SE	&	RMSE	&	CP	&	bias	&	SD	&	SE	&	RMSE	&	CP	\\ \midrule
$\hat{\tau}_{\adj}$	&	0.01	&	1.27	&	1.27	&	1.27	&	0.95	&	0.00&	0.82	&	0.83	&	0.82	&	0.95	&	0.00&	0.45	&	0.44	&	0.45	&	0.95	\\
$\hat{\tau}_{\adj,\ccov}$	&	-0.02	&	1.27	&	1.26	&	1.27	&	0.95	&	0.00&	0.78	&	0.78	&	0.78	&	0.95	&	-0.01	&	0.40&	0.40&	0.40&	0.95	\\
$\hat{\tau}_{\int, \ccov}$	&	0.01	&	1.27	&	1.29	&	1.27	&	0.95	&	0.00&	0.78	&	0.80&	0.78	&	0.96	&	0.00&	0.40&	0.41	&	0.40&	0.95	\\
$\hat{\tau}_{\tom, \ccov}$	&	-0.02	&	1.27	&	1.29	&	1.27	&	0.95	&	0.00&	0.78	&	0.80&	0.78	&	0.96	&	-0.01	&	0.40&	0.41	&	0.40&	0.95	\\
$\hat{\tau}_{\adj, \ccov, \ss}$	&	-0.08	&	1.28	&	1.27	&	1.28	&	0.95	&	0.01	&	0.78	&	0.79	&	0.78	&	0.95	&	-0.03	&	0.39	&	0.38	&	0.39	&	0.94	\\
$\hat{\tau}_{\int, \ccov, \ss}$	&	0.00&	1.28	&	1.28	&	1.28	&	0.95	&	-0.01	&	0.79	&	0.80&	0.79	&	0.95	&	-0.01	&	0.39	&	0.39	&	0.39	&	0.95	\\
$\hat{\tau}_{\tom, \ccov, \ss}$	&	-0.08	&	1.28	&	1.28	&	1.28	&	0.95	&	0.01	&	0.78	&	0.79	&	0.78	&	0.95	&	-0.03	&	0.39	&	0.39	&	0.39	&	0.94	\\
$\hat{\tau}_{\adj, \imp}$	&	-0.01	&	1.18	&	1.17	&	1.18	&	0.95	&	0.01	&	0.60&	0.60&	0.60&	0.95	&	-0.01	&	0.39	&	0.39	&	0.39	&	0.94	\\
$\hat{\tau}_{\int, \imp}$	&	0.01	&	1.18	&	1.30&	1.18	&	0.97	&	0.01	&	0.60&	0.67	&	0.60&	0.97	&	0.00&	0.39	&	0.43	&	0.39	&	0.97	\\
$\hat{\tau}_{\tom, \imp}$	&	-0.01	&	1.18	&	1.29	&	1.18	&	0.97	&	0.01	&	0.60&	0.67	&	0.60&	0.97	&	-0.01	&	0.39	&	0.43	&	0.39	&	0.97	\\
$\hat{\tau}_{\adj, \imp, \ss}$	&	-0.05	&	1.02	&	1.01	&	1.02	&	0.95	&	0.04	&	0.63	&	0.63	&	0.63	&	0.95	&	-0.02	&	0.39	&	0.38	&	0.39	&	0.94	\\
$\hat{\tau}_{\int, \imp, \ss}$	&	0.04	&	1.03	&	1.05	&	1.03	&	0.95	&	0.03	&	0.64	&	0.65	&	0.64	&	0.95	&	0.02	&	0.40&	0.40&	0.40&	0.95	\\
$\hat{\tau}_{\tom, \imp, \ss}$	&	-0.05	&	1.02	&	1.03	&	1.02	&	0.95	&	0.04	&	0.63	&	0.64	&	0.63	&	0.95	&	-0.02	&	0.39	&	0.40&	0.39	&	0.95	\\
$\hat{\tau}_{\adj, \mim}$	&	-0.05	&	1.17	&	1.17	&	1.17	&	0.95	&	-0.03	&	0.58	&	0.58	&	0.58	&	0.95	&	-0.05	&	0.37	&	0.37	&	0.38	&	0.94	\\
$\hat{\tau}_{\int, \mim}$	&	0.01	&	1.17	&	1.43	&	1.17	&	0.98	&	0.00&	0.58	&	0.71	&	0.58	&	0.98	&	0.00&	0.37	&	0.46	&	0.37	&	0.98	\\
$\hat{\tau}_{\tom, \mim}$	&	-0.05	&	1.17	&	1.42	&	1.17	&	0.98	&	-0.03	&	0.58	&	0.71	&	0.58	&	0.98	&	-0.05	&	0.37	&	0.45	&	0.38	&	0.98	\\
$\hat{\tau}_{\adj, \mim, \ss}$	&	-0.24	&	1.03	&	1.02	&	1.05	&	0.94	&	-0.14	&	0.62	&	0.63	&	0.64	&	0.94	&	-0.20&	0.39	&	0.37	&	0.44	&	0.90\\
$\hat{\tau}_{\int, \mim, \ss}$	&	-0.02	&	1.09	&	1.16	&	1.09	&	0.96	&	-0.02	&	0.66	&	0.71	&	0.67	&	0.96	&	-0.03	&	0.40&	0.43	&	0.40&	0.96	\\
$\hat{\tau}_{\tom, \mim, \ss}$	&	-0.24	&	1.03	&	1.06	&	1.05	&	0.95	&	-0.14	&	0.62	&	0.65	&	0.64	&	0.95	&	-0.20&	0.39	&	0.39	&	0.44	&	0.91	\\ \bottomrule
\end{tabular}

\end{threeparttable}
\end{table}

\begin{table}[H]
\tiny
\centering
\caption{Simulation results under stratified block randomization, $p = 7$, $\pi = 2/3$}
\label{tab4}
\begin{threeparttable}
\begin{tabular}{cccccccccccccccc}
\toprule
\multicolumn{3}{c}{~} & \multicolumn{4}{c}{Model 1} & \multicolumn{4}{c}{Model 2} & \multicolumn{4}{c}{Model 3} \\ \hline
estimator	&	bias	&	SD	&	SE	&	RMSE	&	CP	&	bias	&	SD	&	SE	&	RMSE	&	CP	&	bias	&	SD	&	SE	&	RMSE	&	CP	\\ \midrule
$\hat{\tau}_{\int}$	&	0.00&	1.28	&	1.30&	1.28	&	0.96	&	0.02	&	0.87	&	0.88	&	0.87	&	0.95	&	0.01	&	0.43	&	0.42	&	0.43	&	0.95	\\
$\hat{\tau}_{\adj,\ccov}$	&	-0.02	&	1.28	&	1.37	&	1.28	&	0.96	&	0.02	&	0.83	&	0.82	&	0.83	&	0.95	&	0.00&	0.38	&	0.46	&	0.38	&	0.98	\\
$\hat{\tau}_{\int, \ccov}$	&	0.00&	1.28	&	1.33	&	1.28	&	0.96	&	0.02	&	0.83	&	0.85	&	0.83	&	0.95	&	0.01	&	0.38	&	0.39	&	0.38	&	0.95	\\
$\hat{\tau}_{\tom, \ccov}$	&	-0.02	&	1.28	&	1.33	&	1.28	&	0.96	&	0.02	&	0.83	&	0.85	&	0.83	&	0.95	&	0.00&	0.38	&	0.39	&	0.38	&	0.95	\\
$\hat{\tau}_{\adj, \ccov, \ss}$	&	-0.09	&	1.29	&	1.38	&	1.29	&	0.96	&	0.02	&	0.83	&	0.83	&	0.83	&	0.95	&	-0.02	&	0.37	&	0.44	&	0.37	&	0.98	\\
$\hat{\tau}_{\int, \ccov, \ss}$	&	0.00&	1.30&	1.32	&	1.30&	0.95	&	0.02	&	0.84	&	0.85	&	0.85	&	0.95	&	0.00&	0.37	&	0.37	&	0.37	&	0.95	\\
$\hat{\tau}_{\tom, \ccov, \ss}$	&	-0.08	&	1.29	&	1.32	&	1.29	&	0.95	&	0.03	&	0.84	&	0.84	&	0.84	&	0.95	&	-0.02	&	0.37	&	0.37	&	0.37	&	0.95	\\
$\hat{\tau}_{\adj, \imp}$	&	-0.02	&	1.20&	1.26	&	1.20&	0.96	&	0.02	&	0.64	&	0.63	&	0.64	&	0.95	&	0.00&	0.38	&	0.44	&	0.38	&	0.98	\\
$\hat{\tau}_{\int, \imp}$	&	0.00&	1.20&	1.38	&	1.20&	0.97	&	0.01	&	0.65	&	0.74	&	0.65	&	0.97	&	0.01	&	0.37	&	0.41	&	0.37	&	0.97	\\
$\hat{\tau}_{\tom, \imp}$	&	-0.02	&	1.20&	1.38	&	1.20&	0.97	&	0.02	&	0.65	&	0.74	&	0.65	&	0.97	&	0.00&	0.37	&	0.41	&	0.37	&	0.97	\\
$\hat{\tau}_{\adj, \imp, \ss}$	&	-0.05	&	1.06	&	1.08	&	1.06	&	0.95	&	0.05	&	0.67	&	0.66	&	0.68	&	0.95	&	-0.01	&	0.39	&	0.44	&	0.39	&	0.97	\\
$\hat{\tau}_{\int, \imp, \ss}$	&	0.00&	1.11	&	1.11	&	1.11	&	0.95	&	0.01	&	0.72	&	0.71	&	0.72	&	0.94	&	0.01	&	0.38	&	0.39	&	0.38	&	0.95	\\
$\hat{\tau}_{\tom, \imp, \ss}$	&	-0.09	&	1.07	&	1.08	&	1.08	&	0.95	&	0.04	&	0.69	&	0.69	&	0.69	&	0.95	&	-0.03	&	0.38	&	0.38	&	0.38	&	0.95	\\
$\hat{\tau}_{\adj, \mim}$	&	-0.07	&	1.20&	1.24	&	1.21	&	0.95	&	-0.03	&	0.63	&	0.59	&	0.63	&	0.93	&	-0.05	&	0.38	&	0.41	&	0.38	&	0.96	\\
$\hat{\tau}_{\int, \mim}$	&	-0.01	&	1.21	&	1.62	&	1.21	&	0.99	&	0.00&	0.64	&	0.87	&	0.64	&	0.99	&	0.00&	0.36	&	0.46	&	0.36	&	0.99	\\
$\hat{\tau}_{\tom, \mim}$	&	-0.07	&	1.21	&	1.62	&	1.21	&	0.99	&	-0.03	&	0.64	&	0.86	&	0.64	&	0.99	&	-0.05	&	0.37	&	0.46	&	0.37	&	0.98	\\
$\hat{\tau}_{\adj, \mim, \ss}$	&	-0.24	&	1.09	&	1.07	&	1.11	&	0.94	&	-0.13	&	0.68	&	0.63	&	0.69	&	0.92	&	-0.20&	0.41	&	0.41	&	0.45	&	0.92	\\
$\hat{\tau}_{\int, \mim, \ss}$	&	-0.06	&	1.26	&	1.27	&	1.26	&	0.95	&	-0.05	&	0.81	&	0.81	&	0.82	&	0.94	&	-0.06	&	0.42	&	0.44	&	0.42	&	0.95	\\
$\hat{\tau}_{\tom, \mim, \ss}$	&	-0.29	&	1.11	&	1.16	&	1.14	&	0.95	&	-0.17	&	0.71	&	0.74	&	0.73	&	0.95	&	-0.23	&	0.40&	0.40&	0.46	&	0.91	\\ \bottomrule
\end{tabular}

\end{threeparttable}
\end{table}

\begin{table}[H]
\footnotesize
\centering
\caption{Results based on NIDA-CNT-0027 data with $p = 5$ and $\pi = 2/3$}
\label{TAB0}
\begin{threeparttable}
\begin{tabular}{cccccccccccc}
\toprule
estimator	&	bias	&	SD	&	SE	&	RMSE	&	CP &
estimator	&	bias	&	SD	&	SE	&	RMSE	&	CP \\ \midrule
$\hat{\tau}_{\int}$	&	-0.007	&	0.294	&	0.297	&	0.294	&	0.953	& $\hat{\tau}_{\adj, \imp, \ss}$	&	-0.014	&	0.268	&	0.274	&	0.269	&	0.949 \\
$\hat{\tau}_{\adj,\ccov}$	&	-0.007	&	0.291	&	0.297	&	0.291	&	0.956	& $\hat{\tau}_{\int, \imp, \ss}$	&	-0.014	&	0.273	&	0.272	&	0.273	&	0.949 \\
$\hat{\tau}_{\int, \ccov}$	&	-0.008	&	0.292	&	0.297	&	0.292	&	0.955 &  $\hat{\tau}_{\tom, \imp, \ss}$	&	-0.013	&	0.271	&	0.272	&	0.271	&	0.950	\\
$\hat{\tau}_{\tom, \ccov}$	&	-0.007	&	0.292	&	0.297	&	0.292	&	0.955	& $\hat{\tau}_{\adj, \mim}$	&	-0.015	&	0.268	&	0.278	&	0.269	&	0.951 \\
$\hat{\tau}_{\adj, \ccov, \ss}$	&	-0.005	&	0.292	&	0.296	&	0.292	&	0.955	& $\hat{\tau}_{\int, \mim}$	&	-0.016	&	0.270&	0.277	&	0.271	&	0.952 \\
$\hat{\tau}_{\int, \ccov, \ss}$	&	-0.004	&	0.293	&	0.295	&	0.293	&	0.954	& $\hat{\tau}_{\tom, \mim}$	&	-0.015	&	0.270&	0.277	&	0.270&	0.948 \\
$\hat{\tau}_{\tom, \ccov, \ss}$	&	-0.003	&	0.293	&	0.295	&	0.293	&	0.954 & $\hat{\tau}_{\adj, \mim, \ss}$	&	-0.014	&	0.270&	0.276	&	0.270&	0.952	\\
$\hat{\tau}_{\adj, \imp}$	&	-0.014	&	0.268	&	0.275	&	0.268	&	0.952	& $\hat{\tau}_{\int, \mim, \ss}$	&	-0.012	&	0.274	&	0.274	&	0.275	&	0.947 \\
$\hat{\tau}_{\int, \imp}$	&	-0.015	&	0.269	&	0.274	&	0.270&	0.951	& $\hat{\tau}_{\tom, \mim, \ss}$	&	-0.012	&	0.272	&	0.273	&	0.272	&	0.950 \\
$\hat{\tau}_{\tom, \imp}$	&	-0.015	&	0.269	&	0.274	&	0.269	&	0.951 \\ \bottomrule
\end{tabular}

\end{threeparttable}
\end{table}

\newpage

\begin{center}
{\Large\bf Supplementary material for ``Regression adjustment in covariate-adaptive randomized experiments with missing covariates"}
\end{center}

\setcounter{section}{0}
\renewcommand\thesection{\Alph{section}}

\Cref{secIII} provides a complete case analysis. \Cref{secI} provides proofs of the theoretical results. \Cref{secII} provides additional simulation results.

\section{Complete case analysis}
\label{secIII}

In this section, we extend the result that complete case analysis (CC) may introduce bias from the finite-population setting to the super-population setting. For the CC method, due to the deletion of observations with missing data, a key technical question arises: For CC, s there a result parallel to that in Theorems \ref{th1}--\ref{th2}? To address the issue, we need the \ref{as6} below, which is stronger than Assumption \ref{as4}:

\begin{assumption}
\label{as6}

$\{\{n^{-\frac{1}{2}} D_{n[k]}\}_{k=1,\ldots,K}\vert B^{(n)}\} \xrightarrow{d} N(0,\Sigma_D) \ a.s.$, where $\Sigma_D = diag\{p_{[k]}\qk \colon k=1,\ldots,K\}$, $0 \leq \qk \leq \pi(1-\pi), \ k=1,\ldots,K$.

\end{assumption}

Assumption~\ref{as6}, proposed by \cite{Bugni:2018}, was introduced to analyze statistical inference under covariate-adaptive randomization. This assumption delineates the asymptotic behavior of jointly independent imbalances within strata, making it particularly pertinent to stratified randomization.

We demonstrate that if Assumption \ref{as6} holds for all $n$ observations with parameter $\qk$, it continues to hold for the remaining observations with a different $\qk^{\prime}$. Let $C_i = 1$ represent that the $i$th observation contains missing values and $C_i=0$ otherwise. 


\begin{proposition}
\label{ps2}
If Assumption \ref{as6} holds for all $n$ observations with parameter $\qk$, then it still holds for the remaining observations after removing those with missing values, with $\qk$ replaced by $\qk^{\prime}$:
\[
	\qk^{\prime} = \frac{\pi(1-\pi)P(C_i = 0 \vert B_i = k)\{1-P(C_i = 0 \vert B_i = k)\} + \qk P(C_i = 0 \vert B_i = k)^2}{\sum_{k^{\prime}=1}^K p_{[k^{\prime}]}P(C_i = 0 \vert B_i = k^{\prime})}.
\]
\end{proposition}

Proposition~\ref{ps2} implies that if we replace Assumption \ref{as4} with Assumption \ref{as6}, Theorems \ref{th1}--\ref{th4} still hold for CC, with $\tau$ replaced by $\tau_{cc} = E\{Y_i(1) - Y_i(0) \vert C_i = 0\}$ and $\qk$ replaced by $\qk^{\prime}$. Therefore, consistency may be violated for CC, as $\tau_{cc}$ may not equal $\tau$. A significant bias may arise for CC when $M_i$ is strongly related to $Y_i(a)$. For instance, in clinical trials, it becomes more difficult to collect baseline information for patients with more severe conditions. Hence, we do not recommend using CC in practice.

\section{Proofs}
\label{secI}

In the following proofs, Assumption~\ref{as4} is only required by Fisher's regression and can be weakened to Assumption~\ref{as3} below for Lin's regression and ToM regression. 

\begin{assumption}
\label{as3}
$\pi_{n[k]} \xrightarrow {P} \pi \in (0,1).$
\end{assumption}

\subsection{Useful lemmas}

\begin{lemma}
\label{lm1}
Under Assumptions~\ref{as1}--\ref{as5} and $\hat{c} \xrightarrow {P} c$, for $\dag \in \{\ccov, \imp, \mim\}$, we have
\begin{align*}
    &  \frac{1}{\nk} \sum_{i \in [k]} U_{i, \dag}(\hat{c})U_{i, \dag}(\hat{c})^\top = \frac{1}{\nk} \sum_{i \in [k]} U_{i, \dag}({c})U_{i, \dag}({c})^\top + o_p(1), \\
    &\frac{1}{\nk} \sum_{i \in [k]} U_{i, \dag}(\hat{c}) = \frac{1}{\nk} \sum_{i \in [k]} U_{i, \dag}({c}) + o_p(1), \quad \frac{1}{\nk} \sum_{i \in [k]} U_{i, \dag}(\hat{c})Y_i = \frac{1}{\nk} \sum_{i \in [k]} U_{i, \dag}({c})Y_i+ o_p(1), \\
    & \bar{U}_{[k], \dag}(\hat{c}) \bar{Y}_{[k]} = \bar{U}_{[k], \dag}({c}) \bar{Y}_{[k]} + o_p(1).
\end{align*}
\end{lemma}

\begin{proof}
All the results are obvious for CCOV because it does not involve imputation. For IMP, when $c$ is replaced by its consistent estimator $\hat{c}$, $U_{i, \dag}(\hat{c}) - U_{i, \dag}({c}) = M_i \odot (\hat{c} - c)$, where $\odot$ denotes the Hadamard product. Therefore, 
\begin{align}
&\frac{1}{n} \sum_{i \in [k]} U_{i, \dag}(\hat{c})U_{i, \dag}(\hat{c})^\top  \notag \\
=& \frac{1}{n} \sum_{i \in [k]} U_{i, \dag}({c})\{U_{i, \dag}({c})\}^\top + \frac{1}{n} \sum_{i \in [k]} U_{i, \dag}({c})\{U_{i, \dag}(\hat{c}) - U_{i, \dag}({c})\}^\top \notag \\
& + \frac{1}{n} \sum_{i \in [k]} \{U_{i, \dag}(\hat{c}) - U_{i, \dag}({c})\}\{U_{i, \dag}({c})\}^\top + \frac{1}{n} \sum_{i \in [k]} \{U_{i, \dag}(\hat{c}) - U_{i, \dag}({c})\}\{U_{i, \dag}(\hat{c}) - U_{i, \dag}({c})\}^\top  \notag \\
=& \frac{1}{n} \sum_{i \in [k]} U_{i, \dag}({c})\{U_{i, \dag}({c})\}^\top + \frac{1}{n} \sum_{i \in [k]} U_{i, \dag}(c) \{{M}_i \odot (\hat{c} - c)\}^\top \notag \\
& + \frac{1}{n} \sum_{i \in [k]} \{{M}_i \odot (\hat{c} - c)\}\{U_{i, \dag}(c)\}^\top + \frac{1}{n} \sum_{i \in [k]} \{{M}_i \odot (\hat{c} - c)\}\{{M}_i \odot (\hat{c} - c)\}^\top \notag \\
=& \frac{1}{n} \sum_{i \in [k]} U_{i, \dag}({c})\{U_{i, \dag}({c})\}^\top + o_p(1), \notag
\end{align}
where the last equation follows from Lemma B.3 of \cite{Bugni:2018}, $({1}/{n}) \sum_{i \in [k]} U_{i, \dag}(c)$ has a finite probability limit, $\vert M_{ij} \vert \leq 1$, and $\hat{c} \xrightarrow{P} c$. Similarly, we have
\begin{align*}
    \frac{1}{n} \sum_{i \in [k]} U_{i, \dag}(\hat{c}) = \frac{1}{n} \sum_{i \in [k]} U_{i, \dag}({c}) + \frac{1}{n} \sum_{i \in [k]}\{{M}_{i} \odot (\hat{c} - c)\} = \frac{1}{n} \sum_{i \in [k]} U_{i, \dag}({c}) + o_p(1).
\end{align*}

	For the term $({1}/{n}) \sum_{i \in [k]} U_{i, \dag}(\hat{c})Y_i$, we have
\begin{align}
\frac{1}{n} \sum_{i \in [k]} U_{i, \dag}(\hat{c})Y_i =& \frac{1}{n} \sum_{i \in [k]} A_iU_{i, \dag}(\hat{c})Y_i(1) + \frac{1}{n} \sum_{i \in [k]} (1-A_i)U_{i, \dag}(\hat{c})Y_i(0) \notag \\
=& \frac{1}{n} \sum_{i \in [k]} A_iU_{i, \dag}({c})Y_i(1) + \frac{1}{n} \sum_{i \in [k]} (1-A_i)U_{i, \dag}({c})Y_i(0) \notag \\
&+ \frac{1}{n} \sum_{i \in [k]} A_iM_i \odot (\hat{c} - c)Y_i(1) + \frac{1}{n} \sum_{i \in [k]} (1-A_i)M_i \odot (\hat{c} - c)Y_i(0) \notag \\
=& \frac{1}{n} \sum_{i \in [k]} A_iU_{i, \dag}({c})Y_i(1) + \frac{1}{n} \sum_{i \in [k]} (1-A_i)U_{i, \dag}({c})Y_i(0)+ o_p(1) \notag \\
=& \frac{1}{n} \sum_{i \in [k]} U_{i, \dag}({c})Y_i+ o_p(1), \notag
\end{align}
where the third equation follows from Lemma B.3 of \cite{Bugni:2018} and $\hat{c} \xrightarrow{P} c$. By the same lemma, we also have
\begin{align*}
    \frac{1}{n} \sum_{i \in [k]} Y_i =& \frac{1}{n} \sum_{i \in [k]} A_iY_i(1) + \frac{1}{n} \sum_{i \in [k]} (1-A_i)Y_i(0)  \notag \\
    \xrightarrow{P} & \pi p_{[k]} E\{Y_i(1) \vert B_i = k\}+(1-\pi)p_{[k]}E\{Y_i(0) \vert B_i = k\}.
\end{align*}
Combined with $({1}/{n}) \sum_{i \in [k]} U_{i, \dag}(\hat{c}) = ({1}/{n}) \sum_{i \in [k]} U_{i, \dag}({c}) + o_p(1)$, we have 
$$\bar{U}_{[k], \dag}(\hat{c}) \bar{Y}_{[k]} = \bar{U}_{[k], \dag}({c}) \bar{Y}_{[k]} + o_p(1).$$
\end{proof}

\begin{lemma}
\label{lm2}
Under Assumptions~\ref{as1}--\ref{as5} and $\hat{c} \xrightarrow {P} c$, for $\dag \in \{\ccov, \imp, \mim\}$, we have
\begin{align*}
    &\bar{U}_{[k]1, \dag}(\hat{c}) - \bar{U}_{[k]0, \dag}(\hat{c}) = \bar{U}_{[k]1, \dag}({c}) - \bar{U}_{[k]0, \dag}({c})+ o_p(n^{-1/2}), \\
    & \bar{U}_{[k]1, \dag}(\hat{c}) - \bar{U}_{[k], \dag}(\hat{c}) = \bar{U}_{[k]1, \dag}({c}) - \bar{U}_{[k], \dag}({c})+ o_p(n^{-1/2}), \\
    &\bar{U}_{[k], \dag}(\hat{c}) - \bar{U}_{[k]0, \dag}(\hat{c}) = \bar{U}_{[k], \dag}({c}) - \bar{U}_{[k]0, \dag}({c})+ o_p(n^{-1/2}).
\end{align*}
\end{lemma}

\begin{proof}
All the results are obvious for CCOV because it does not involve imputation. 
For the term $\bar{U}_{[k]1, \dag}(\hat{c}) - \bar{U}_{[k]0, \dag}(\hat{c})$ with $\dag = \imp$, we have
\begin{align}
&\bar{U}_{[k]1, \dag}(\hat{c}) - \bar{U}_{[k]0, \dag}(\hat{c}) \notag \\
=& \bar{U}_{[k]1, \dag}({c}) - \bar{U}_{[k]0, \dag}({c}) + \{\bar{U}_{[k]1, \dag}(\hat{c})-\bar{U}_{[k]1, \dag}({c})\}-\{\bar{U}_{[k]0, \dag}(\hat{c})-\bar{U}_{[k]0, \dag}({c})\} \notag \\
=& \bar{U}_{[k]1, \dag}({c}) - \bar{U}_{[k]0, \dag}({c}) + \bigl\{\frac{1}{n_{[k]0}}\sum_{i \in [k]} A_iM_i - \frac{1}{n_{[k]0}} \sum_{i \in [k]} (1-A_i)M_i\bigr\} \odot (\hat{c} - c) \notag \\
=& \bar{U}_{[k]1, \dag}({c}) - \bar{U}_{[k], \dag}({c}) + o_p(n^{-1/2}), \notag
\end{align}
\noindent where the last equation follows from $\hat{c} \xrightarrow{P} c$ and $\bigl\{({1}/{n_{[k]1}}) \sum_{i \in [k]} A_iM_i - ({1}/{n_{[k]0})} \sum_{i \in [k]} (1-A_i)M_i\bigr\} = O_p(n^{-1/2})$ by Theorem 4.1 of \cite{Bugni:2018} and Assumption \ref{as5}. 


    Note that $\bar{U}_{[k]1, \dag}(\hat{c}) - \bar{U}_{[k], \dag}(\hat{c}) = ({n_{[k]0}}/{n_{[k]}})\{\bar{U}_{[k]1, \dag}(\hat{c}) - \bar{U}_{[k]0, \dag}(\hat{c})\}$, we have $\bar{U}_{[k]1, \dag}(\hat{c}) - \bar{U}_{[k], \dag}(\hat{c}) = \bar{U}_{[k]1, \dag}({c}) - \bar{U}_{[k]0, \dag}({c}) + o_p(n^{-1/2})$. Similarly, $\bar{U}_{[k], \dag}(\hat{c}) - \bar{U}_{[k]0, \dag}(\hat{c}) = \bar{U}_{[k], \dag}({c}) - \bar{U}_{[k]0, \dag}({c})+ o_p(n^{-1/2}).$ 

    As $M_i$ does not involve imputation, the result for $\dag = \mim$ follows from the above proof for $\dag = \imp$.
\end{proof}

\begin{lemma}
\label{lm3}

Under Assumptions~\ref{as1}--\ref{as5} and $\hat{c} \xrightarrow {P} c$, for $\dag \in \{\ccov, \imp, \mim\}$, we have
\[
	\frac{1}{\nkt} \sum_{i \in [k]} A_iU_{i, \dag}(\hat{c})Y_i(1)  = \frac{1}{\nkt} \sum_{i \in [k]} A_iU_{i, \dag}({c})Y_i(1) + o_p(1),
\]
\[
	\frac{1}{\nkc} \sum_{i \in [k]} (1-A_i)U_{i, \dag}(\hat{c})Y_i(0) = \frac{1}{\nkc} \sum_{i \in [k]} (1-A_i)U_{i, \dag}({c})Y_i(0) + o_p(1),
\]
\[
	\bar{U}_{[k]a, \dag}(\hat{c}) \bar{Y}_{[k]a} = \bar{U}_{[k]a, \dag}({c}) \bar{Y}_{[k]a} + o_p(1), \quad a=0,1.
\]
\end{lemma}

\begin{proof}
    The first two statements have been proved in the proof of Lemma \ref{lm1}. For the third statement, we will only prove the result for $a=1$, as the proof for $a=0$ is similar.

    By Lemma B.3 of \cite{Bugni:2018}, we have $({1}/{n}) \sum_{i \in [k]} A_iY_i(1) \xrightarrow{P} \pi p_{[k]} E\{Y_i(1) \vert B_i = k\}$ and $\nkt / n = \pk \pi_{n[k]}$. Therefore, $\bar{Y}_{[k]1} = O_p(1)$. Moreover,
    \begin{align*}
        \bar{U}_{[k]1, \dag}(\hat{c}) - \bar{U}_{[k]1, \dag}({c}) = (1/\nkt)\sum_{i \in [k]} A_iM_i \odot (\hat c - c) = o_p(1),
    \end{align*}
    as $(1/\nkt)\sum_{i \in [k]} A_iM_i$ is bounded and $\hat c \xrightarrow{P} c$. Therefore, 
    $$\bar{U}_{[k]1, \dag}(\hat{c}) \bar{Y}_{[k]1} - \bar{U}_{[k]1, \dag}({c}) \bar{Y}_{[k]1} = \{\bar{U}_{[k]1, \dag}(\hat{c}) - \bar{U}_{[k]1, \dag}({c})\}\bar{Y}_{[k]1} = o_p(1) \cdot O_p(1) = o_p(1).$$
\end{proof}

\subsection{Proof of \Cref{th1}}

\begin{proof}

By the proof of Theorem 2 of \cite{Ma:2020}, with $X_i$ replaced by $U_{i, \dag}(\hat c)$ for ${\dag} \in \{\ccov,\imp,\mim\}$, $\hat{\tau}_{\adj, \dag}$ has the following form:
\[
	{\hat{\tau}_{\adj, \dag} = \sum_{k=1}^K \omega_{[k]}\bigl[\bar{Y}_{[k]1} - \bar{Y}_{[k]0} - \{\bar{U}_{[k]1, \dag}(\hat c) - \bar{U}_{[k]0, \dag}(\hat c)\}^\top\hat{\beta}_{\adj, U_{\dag}(\hat c)}\bigr],}
\]
\noindent where 
\begin{align*}
	\hat{\beta}_{\adj, {U_{\dag}(\hat c)}} =& \biggl\{\hat{S}_{{U_{\dag}(\hat c)U_{\dag}(\hat c)}}^{\adj} - \hat{\tau}_{\adj, U(\hat c), \dag}\hat{\tau}_{\adj, U(\hat c), \dag}^\top \sum_{k=1}^K \pi_{n[k]}(1-\pi_{n[k]})p_{n[k]}\biggr\}^{-1} \\
    &\biggl\{\hat{S}_{U_{\dag}(\hat c)Y}^{\adj}- \hat{\tau}_{\adj} \hat{\tau}_{\adj, U(\hat c), \dag} \sum_{k=1}^K \pi_{n[k]}(1-\pi_{n[k]})p_{n[k]}  \biggr\},
\end{align*}
\[
    \omega_{[k]} = \frac{\pi_{n[k]}(1-\pi_{n[k]})p_{n[k]}}{\sum_{k^{\prime}=1}^K \pi_{n[k^{\prime}]}(1-\pi_{n[k^{\prime}]})p_{n[k^{\prime}]}}, \quad \hat{\tau}_{\adj, U(\hat c), \dag} = \sum_{k=1}^K \omega_{[k]}\{\bar{U}_{[k]1, \dag}(\hat c) - \bar{U}_{[k]0, \dag}(\hat c)\},
\]
\[
	\hat{S}_{U_{\dag}(\hat c)U_{\dag}(\hat c)}^{\adj} = \frac{1}{n} \sum_{k=1}^K \sum_{i \in [k]} \{U_{i, \dag}(\hat c) - \bar{U}_{[k], \dag}(\hat c)\}\{U_{i, \dag}(\hat c) - \bar{U}_{[k], \dag}(\hat c)\}^\top,
\]
\[
    \hat{S}_{U_{\dag}(\hat c)Y}^{\adj} = \frac{1}{n} \sum_{k=1}^K \sum_{i \in [k]} \{U_{i, \dag}(\hat c) - \bar{U}_{[k], \dag}(\hat c)\}{(Y_i - \bar{Y}_{[k]})}.
\]

First, we prove Lemma~\ref{lm4} below.
\begin{lemma}
\label{lm4}
Under Assumptions~\ref{as1}--\ref{as5} and $\hat{c} \xrightarrow {P} c$, for $\dag \in \{\ccov, \imp, \mim\}$, we have
\[
	\hat{S}_{U_{\dag}(\hat c)U_{\dag}(\hat c)}^{\adj} - \Sigma_{\tilde{U}_{\dag}(c)\tilde{U}_{\dag}(c)} = o_p(1), 
\]
\[ 
    \hat{S}_{U_{\dag}(\hat c)Y}^{\adj} - \{\pi\Sigma_{\tilde{U}_{\dag}(c)\tilde{Y}(1)} + (1-\pi)\Sigma_{\tilde{U}(c)_{\dag}\tilde{Y}(0)}\} = o_p(1),
\]
\[ 
\hat{\beta}_{\adj, U_{\dag}(\hat c)} - {\beta}_{U_{\dag}(c)} = o_p(1) \textnormal{ for } \pi = 1/2. 
\]
\end{lemma}


\begin{proof}[Proof of Lemma~\ref{lm4}]

These results are similar to those of Lemma 6 of \cite{Ma:2020}, with the key difference being the imputation of $\hat c$. Note that the proof of Lemma 6 in \cite{Ma:2020} is invariant if we replace $X_i$ with $U_{i, \dag}(c)$ if $U_{i, \dag}(c)$ are i.i.d. with finite second-order moments. For CCOV, $U_{i, \ccov}(\hat c)=X_{i,\ccov}$, the requirements are clearly satisfied. For IMP and MIM, we consider two cases:
    
    {\bf Case 1.} We interpolate a constant $c$, i.e., $\hat c = c$. For IMP, under Assumption \ref{as5}, we have $(X_i, M_i)$ are i.i.d.. Note that $U_{i,\imp}(c) = X_i \odot (1_p-M_i) + c \odot M_i$, where $\odot$ denotes the Hadamard product, $1_p$ is a $p$-dimensional vector with all elemets equal to $1$. Therefore, $U_{i,\imp}(c)$ are i.i.d.. Applying the Minkovski's inequality, we have
\[
	\| X_i \odot (1_p-M_i) + c \odot M_i \|_2 \leq \| X_i \odot (1_p-M_i) \|_2 + \| c \odot M_i \|_2 \leq \| X_i \|_2 + \|c \|_2 < \infty,
\]
where $\|\cdot\|_2$ denotes the $L_2$ norm of a random variable.

	For MIM, the i.i.d. of $U_{i, \min}(c)$ can be derived similarly to IMP, and $U_{i,\mim}(c)$ has a finite second-order moment because all elements of $M_i$ take values in $\{0, 1\}$.

 {\bf Case 2.} We use $\hat{c}$ to impute the unobserved value. By Lemma~\ref{lm1} and $\nk / n \xrightarrow{P} \pk >0$, we have
\begin{align*}
    & \frac{1}{n} \sum_{i \in [k]} \{U_{i, \dag}(\hat c) - \bar{U}_{[k], \dag}(\hat c)\}\{U_{i, \dag}(\hat c) - \bar{U}_{[k], \dag}(\hat c)\}^\top \\
    = & \frac{1}{n}  \sum_{i \in [k]} U_{i, \dag}({\hat c})\{U_{i, \dag}(\hat c)\}^\top - \frac{\nk}{n} \bar{U}_{[k], \dag}(\hat c) \{\bar{U}_{[k], \dag}(\hat c)\}^\top \\
    = & \frac{1}{n}  \sum_{i \in [k]} U_{i, \dag}({c})\{U_{i, \dag}(c)\}^\top - \frac{\nk}{n} \bar{U}_{[k], \dag}(c) \{\bar{U}_{[k], \dag}(c)\}^\top + o_p(1). 
\end{align*}

As the number of strata $K$ is fixed, we have
\begin{align*}
 \hat{S}_{U_{\dag}(\hat{c})U_{\dag}(\hat{c})}^{\adj} = & \frac{1}{n} \sum_{k=1}^K \sum_{i \in [k]} \{U_{i, \dag}(\hat c) - \bar{U}_{[k], \dag}(\hat c)\}\{U_{i, \dag}(\hat c) - \bar{U}_{[k], \dag}(\hat c)\}^\top  \\
    = & \frac{1}{n} \sum_{k=1}^K \sum_{i \in [k]} \{U_{i, \dag}(c) - \bar{U}_{[k], \dag}(c)\}\{U_{i, \dag}(c) - \bar{U}_{[k], \dag}(c)\}^\top + + o_p(1) \\
    = & \hat{S}_{U_{\dag}(c)U_{\dag}(c)}^{\adj}  + o_p(1). 
\end{align*}

In Case 1, we have $\hat{S}_{U_{\dag}(c)U_{\dag}(c)}^{\adj} - \Sigma_{\tilde{U}_{\dag}(c)\tilde{U}_{\dag}(c)} = o_p(1)$. Therefore, $\hat{S}_{U_{\dag}(\hat{c})U_{\dag}(\hat{c})}^{\adj} - \Sigma_{\tilde{U}_{\dag}(c)\tilde{U}_{\dag}(c)} = o_p(1)$. Similarly, we have $\hat{S}_{U_{\dag}(\hat c)Y}^{\adj} - \{\pi\Sigma_{\tilde{U}_{\dag}(c)\tilde{Y}(1)} + (1-\pi)\Sigma_{\tilde{U}_{\dag}(c)\tilde{Y}(0)}\} = o_p(1)$.

By Proposition 2 of \cite{Ma:2020}, we have $\hat{\tau}_{\adj} \xrightarrow{P} \tau$. When the imputation value is a constant $c$, applying Proposition 2 of \cite{Ma:2020} to $U_{i, \dag}(c)$, we have $\hat{\tau}_{\adj, U(c), \dag} = O_p(n^{-1/2})$. When the imputation value is $\hat{c}$, then by Lemma \ref{lm2}, we have $\bar{U}_{[k]1, \dag}(\hat{c}) - \bar{U}_{[k]0, \dag}(\hat{c}) = \bar{U}_{[k]1, \dag}({c}) - \bar{U}_{[k]0, \dag}({c})+ o_p(n^{-1/2})$ for each stratum $k$, and $\omega_{[k]}$ have finite probability limits. Therefore, $\hat{\tau}_{\adj, U(\hat c), \dag} - \hat{\tau}_{\adj, U(c), \dag} = o_p(n^{-1/2})$, and thus $\hat{\tau}_{\adj, U(\hat c), \dag} = O_p(n^{-1/2})$. Combined with $\pi_{n[k]}(1-\pi_{n[k]})p_{n[k]}$ and $\omega_{[k]}$ have finite probability limits, we have {$\hat{\beta}_{\adj, U_{\dag}(\hat c)} - {\beta}_{U_{\dag}(c)} = o_p(1)$.}

\end{proof}

Now, we can continue the proof of \Cref{th1}. By \Cref{lm2}, \Cref{lm4}, and the fact that $\omega_{[k]}$ have finite probability limits, we have
\begin{align}
\hat{\tau}_{\adj, \dag} &= \sum_{k=1}^K \omega_{[k]}\bigl[\bar{Y}_{[k]1} - \bar{Y}_{[k]0} - \{\bar{U}_{[k]1, \dag}(c) - \bar{U}_{[k]0, \dag}(c)\}^\top {\beta}_{U_{\dag}(c)}\bigr] + o_p(n^{-1/2}) \notag \\
& = \check{\tau}_{\adj, \dag} + o_p(n^{-1/2}), \notag
\end{align}
where $\check{\tau}_{\adj, \dag} = \sum_{k=1}^K \omega_{[k]}[\bar{Y}_{[k]1} - \bar{Y}_{[k]0} - \{\bar{U}_{[k]1, \dag}(c) - \bar{U}_{[k]0, \dag}(c)\}^\top {\beta}_{U_{\dag}(c)}]$. Applying Theorem 2 of \cite{Ma:2020} with $\pi= 1/2$, we have
\[
	\sqrt{n}(\check{\tau}_{\adj, \dag}-\tau_{\dag})\xrightarrow{d} N(0,\varsigma_{r_{\dag}}^2 (\pi) + \varsigma_{H{r_{\dag}}}^2),
\]
\noindent where $\tau_{\dag} = E\{r_{i, \dag}(1)-r_{i, \dag}(0)\} = E\{Y_i(1)-Y_i(0)\} = \tau$. By Slutsky's theorem, when $\pi=1/2$, we have 
\[
	\sqrt{n}(\hat{\tau}_{\adj, \dag}-\tau)\xrightarrow{d} N(0,\varsigma_{r_{\dag}}^2 (\pi) + \varsigma_{H{r_{\dag}}}^2).
\]

    Next, we prove the consistency of the OLS variance estimator. Denote $Z_{\adj}$ as the $n \times (K+p_{\dag}+1)$ design matrix whose $i$th row is
\[
    (A_i, 1-A_i, I_{i \in [1]}-p_{n[1]}, \cdots, I_{i \in [K-1]}-p_{n[K-1]}, \{U_{i,\dag}(\hat c) - \bar{U}_{\dag}(\hat c)\}^\top).
\]
    By the property of OLS, $\hat{\sigma}_{\adj, \dag}^{2}$ is the $(1, 1)+(2, 2)-2(1, 2)$ element of $\{(1/n)Z_{\adj}^{\top}Z_{\adj}\}^{-1}$ (the $(1, 1)$th element plus the $(2, 2)$th elememt minus 2 times the $(1, 2)$th element of the matrix) times
\[
    \frac{1}{n-K-p_{\dag}-1} \sum_{k=1}^K \sum_{i \in [k]} \{Y_i-\bar{Y}_{[k]}-\hat{\tau}_{\adj, \dag}(A_i - \pi_{n[k]})-\{U_{i,\dag}(\hat c)-\bar{U}_{[k], \dag}(\hat c)\}^{\top}\hat{\beta}_{\adj, U_{\dag}(\hat c)}\}^2,
\]
\noindent where $p_{\dag}$ is the dimension of $U_{i,\dag}(\hat c)$. By Theorem 2 of \cite{Ma:2020}, we have  $$\{(1/n)Z_{\adj}^\top Z_{\adj}\}^{-1}_{(1, 1)+(2, 2)-2(1, 2)} \xrightarrow[]{P} \{\pi(1-\pi)\}^{-1}.$$ 

Next, we prove that
\[
    \frac{1}{n} \sum_{k=1}^K \sum_{i \in [k]} \bigl[Y_i-\bar{Y}_{[k]}-\hat{\tau}_{\adj, \dag}(A_i - \pi_{n[k]})-\{U_{i,\dag}(\hat c)-\bar{U}_{[k], \dag}(\hat c)\}\top\hat{\beta}_{\adj, U_{\dag}(\hat c)}\bigr]^2 \xrightarrow{P} \pi(1-\pi)\{\varsigma_{r_{\dag}}^2 (1-\pi) + \varsigma_{H{r_{\dag}}}^2\}.
\]
    When $\hat c = c$, applying Theorem 2 of \cite{Ma:2020} with $X_i$ replaced by $U_{i,\dag}(c)$ immediately gives the result. Otherwise, we need to prove that 
\begin{align}
    &\frac{1}{n} \sum_{k=1}^K \sum_{i \in [k]} \bigl[Y_i-\bar{Y}_{[k]}-\hat{\tau}_{\adj, \dag}(A_i - \pi_{n[k]})-\{U_{i,\dag}(\hat c)-\bar{U}_{[k], \dag}(\hat c)\}^{\top}\hat{\beta}_{\adj, U_{\dag}(\hat c)}\bigr]^2 \notag \\ 
    &= \frac{1}{n} \sum_{k=1}^K \sum_{i \in [k]} \bigl[Y_i-\bar{Y}_{[k]}-\hat{\tau}_{\adj, \dag}(A_i - \pi_{n[k]})-\{U_{i,\dag}(c)-\bar{U}_{[k], \dag}(c)\}^{\top}\hat{\beta}_{\adj, U_{\dag}(c)}\bigr]^2 + o_p(1). \notag
\end{align}

    Note that 
\begin{align}
    &\frac{1}{n} \sum_{k=1}^K \sum_{i \in [k]} \bigl[Y_i-\bar{Y}_{[k]}-\hat{\tau}_{\adj, \dag}(A_i - \pi_{n[k]})-\{U_{i,\dag}(\hat c)-\bar{U}_{[k], \dag}(\hat c)\}^{\top}\hat{\beta}_{\adj, U_{\dag}(\hat c)}\bigr]^2 \notag \\
    =& \frac{1}{n} \sum_{k=1}^K \sum_{i \in [k]} \bigl[Y_i-\bar{Y}_{[k]}-\hat{\tau}_{\adj, \dag}(A_i - \pi_{n[k]})-\{U_{i,\dag}(c)-\bar{U}_{[k], \dag}(c)\}^{\top}\hat{\beta}_{\adj, U_{\dag}(c)}\bigr]^2 \notag \\
    & + \frac{1}{n} \sum_{k=1}^K \sum_{i \in [k]} \bigl[\{U_{i,\dag}(\hat c)-\bar{U}_{[k], \dag}(\hat c)\}^{\top}\hat{\beta}_{\adj, U_{\dag}(\hat c)}-\{U_{i,\dag}(c)-\bar{U}_{[k], \dag}(c)\}^{\top}\hat{\beta}_{\adj, U_{\dag}(c)}\bigr]^2 \notag \\
    & + \frac{2}{n} \sum_{k=1}^K \sum_{i \in [k]} \bigl[Y_i-\bar{Y}_{[k]}-\hat{\tau}_{\adj, \dag}(A_i - \pi_{n[k]})-\{U_{i,\dag}(c)-\bar{U}_{[k], \dag}(c)\}^{\top}\hat{\beta}_{\adj, U_{\dag}(c)}\bigr] \notag \\
    &~~~~~~~~~~~~\quad \bigl[\{U_{i,\dag}(\hat c)-\bar{U}_{[k], \dag}(\hat c)\}^{\top}\hat{\beta}_{\adj, U_{\dag}(\hat c)}-\{U_{i,\dag}(c)-\bar{U}_{[k], \dag}(c)\}^{\top}\hat{\beta}_{\adj, U_{\dag}(c)}\bigr].\notag 
\end{align}
The first term is $O_p(1)$ by Theorem 2 of \cite{Ma:2020}. If the second term is $o_p(1)$, then the third term will be $o_p(1)$ due to Cauchy-Schwarz inequality, and the desired result holds. Therefore, we only need to prove that the second term is $o_p(1)$. In fact, 
\begin{align}
    &\frac{1}{n} \sum_{k=1}^K \sum_{i \in [k]} \bigl[\{U_{i,\dag}(\hat c)-\bar{U}_{[k], \dag}(\hat c)\}^{\top}\hat{\beta}_{\adj, U_{\dag}(\hat c)}-\{U_{i,\dag}(c)-\bar{U}_{[k], \dag}(c)\}^{\top}\hat{\beta}_{\adj, U_{\dag}(c)}\bigr]^2 \notag \\
    =&\frac{1}{n} \sum_{k=1}^K \sum_{i \in [k]} \bigl[\{U_{i,\dag}(c)-\bar{U}_{[k], \dag}(c)\}^{\top}\hat{\beta}_{\adj, U_{\dag}(\hat c)}-\{U_{i,\dag}(c)-\bar{U}_{[k], \dag}(c)\}^{\top}\hat{\beta}_{\adj, U_{\dag}(c)}\bigr]^2 \notag \\
    & + \frac{1}{n} \sum_{k=1}^K \sum_{i \in [k]} \bigl[\{U_{i,\dag}(\hat c)-\bar{U}_{[k], \dag}(\hat c)\}^{\top}\hat{\beta}_{\adj, U_{\dag}(\hat c)}-\{U_{i,\dag}(c)-\bar{U}_{[k], \dag}(c)\}^{\top}\hat{\beta}_{\adj, U_{\dag}(\hat c)}\bigr]^2 \notag \\
    & + \frac{2}{n} \sum_{k=1}^K \sum_{i \in [k]} \bigl[\{U_{i,\dag}(c)-\bar{U}_{[k], \dag}(c)\}^{\top}\hat{\beta}_{\adj, U_{\dag}(\hat c)}-\{U_{i,\dag}(c)-\bar{U}_{[k], \dag}(c)\}^{\top}\hat{\beta}_{\adj, U_{\dag}(c)}\bigr]\notag \\
    &~~~~~~~~~~~~\quad \bigl[\{U_{i,\dag}(\hat c)-\bar{U}_{[k], \dag}(\hat c)\}^{\top}\hat{\beta}_{\adj, U_{\dag}(\hat c)}-\{U_{i,\dag}(c)-\bar{U}_{[k], \dag}(c)\}^{\top}\hat{\beta}_{\adj, U_{\dag}(\hat c)}\bigr].\notag 
\end{align}   
The first term is $o_p(1)$ by \Cref{lm2} and \Cref{lm4}, and the second term is $o_p(1)$ by \Cref{lm4}. By Cauchy--Schwarz inequality, the whole term is $o_p(1)$. Therefore, $\hat{\sigma}_{\adj, \dag}^{2} \xrightarrow{P} \varsigma_{r_{\dag}}^2 (1-\pi) + \varsigma_{H{r_{\dag}}}^2 = \sigma^2_{\dag} $ when $\pi = 1/2$.

Finally, the relationship between different asymptotic variances ($\sigma_{\mim}^2  \leq \sigma_{\imp}^2  \leq  
 \sigma_{\ccov}^2 \leq \sigma_{\adj}^2$) can be derived from \Cref{th2}, since when $\pi = 1/2$, the asymptotic variances  are the same as those in \Cref{th2}. We will prove this result in the following proof of \Cref{th2}.
	

\end{proof}

\subsection{Proof of \Cref{th2}}

\begin{proof}

By the proof of Theorem 3 of \cite{Ma:2020}, with $X_i$ replaced by $U_{i, \dag}(\hat c)$ for ${\dag} \in \{\ccov,\imp,\mim\}$, $\hat{\tau}_{\int, \dag}$ has the following expression:
\begin{align*}
   \hat{\tau}_{\int, \dag} = \sum_{k=1}^K p_{n[k]} \bigg\{ & \bigl[\bar{Y}_{[k]1} - \{\bar{U}_{[k]1, \dag}(\hat c) - \bar{U}_{[k], \dag}(\hat c)\}^\top \hat{\beta}_{U_{\dag}(\hat c)}(1)\bigr] \\
   &- \bigl[\bar{Y}_{[k]0} - \{\bar{U}_{[k]0, \dag}(\hat c) - \bar{U}_{[k], \dag}(\hat c)\}^\top \hat{\beta}_{U_{\dag}(\hat c)}(0)\bigr]\biggr\}, 
\end{align*}
where $\hat{\beta}_{U_{\dag}(\hat c)}(a) = {{S}_{\tilde{U}_{\dag}(\hat c)\tilde{U}_{\dag}(\hat c)}^{-1}}(a){S}_{\tilde{U}_{\dag}(\hat c)\tilde{Y}(a)}$, $a=0,1$, with
\[
	{S}_{\tilde{U}_{\dag}(\hat c)\tilde{U}_{\dag}(\hat c)}(1) = \frac{1}{n_1} \sum_{k=1}^K \sum_{i \in [k]} A_i \{U_{i, \dag}(\hat c) - \bar{U}_{[k]1, \dag}(\hat c)\}\{U_{i, \dag}(\hat c) - \bar{U}_{[k]1, \dag}(\hat c)\}^\top, 
\]
\[
	{S}_{\tilde{U}_{\dag}(\hat c)\tilde{U}_{\dag}(\hat c)}(0) = \frac{1}{n_0} \sum_{k=1}^K \sum_{i \in [k]} (1-A_i) \{U_{i, \dag}(\hat c) - \bar{U}_{[k]0, \dag}(\hat c)\}\{U_{i, \dag}(\hat c) - \bar{U}_{[k]0, \dag}(\hat c)\}^\top,
\]
\[
	{S}_{\tilde{U}_{\dag}(\hat c)\tilde{Y}(1)} = \frac{1}{n_1} \sum_{k=1}^K \sum_{i \in [k]} A_i \{U_{i, \dag}(\hat c) - \bar{U}_{[k]1, \dag}(\hat c)\}{(Y_i - \bar{Y}_{[k]1})}, \]
\[ {S}_{\tilde{U}_{\dag}(\hat c)\tilde{Y}(0)} = \frac{1}{n_0} \sum_{k=1}^K \sum_{i \in [k]} (1-A_i) \{U_{i, \dag}(\hat c) - \bar{U}_{[k]0, \dag}(\hat c)\}{(Y_i - \bar{Y}_{[k]0})}.
\]

	Next, we prove that
\[
	{S}_{\tilde{U}_{\dag}(\hat c)\tilde{U}_{\dag}(\hat c)}(a) - \Sigma_{\tilde{U}_{\dag}(c)\tilde{U}_{\dag}(c)} = o_p(1),~~ {S}_{\tilde{U}_{\dag}(\hat c)\tilde{Y}(a)} - \Sigma_{\tilde{U}_{\dag}(c)\tilde{Y}(a)} = o_p(1),~~\hat{\beta}_{U_{\dag}(\hat c)}(a) - {\beta}_{U_{\dag}(c)}(a) = o_p(1).
\]

We will only prove the results for $a=1$, as the proof for $a=0$ is similar. When the imputation value is a constant $c$, the above statements are direct results of Lemma 7 of \cite{Ma:2020}. Thus, we only need to prove that these conclusions hold with $c$ replaced by $\hat c$. Note that
$$
{S}_{\tilde{U}_{\dag}(\hat c)\tilde{U}_{\dag}(\hat c)}(1) = \sumk \frac{\nkt}{n_1} \frac{1}{\nkt} \sum_{i \in [k]} A_iU_{i, \dag}({\hat c})\{U_{i, \dag}(\hat c)\}^\top - \sumk \frac{\nkt}{n_1} \bar{U}_{[k]1, \dag}(\hat c) \{\bar{U}_{[k]1, \dag}(\hat c)\}^\top.
$$

By \Cref{lm2}, \Cref{lm3}, and the fact that ${\nkt} / {n_1} - \pk = o_p(1)$, we have
\begin{align}
{S}_{\tilde{U}_{\dag}(\hat c)\tilde{U}_{\dag}(\hat c)}(1) &= \sumk \frac{\nkt}{n_1} \frac{1}{\nkt} \sum_{i \in [k]} A_iU_{i, \dag}({c})\{U_{i, \dag}(c)\}^\top - \sumk \frac{\nkt}{n_1} \bar{U}_{[k]1, \dag}(c) \{\bar{U}_{[k]1, \dag}(c)\}^\top + o_p(1)  \notag \\
&= {S}_{\tilde{U}_{\dag}(c)\tilde{U}_{\dag}(c)}(1)  + o_p(1). \notag 
\end{align}

Combined with the fact that ${S}_{\tilde{U}_{\dag}(c)\tilde{U}_{\dag}(c)}(1) - \Sigma_{\tilde{U}_{\dag}(c)\tilde{U}_{\dag}(c)} = o_p(1)$, we have ${S}_{\tilde{U}_{\dag}(\hat c)\tilde{U}_{\dag}(\hat c)}(1) - \Sigma_{\tilde{U}_{\dag}(c)\tilde{U}_{\dag}(c)} = o_p(1)$. Similarly, we have ${S}_{\tilde{U}_{\dag}(\hat c)\tilde{Y}(a)} - \Sigma_{\tilde{U}_{\dag}(c)\tilde{Y}(a)} = o_p(1)$. The third statement is a direct result of the first two statements.

Using these properties and Lemma \ref{lm2}, we have 
\begin{align}
\hat{\tau}_{\int, \dag} =& \sum_{k=1}^K p_{n[k]} \bigg\{\bigl[\bar{Y}_{[k]1} - \{\bar{U}_{[k]1, \dag}(c) - \bar{U}_{[k], \dag}(c)\}^\top {\hat{\beta}_{U_{\dag}(c)}(1)}\bigr] - \notag \\
& \{\bar{Y}_{[k]0} - \bigl[\bar{U}_{[k]0, \dag}(c) - \bar{U}_{[k], \dag}(c)\}^\top \hat{\beta}_{U_{\dag}}(0)\bigr]\biggr\} + o_p(n^{-1/2}) \notag \\
=& \sum_{k=1}^K p_{n[k]} \biggl\{\bar{Y}_{[k]1} - \bar{Y}_{[k]0} - \{\bar{U}_{[k]1, \dag}(c) - \bar{U}_{[k]0, \dag}(c)\}^\top \bigl[(1-\pi_{n[k]})\hat{\beta}_{U_{\dag}(c)}(1) + \pi_{n[k]} \hat{\beta}_{U_{\dag}(c)}(0)\bigr]\biggr\} \notag \\
&+ o_p(n^{-1/2}) \notag \\
=& \sum_{k=1}^K p_{n[k]} \bigl[\bar{Y}_{[k]1} - \bar{Y}_{[k]0} - \{\bar{U}_{[k]1, \dag}(c) - \bar{U}_{[k]0, \dag}(c)\}^\top \beta_{U_{\dag}(c)}\bigr] + o_p(n^{-1/2}) \notag \\
&- \sum_{k=1}^K p_{n[k]}\{{\bar{U}_{[k]1, \dag}(c) - \bar{U}_{[k]0, \dag}(c)}\}^\top \{(1-\pi_{n[k]}){\hat{\beta}_{U_{\dag}(c)}(1)} + \pi_{n[k]} \hat{\beta}_{U_{\dag}(c)}(0) - {\beta}_{U_{\dag}(c)}\}. \notag 
\end{align}

    By Theorem 4.1 of \cite{Bugni:2018}, we have $\bar{U}_{[k]1, \dag}(c) - \bar{U}_{[k]0, \dag}(c) = O_p(n^{-1/2})$. By Lemma \ref{lm2}, the conclusion also holds if we replace $c$ with $\hat{c}$. Noting that $p_{n[k]}$ and $\pi_{n[k]}$ have finite probability limits and $\hat{\beta}_{U_{\dag}(\hat c)}(a) - {\beta}_{U_{\dag}(c)}(a) = o_p(1)$, we have
\[
	\hat{\tau}_{\int, \dag} = \sum_{k=1}^K p_{n[k]} \bigl[\bar{Y}_{[k]1} - \bar{Y}_{[k]0} - \{\bar{U}_{[k]1, \dag}(c) - \bar{U}_{[k]0, \dag}(c)\}^\top {\beta_{U_{\dag}(c)}}\bigr] + o_p(n^{-1/2}).
\]

	Denote the first item of the right-hand side as $\check{\tau}_{\int, \dag}$. By Proposition 3 of \cite{Ma:2020} and $E[\{U_{i, \dag}(c)\}(1)] - E\bigl[\{U_{i, \dag}(c)\}(0)\bigr] = 0$, we have 
\[
	\sqrt{n}(\check{\tau}_{\int, \dag}-\tau)\xrightarrow{d} N(0,\varsigma_{r_{\dag}}^2 (\pi) + \varsigma_{H{r_{\dag}}}^2).
\]
By Slutsky's theorem, we have
\[
	\sqrt{n}(\hat{\tau}_{\int, \dag}-\tau)\xrightarrow{d} N(0,\varsigma_{r_{\dag}}^2 (\pi) + \varsigma_{H{r_{\dag}}}^2).
\]

	Next, we prove the consistency of the variance estimator. When the interpolation value is a constant $c$, the result is a direct application of Theorem 3 of \cite{Ma:2020}. Otherwise, if $\sum_{k=1}^K ({p_{n[k]}}/{n_{[k]1}}) \sum_{i \in [k]} A_i \{U_{i,\dag}(\hat c) - \bar{U}_{[k]1, \dag}(\hat c)\}\{U_{i,\dag}(\hat c) - \bar{U}_{[k]1, \dag}(\hat c)\}^{\top} - \Sigma_{\tilde{U, \dag}^{}( c)\tilde{U, \dag}^{}(c)} = o_p(1)$ and $\hat \beta_{\int, U_{\dag}(\hat c) }(a) - \beta_{U_{\dag}(c) }(a) = o_p(1)$, then the imputation does not affect the asymptotic limit of the variance estimator, and thus, $\hat \sigma_{\int,\dag}^2 \xrightarrow{P} \sigma^2_{\dag}$. Therefore, we only need to verify these two conditions. For the first condition, note that when $\hat c = c$, it is a direct result of the weak law of large numbers under covariate-adaptive randomization. Otherwise, by \Cref{lm3}, we have
    \begin{align*}
        &\sum_{k=1}^K ({p_{n[k]}}/{n_{[k]1}}) \sum_{i \in [k]} A_i \{U_{i,\dag}(\hat c) - \bar{U}_{[k]1, \dag}(\hat c)\}\{U_{i,\dag}(\hat c) - \bar{U}_{[k]1, \dag}(\hat c)\}^{\top} \\
        =&\sum_{k=1}^K ({p_{n[k]}}/{n_{[k]1}}) \sum_{i \in [k]} A_i \{U_{i,\dag}(c) - \bar{U}_{[k]1, \dag}(c)\}\{U_{i,\dag}(c) - \bar{U}_{[k]1, \dag}(c)\}^{\top} + o_p(1).
    \end{align*}
The second condition $\hat{\beta}_{U_{\dag}(\hat c)}(a) - {\beta}_{U_{\dag}(c)}(a) = o_p(1)$ has been proved already. Therefore, the consistency of the variance estimator is proved.


Finally, we derive the order of different asymptotic variances.  
By definition, $r_{i, \dag}(a) = Y_i(a) - U_{i, \dag}(c)^\top \beta_{U_{\dag}(c)}$. Therefore,
\begin{align*}
\varsigma_{Hr_{\dag}}^2 &=E \big([ E\{ r_{i,\dag}(1) \mid B_i \} - E\{ r_{i,\dag}(1) \} ] - [ E\{ r_{i,\dag}(0) \mid B_i \} - E\{ r_{i,\dag}(0) \} ] \big)^2 \\
&= E\big(E\{Y_i(1) - Y_i(0) \vert B_i \} - E\{Y_i(1) - Y_i(0)\}\big)^2, 
\end{align*}
which is irrelevant to $U_{i, \dag}(c)$, and thus $\varsigma_{H{r_{\dag}}}^2$ is irrelevant to $U_{i, \dag}(c)$.

Recall that $\tilde V_i = V_i - E(V_i \mid B_i)$. For the first term in the asymptotic variance, we have
\begin{align}
\sigma_{r(a)}^2 &= \var\bigl[Y_i(a) - U_{i, \dag}(c)^\top\beta_{U_{\dag}(c)} - E\{Y_i(a) - U_{i, \dag}(c)^\top\beta_{U_{\dag}(c)} \vert B_i\}\bigr] \notag \\
&= E\bigl[Y_i(a) - U_{i, \dag}(c)^\top\beta_{U_{\dag}(c)} - E\{Y_i(a) - U_{i, \dag}(c)^\top\beta_{U_{\dag}(c)} \vert B_i\}\bigr]^2 \notag \\
&= E\{\tilde{Y}_i(a) - \tilde{U}_{i, \dag}(c)^\top\beta_{U_{\dag}(c)}\}^2 \notag \\
&= E\{\tilde{Y}_i(a)\}^2 - 2\cov\{\tilde{U}_{i, \dag}(c), Y_i(a)\}^\top\beta_{U_{\dag}(c)} + \beta_{U_{\dag}(c)}^\top\Sigma_{\tilde{U}_{\dag}(c)\tilde{U}_{\dag}(c)}\beta_{U_{\dag}(c)} \notag \\
&= E\{\tilde{Y}_i(a)\}^2 - 2\beta_{U_{\dag}(c)}(a)^\top\Sigma_{\tilde{U}_{\dag}(c)\tilde{U}_{\dag}(c)}\beta_{U_{\dag}(c)} + \beta_{U_{\dag}(c)}^\top\Sigma_{\tilde{U}_{\dag}(c)\tilde{U}_{\dag}(c)}\beta_{U_{\dag}(c)}. \notag
\end{align}
The first term is irrelevant to $U_{i, \dag}$, so we only need to consider the latter two terms. Simple calculation gives
\begin{align}
\varsigma_{r_{\dag}}^2 (\pi) =& \frac{1}{\pi}\beta_{U_{\dag}(c)}^\top\Sigma_{\tilde{U}_{\dag}(c)\tilde{U}_{\dag}(c)}\beta_{U_{\dag}(c)} - \frac{2}{\pi}\beta_{U_{\dag}(c)}(1)^\top\Sigma_{\tilde{U}_{\dag}(c)\tilde{U}_{\dag}(c)}\beta_{U_{\dag}(c)} + \frac{1}{1-\pi} \beta_{U_{\dag}(c)}^\top\Sigma_{\tilde{U}_{\dag}(c)\tilde{U}_{\dag}(c)}\beta_{U_{\dag}(c)} \notag \\
&- \frac{2}{1-\pi}\beta_{U_{\dag}(c)}(0)^\top\Sigma_{\tilde{U}_{\dag}(c)\tilde{U}_{\dag}(c)}\beta_{U_{\dag}(c)} + Const \notag \\
=& -\{\pi(1-\pi)\}^{-1}\beta_{U_{\dag}(c)}^\top\Sigma_{\tilde{U}_{\dag}(c)\tilde{U}_{\dag}(c)}\beta_{U_{\dag}(c)} + Const. \notag
\end{align}
\noindent where $Const$ is a quantity irrelevant to $U_{i, \dag}(c)$, and the last equation is due to $\beta_{U_{\dag}(c)} = (1-\pi)\beta_{U_{\dag}(c)}(1) + \pi\beta_{U_{\dag}(c)}(0)$. We can rewrite $\beta_{U_{\dag}(c)}^\top\Sigma_{\tilde{U}_{\dag}(c)\tilde{U}_{\dag}(c)}\beta_{U_{\dag}(c)}$ as
\[
	\beta_{U_{\dag}(c)}^\top\Sigma_{\tilde{U}_{\dag}(c)\tilde{U}_{\dag}(c)}\beta_{U_{\dag}(c)} = \{(1-\pi)\Sigma_{\tilde{U}_{\dag}(c)\tilde{Y}(1)} + \pi \Sigma_{\tilde{U}_{\dag}(c)\tilde{Y}(0)}\}^\top\Sigma_{\tilde{U}_{\dag}(c)\tilde{U}_{\dag}(c)}^{-1}\{(1-\pi)\Sigma_{\tilde{U}_{\dag}(c)\tilde{Y}(1)} + \pi \Sigma_{\tilde{U}_{\dag}(c)\tilde{Y}(0)}\}.
\]

It suffices for $\sigma_{\mim}^2  \leq \sigma_{\imp}^2  \leq  
 \sigma_{\ccov}^2 \leq \sigma_{\int}^2$ to show that if $U_{i,(1)}$ is a subvector of $U_{i,(2)}$ for $i = 1, \ldots, n$, then $\beta_{U_{(1)}}^\top\Sigma_{\tilde{U}_{(1)}\tilde{U}_{(1)}}\beta_{U_{(1)}} \leq \beta_{U_{(2)}}^\top\Sigma_{\tilde{U}_{(2)}\tilde{U}_{(2)}}\beta_{U_{(2)}}$. To ease notation, denote $(1-\pi)\Sigma_{\tilde{U}_{(1)}\tilde{Y}(1)} + \pi \Sigma_{\tilde{U}_{(1)}\tilde{Y}(0)}$ as $V_{\tilde{U}_{(1)}\tilde{Y}}$ and $U_{i,(2)} = ( U_{i,(1)}^\top, U_{i,(3)}^\top )^\top$.
Simple calculation gives
\begin{align}
 & \beta_{U_{(1)}}^\top\Sigma_{\tilde{U}_{(1)}\tilde{U}_{(1)}}\beta_{U_{(1)}} - \beta_{U_{(2)}}^\top\Sigma_{\tilde{U}_{(2)}\tilde{U}_{(2)}}\beta_{U_{(2)}} \notag \\
 =& V_{\tilde{U}_{(1)}\tilde{Y}}^\top\Sigma_{\tilde{U}_{(1)}\tilde{U}_{(1)}}^{-1}V_{\tilde{U}_{(1)}\tilde{Y}} - 
\left( 
\begin{matrix}
	V_{\tilde{U}_{(1)}\tilde{Y}}^\top  & V_{\tilde{U}_{(3)}\tilde{Y}}^\top
\end{matrix}
\right)
\left( 
\begin{matrix}
		\Sigma_{\tilde{U}_{(1)}\tilde{U}_{(1)}}^{-1} + FDF^\top  & B \\
		B^\top                                                   & D 
\end{matrix}
\right)
\left( 
\begin{matrix}
	V_{\tilde{U}_{(1)}\tilde{Y}}  \\ V_{\tilde{U}_{(3)}\tilde{Y}}
\end{matrix}
\right) \notag \\
=& -\Big(V_{\tilde{U}_{(1)}\tilde{Y}}^\top FDF^\top V_{\tilde{U}_{(1)}\tilde{Y}} + 2V_{\tilde{U}_{(1)}\tilde{Y}}^\top BV_{\tilde{U}_{(3)}\tilde{Y}} + V_{\tilde{U}_{(3)}\tilde{Y}}^\top D V_{\tilde{U}_{(3)}\tilde{Y}} \Big) \notag \\
=& -\Big(F^\top V_{\tilde{U}_{(1)}\tilde{Y}} - V_{\tilde{U}_{(3)}\tilde{Y}})^\top D(F^\top V_{\tilde{U}_{(1)}\tilde{Y}} - V_{\tilde{U}_{(3)}\tilde{Y}} \Big) \notag \leq 0,
\end{align}
\noindent where
\[
	F = \Sigma_{\tilde{U}_{(1)}\tilde{U}_{(1)}}^{-1}\Sigma_{\tilde{U}_{(1)}\tilde{U}_{(3)}},~~~~ B = -\Sigma_{\tilde{U}_{(1)}\tilde{U}_{(1)}}^{-1}\Sigma_{\tilde{U}_{(1)}\tilde{U}_{(3)}}D = -FD,
\]
\[
	D^{-1} = \Sigma_{\tilde{U}_{(3)}\tilde{U}_{(3)}} - \Sigma_{\tilde{U}_{(3)}\tilde{U}_{(1)}}\Sigma_{\tilde{U}_{(1)}\tilde{U}_{(1)}}^{-1}\Sigma_{\tilde{U}_{(1)}\tilde{U}_{(3)}}.
\]

\end{proof}

\subsection{Proof of Theorem \ref{th3}}

\begin{proof}
    
We first derive the expression of $\hat{\tau}_{\tom, {\dag}}^{}$, which is a direct result of Proposition 4 in \cite{Lu:2022}.

\begin{lemma}{\cite[][Proposition 4]{Lu:2022}}
\label{lm7}
For $\dag \in \{\ccov, \imp, \mim\}$, $\hat{\tau}_{\tom, {\dag}}^{} = \hat{\tau}_{\int} - \hat{\beta}_{\tom, U_{\dag}(\hat c)}^\top \hat{\tau}_{\int, U_{\dag} (\hat c) }$, where
\[
    \hat{\tau}_{\int} = \sum_{k=1}^K p_{n[k]}(\bar{Y}_{[k]1}-\bar{Y}_{[k]0}),
~~~~\hat{\tau}_{\int, U_{\dag} (\hat c) } = \sum_{k=1}^K p_{n[k]}\{\bar{U}_{[k]1, \dag}(\hat c)-\bar{U}_{[k]0,\dag}(\hat c)\},
\]
\begin{align*}
    \hat{\beta}_{\tom, U_{\dag}(\hat c)} =& \bigg[\sum_{k=1}^K p_{n[k]}\big\{\pi_{n[k]}^{-1}\hat{S}_{ U_{\dag}(\hat c) U_{\dag}(\hat c)[k]}^{(1)} + (1-\pi_{n[k]})^{-1}\hat{S}_{ U_{\dag}(\hat c) U_{\dag}(\hat c)[k]}^{(0)}\big\}\bigg]^{-1} \\   
    &\bigg[\sum_{k=1}^K p_{n[k]}\big\{\pi_{n[k]}^{-1}\hat{S}_{ U_{\dag}(\hat c)Y(1)[k]} + (1-\pi_{n[k]})^{-1}\hat{S}_{ U_{\dag}(\hat c)Y(0)[k]}\big\}\bigg],
\end{align*}
\[
    \hat{S}_{ U_{\dag}(\hat c) U_{\dag}(\hat c)[k]}^{(1)} = \frac{1}{n_{[k]1}} \sum_{i \in [k]} A_i\{U_{i,\dag}(\hat c) - \bar{U}_{[k]1, \dag}(\hat c)\}\{U_{i,\dag}(\hat c) - \bar{U}_{[k]1, \dag}(\hat c)\}^\top,
\]
\[
    \hat{S}_{ U_{\dag}(\hat c) U_{\dag}(\hat c)[k]}^{(0)} = \frac{1}{n_{[k]0}} \sum_{i \in [k]} (1-A_i)\{U_{i,\dag}(\hat c) - \bar{U}_{[k]0,\dag}(\hat c)\}\{U_{i,\dag}(\hat c) - \bar{U}_{[k]0,\dag}(\hat c)\}^\top,
\]
\[
    \hat{S}_{ U_{\dag}(\hat c)Y(1)[k]} = \frac{1}{n_{[k]1}} \sum_{i \in [k]} A_i\{U_{i,\dag}(\hat c) - \bar{U}_{[k]1, \dag}(\hat c)\}(Y_i - \bar{Y}_{[k]1}),
\]
\[
    \hat{S}_{ U_{\dag}(\hat c)Y(0)[k]} = \frac{1}{n_{[k]0}} \sum_{i \in [k]} (1-A_i)\{U_{i,\dag}(\hat c) - \bar{U}_{[k]0,\dag}(\hat c)\}(Y_i - \bar{Y}_{[k]0}).
\]

\end{lemma}
    
    First, notice that
\begin{align}
\hat{\tau}_{\tom, {\dag}}^{} =& \hat{\tau}_{\int} - \hat{\beta}_{\tom, U_{\dag}(\hat c)}^\top \hat{\tau}_{\int, U_{\dag} (\hat c) } 
= \sum_{k=1}^K p_{n[k]}\bigl[\bar{Y}_{[k]1}-\bar{Y}_{[k]0}-\hat{\beta}_{\tom, U_{\dag}(\hat c)}^\top\{\bar{U}_{[k]1, \dag}(\hat c) - \bar{U}_{[k]0, \dag}(\hat c)\}\bigr] \notag \\
=& \sum_{k=1}^K p_{n[k]}\bigl[\bar{Y}_{[k]1}-\bar{Y}_{[k]0}-{\beta}_{U_{\dag}(c)}^\top\{\bar{U}_{[k]1, \dag}(\hat c) - \bar{U}_{[k]0, \dag}(\hat c)\}\bigr] \notag \\
&+ \sum_{k=1}^K p_{n[k]}\{{\beta}_{U_{\dag}(c)} - \hat{\beta}_{\tom, U_{\dag}(\hat c)}\}^\top\{\bar{U}_{[k]1, \dag}(\hat c) - \bar{U}_{[k]0, \dag}(\hat c)\}. \notag
\end{align}
Then we prove the following properties:
\[
    \hat{S}_{U_{\dag}(\hat c)U_{\dag}(\hat c)[k]}^{(a)} - \Sigma_{U_{\dag}(c)U_{\dag}(c)[k]} = o_p(1),~~~~a = 0, 1;
\]
\[
    \hat{S}_{U_{\dag}(\hat c)Y(a)[k]} - \Sigma_{U_{\dag}(c)Y(a)[k]} = o_p(1),~~~~a= 0, 1;
\]
\[
    \hat{\beta}_{\tom, U_{\dag}(\hat c)} - {\beta}_{U_{\dag}(c)} = o_p(1),
\]
where $\Sigma_{PQ[k]} = E\bigl[\{P-E(P \vert B_i = k)\}\{Q-E(Q \vert B_i = k)\}^\top \vert B_i = k\bigr]$. 
	The proof of the first two conclusions is straightforward, following the same method as in the proof of \Cref{th1}. Noting that $p_{n[k]} = \pk + o_p(1)$ and $\pi_{n[k]} = \pi + o_p(1)$, the third conclusion directly follows from the first two.
    
    We have shown in the proof of \Cref{th2} that $\bar{U}_{[k]1, \dag}(\hat c) - \bar{U}_{[k]0, \dag}(\hat c) = O_p(n^{-1/2})$. Together with $\hat{\beta}_{\tom, U_{\dag}(\hat c)}-{\beta}_{U_{\dag}(c)}=o_p(1)$, we have
\[
    \sum_{k=1}^K p_{n[k]}\{{\beta}_{U_{\dag}(c)} - \hat{\beta}_{\tom, U_{\dag}(\hat c)}\}^\top\{\bar{U}_{[k]1, \dag}(\hat c) - \bar{U}_{[k]0, \dag}(\hat c)\} = o_p(n^{-1/2}).
\]
    In the proof of \Cref{th2}, we have obtained
\[
    \hat{\tau}_{\int,\dag} = \sum_{k=1}^K p_{n[k]}\bigl[\bar{Y}_{[k]1}-\bar{Y}_{[k]0}-{\beta}_{U_{\dag}(c)}^\top\{\bar{U}_{[k]1, \dag}(\hat c) - \bar{U}_{[k]0, \dag}(\hat c)\}\bigr] + o_p(n^{-1/2}). \notag
\]
Therefore, $\hat{\tau}_{\tom, {\dag}}^{} - \hat{\tau}_{\int,\dag} = o_p(n^{-1/2})$, that is, $\hat{\tau}_{\tom, {\dag}}^{}$ and $\hat{\tau}_{\int,\dag}$ has the same asymptotic distribution. Thus, $\sqrt{n}(\hat{\tau}_{\tom, \dag}-\tau)\xrightarrow{d} N(0,\sigma_{\dag}^2)
$. Noting that $\hat{\beta}_{\tom, U_{\dag}(\hat c)} - {\beta}_{U_{\dag}(c)} = o_p(1)$, the proof of the consistency of the variance estimator follows from that of \Cref{th2} if we replace $\hat{\beta}_{U_{\dag}(\hat c)}$ by $\hat{\beta}_{\tom, U_{\dag}(\hat c)}$. Finally, the relationship among the asymptotic variances still holds since they are the same as those in \Cref{th2}.

\end{proof}

\subsection{Proof of Theorem \ref{th4}}

\begin{proof}

    (i) {\bf Stratum-specific Fisher's regression}.	First, we show that
\[
	\hat{\tau}_{\adj, \dag, \ss} = \sum_{k=1}^K p_{n[k]} \bigl[\bar{Y}_{[k]1} - \bar{Y}_{[k]0} - \{\bar{U}_{[k]1, \dag}(c) - \bar{U}_{[k]0, \dag}(c)\}^\top \beta_{U_{\dag}(c)[k]}\bigr] + o_p(n^{-1/2}).
\]
Within each stratum, using the same procedure as in the proof of \Cref{th1},  we have
\[
	\hat{\tau}_{\adj [k], \dag} = \bar{Y}_{[k]1} - \bar{Y}_{[k]0} - \{\bar{U}_{[k]1, \dag}(c) - \bar{U}_{[k]0, \dag}(c)\}^\top \beta_{U_{\dag}(c)[k]} + o_p(n^{-1/2}).
\]  
	Noting that $\hat{\tau}_{\adj, \dag, \ss} = \sum_{k=1}^K p_{n[k]} \hat{\tau}_{\adj [k], \dag, \ss}$, we have
\begin{align}
\hat{\tau}_{\adj, \dag, \ss} &= \sum_{k=1}^K p_{n[k]} \bigl[\bar{Y}_{[k]1} - \bar{Y}_{[k]0} - \{\bar{U}_{[k]1, \dag}(c) - \bar{U}_{[k]0, \dag}(c)\}^\top \beta_{U_{\dag}(c)[k]} + o_p(n^{-1/2})\bigr] \notag \\
&= \sum_{k=1}^K p_{n[k]} \bigl[\bar{Y}_{[k]1} - \bar{Y}_{[k]0} - \{\bar{U}_{[k]1, \dag}(c) - \bar{U}_{[k]0, \dag}(c)\}^\top \beta_{U_{\dag}(c)[k]}\bigr] + o_p(n^{-1/2}) \notag \\
&:=\check{\tau}_{\adj, \dag, \ss} + o_p(n^{-1/2}). \notag
\end{align}
Applying Proposition 3 of \cite{Ma:2020} and Slutsky's theorem immediately yields
\[
	\sqrt{n}(\hat{\tau}_{\adj, \dag, \ss}-\tau) \xrightarrow{d} N(0,\varsigma_{r_{\dag, \ss}}^2 (\pi) + \varsigma_{H{r_{\dag, \ss}}}^2).
\]

Define the following quantities:
$$
\varsigma_{r[k]}^2(\pi) = \pi^{-1} \var\big[ r_i(1) - E\{ r_i(1) \mid B_i \} \mid B_i = k\big] +  (1 - \pi)^{-1} \var\big[ r_i(0) - E\{ r_i(0) \mid B_i \}\mid B_i = k\big],
$$
$$
\varsigma_{Hr[k]}^2 =E \big([ E\{ r_i(1) \mid B_i \} - E\{ r_i(1) \} ] - [ E\{ r_i(0) \mid B_i \} - E\{ r_i(0) \} ]\mid B_i = k \big)^2.
$$
\begin{eqnarray*}
\hat{\varsigma}_{r[k]}^2(\pi) &=& \frac{1}{\pi} \Bigl\{\frac{1}{n_{[k]1}} \sum_{i \in [k]}A_i{(r_i-\bar{r}_{[k]1})}^2\Bigr\} +  \frac{1}{1- \pi} \Bigl\{\frac{1}{n_{[k]0}} \sum_{i \in [k]}(1-A_i)  {(r_i-\bar{r}_{[k]0})}^2\Bigr\},    
\end{eqnarray*}
\begin{eqnarray*}
\hat{\varsigma}_{Hr[k]}^2 = {\bigl\{(\bar{r}_{[k]1}-\bar{r}_{1})-(\bar{r}_{[k]0}-\bar{r}_{0})\bigr\}}^2. 
\end{eqnarray*}
 
    For the OLS variance estimator, applying \Cref{th1} to each stratum, we have $\hat{\sigma}_{\adj [k], \dag}^2 \xrightarrow{P}  \sigma^2_{[k], \dag, \ss}$, where $ \sigma^2_{[k], \dag, \ss} = \varsigma_{r_{\dag,\ss}[k]}^2 (\pi) + \varsigma_{Hr_{\dag,\ss}[k]}^2$. Therefore, when $\pi= 1/2$, we have $\hat \sigma^2_{\adj, \dag, \ss} \xrightarrow{P} \sumk \pk \sigma^2_{[k], \dag, \ss} = \sigma_{\dag,\ss}^2$ because of $p_{n[k]} \xrightarrow{P} p_{[k]}$ and the continuity theorem.

	(ii) {\bf Stratum-specific Lin's regression}. Note that
\begin{align*}
\hat{\tau}_{\int,\dag,\ss} = & \sum_{k=1}^K p_{n[k]}  \bigg\{\bigl[\bar{Y}_{[k]1} - \{\bar{U}_{[k]1, \dag}(\hat c) - \bar{U}_{[k], \dag}(\hat c)\}^\top \hat \beta_{U_{\dag}(\hat c)[k]}(1)\bigr] \\
&  - \bigl[\bar{Y}_{[k]0} - \{\bar{U}_{[k]0, \dag}(\hat c) - \bar{U}_{[k], \dag}(\hat c)\}^\top \hat\beta_{U_{\dag}(\hat c)[k]}(0)\bigr]\biggr\} \\
:= & \sumk p_{n[k]} \hat{\tau}_{\int [k],\dag,\ss}.
\end{align*}
Within each stratum, using the same procedure as in the proof of \Cref{th2},  we have
\[
	\hat{\tau}_{\int [k], \dag, \ss} = \bar{Y}_{[k]1} - \bar{Y}_{[k]0} - \{\bar{U}_{[k]1, \dag}(c) - \bar{U}_{[k]0, \dag}(c)\}^\top {\beta}_{U_{\dag}(c)[k]} + o_p(n^{-1/2}).
\]  
	Noting that $\hat{\tau}_{\int, \dag, \ss} = \sum_{k=1}^K p_{n[k]} \hat{\tau}_{\int [k], \dag, \ss}$, we have
\begin{align}
\hat{\tau}_{\int, \dag, \ss} &= \sum_{k=1}^K p_{n[k]} \bigl[\bar{Y}_{[k]1} - \bar{Y}_{[k]0} - \{\bar{U}_{[k]1, \dag}(c) - \bar{U}_{[k]0, \dag}(c)\}^\top {\beta}_{U_{\dag}(c)[k]} + o_p(n^{-1/2})\bigr] \notag \\
&= \sum_{k=1}^K p_{n[k]} \bigl[\bar{Y}_{[k]1} - \bar{Y}_{[k]0} - \{\bar{U}_{[k]1, \dag}(c) - \bar{U}_{[k]0, \dag}(c)\}^\top {\beta}_{U_{\dag}(c)[k]} \bigr] + o_p(n^{-1/2}) \notag \\
&:=\check{\tau}_{\int, \dag, \ss} + o_p(n^{-1/2}). \notag
\end{align}
Applying Proposition 3 of \cite{Ma:2020} and Slutsky's theorem immediately yield
\[
	\sqrt{n}(\hat{\tau}_{\int, \dag, \ss}-\tau) \xrightarrow{d} N(0,\varsigma_{r_{\dag, \ss}}^2 (\pi) + \varsigma_{H{r_{\dag, \ss}}}^2).
\]



Applying \Cref{th2} to each stratum $k$, we have $\hat \sigma^2_{\int [k],\dag,\ss} \xrightarrow{P} \sigma_{[k],\dag,\ss}^2$, where $\hat \sigma^2_{\int[k],\dag,\ss} = \hat{\varsigma}_{\hat{r}_{\int,\dag, \ss}[k]}^2 (\pi) + \hat{\varsigma}_{H{ \hat r_{\int,\dag, \ss}[k]}}^2 $ with $\hat \sigma^2_{\int, \dag,\ss} = \sumk p_{n[k]} \hat \sigma_{\int[k],\dag,\ss}^2$. Therefore, $\hat \sigma^2_{\int, \dag, \ss} \xrightarrow{P} \sumk \pk \sigma^2_{[k],\dag, \ss} = \sigma_{\dag,\ss}^2$ because of $p_{n[k]} \xrightarrow{P} p_{[k]}$ and the continuity theorem.

    (iii) {\bf Stratum-specific ToM regression}. Note that
\begin{align*}
\hat{\tau}_{\tom,\dag,\ss} = & \sum_{k=1}^K p_{n[k]}  \bigg\{\bigl[\bar{Y}_{[k]1} - \{\bar{U}_{[k]1, \dag}(\hat c) - \bar{U}_{[k], \dag}(\hat c)\}^\top \hat \beta_{\tom, U_{\dag}(\hat c)[k]}\bigr] \\
&  - \bigl[\bar{Y}_{[k]0} - \{\bar{U}_{[k]0, \dag}(\hat c) - \bar{U}_{[k], \dag}(\hat c)\}^\top \hat \beta_{\tom, U_{\dag}(\hat c)[k]}\bigr]\biggr\} \\
:= & \sumk p_{n[k]} \hat{\tau}_{\tom [k],\dag}.
\end{align*}
Within each stratum, using the same procedure as in the proof of \Cref{th3},  we have
\[
	\hat{\tau}_{\tom [k],\dag} = \bar{Y}_{[k]1} - \bar{Y}_{[k]0} - \{\bar{U}_{[k]1, \dag}(c) - \bar{U}_{[k]0, \dag}(c)\}^\top {\beta}_{U_{\dag}(c)[k]} + o_p(n^{-1/2}).
\]  


	Noting that $\hat{\tau}_{\tom, \dag, \ss} = \sum_{k=1}^K p_{n[k]} \hat{\tau}_{\tom [k],\dag}$, we have
\begin{align}
\hat{\tau}_{\tom, \dag, \ss} &= \sum_{k=1}^K p_{n[k]} \bigl[\bar{Y}_{[k]1} - \bar{Y}_{[k]0} - \{\bar{U}_{[k]1, \dag}(c) - \bar{U}_{[k]0, \dag}(c)\}^\top {\beta}_{U_{\dag}(c)[k]} + o_p(n^{-1/2})\bigr] \notag \\
&= \sum_{k=1}^K p_{n[k]} \bigl[\bar{Y}_{[k]1} - \bar{Y}_{[k]0} - \{\bar{U}_{[k]1, \dag}(c) - \bar{U}_{[k]0, \dag}(c)\}^\top {\beta}_{U_{\dag}(c)[k]} \bigr] + o_p(n^{-1/2}) \notag \\
&:=\check{\tau}_{\int, \dag, \ss} + o_p(n^{-1/2}). \notag
\end{align}
Applying Proposition 3 of \cite{Ma:2020} and Slutsky's theorem immediately yields
\[
	\sqrt{n}(\hat{\tau}_{\tom, \dag, \ss}-\tau) \xrightarrow{d} N(0,\varsigma_{r_{\dag, \ss}}^2 (\pi) + \varsigma_{H{r_{\dag, \ss}}}^2).
\]

Applying \Cref{th3} to each stratum $k$, we have $\hat \sigma^2_{\tom [k],\dag,\ss} \xrightarrow{P} \sigma_{[k],\dag,\ss}^2$, where $\hat \sigma^2_{\tom [k],\dag,\ss} = \hat{\varsigma}_{\hat{r}_{\tom,\dag, \ss}[k]}^2 (\pi) + \hat{\varsigma}_{H{ \hat r_{\tom,\dag, \ss}[k]}}^2 $ with $\hat \sigma^2_{\tom, \dag,\ss} = \sumk p_{n[k]} \hat \sigma_{\int[k],\dag}^2$. Therefore, $\hat \sigma^2_{\tom, \dag, \ss} \xrightarrow{P} \sumk \pk \sigma^2_{[k],\dag, \ss} = \sigma_{\dag,\ss}^2$ because of $p_{n[k]} \xrightarrow{P} p_{[k]}$ and the continuity theorem.

(iv) {\bf Variance comparison.} For the first part, it suffices to show that if $U_{i,(1)}$ is a subvector of $U_{i,(2)}$ for $i = 1, \ldots, n$, then $\sigma_{U_{(1)}, \ss}^2 \leq \sigma_{U_{(2)}, \ss}^2$, where $\sigma_{U_{(1)}, \ss}^2$ and $\sigma_{U_{(2)}, \ss}^2$ are the asymptotic variances of the stratum-specific estimators using $U_{(1)}$ and $U_{(2)}$ for covariate adjustments. By \Cref{th3}, we have $\sigma_{U_{(1)}[k], \ss}^2 \leq \sigma_{U_{(2)}[k], \ss}^2$, where $\sigma_{U_{(1)}[k], \ss}^2$ and $\sigma_{U_{(2)}[k], \ss}^2$ are the corresponding asymptotic variances within stratum $k$. Since $\sigma_{U_{(1)}, \ss}^2 = \sum_{k=1}^K p_{n[k]}\sigma_{U_{(1)}[k], \ss}^2$ and $\sigma_{U_{(2)}, \ss}^2 = \sum_{k=1}^K p_{n[k]}\sigma_{U_{(2)}[k], \ss}^2$, then we obtain the first inequality in (iv).

    For the second part, note that $\sigma^2_{[k], \dag, \ss} \leq \sigma^2_{[k], \dag}$ by using Theorem 4 of \cite{Ma:2020}. The desired conclusion holds since $\sigma^2_{\dag, \ss} = \sum_{k=1}^K p_{[k]}\sigma^2_{[k], \dag, \ss}$ and $\sigma^2_{\dag} = \sum_{k=1}^K p_{[k]}\sigma^2_{[k], \dag}$.

\end{proof}

\subsection{Proof of \Cref{ps2}}

\begin{proof}
    
Recall that $C_i = 1$ represents that the $i$th observation contains missing values and $C_i=0$ otherwise. Define $D_{n[k]}^{\prime}$ as follows:
\begin{align}
D_{n[k]}^{\prime} &= \sum_{i=1}^n (A_i - \pi)I_{i \in [k]}I_{C_1 = 0} \notag \\
&= \sum_{i=1}^n (A_i - \pi)I_{i \in [k]}P(C_i = 0 \vert B_i = k) + \sum_{i=1}^n (A_i - \pi)I_{i \in [k]}\{I_{C_i = 0} - P(C_i = 0 \vert B_i = k)\} \notag \\
&:= I_1 + I_2, \notag
\end{align}
\noindent where $I_1 = \sum_{i=1}^n (A_i - \pi)I_{i \in [k]}P(C_i = 0 \vert B_i = k)$, $I_2 = \sum_{i=1}^n (A_i - \pi)I_{i \in [k]}\{I_{C_i = 0} - P(C_i = 0 \vert B_i = k)\}$. Given $A^{(n)}$ and $B^{(n)}$,
\begin{align}
& n^{-1/2}I_2 \notag \\
=& n^{-1/2}\bigl[(1-\pi)\sum_{i \in [k], A_i = 1}\{I_{C_i = 0} - P(C_i = 0 \vert B_i = k)\} \notag \\
&+ \pi \sum_{i \in [k], A_i = 0}\{I_{C_i = 0} - P(C_i = 0 \vert B_i = k)\}\bigr] \notag \\
=& \sqrt{n} \biggl[\frac{1-\pi}{n}\sum_{i \in [k], A_i = 1}\{I_{C_i = 0} - P(C_i = 0 \vert B_i = k)\} - \frac{\pi}{n} \sum_{i \in [k], A_i = 0}\{I_{C_i = 0} - P(C_i = 0 \vert B_i = k)\}\biggr] \notag \\
\xrightarrow{d} & N\bigl(0, \pi(1-\pi)p_{[k]}P(C_i = 0 \vert B_i = k)\{1-P(C_i = 0 \vert B_i = k)\}\bigr). \notag
\end{align}
where the asymptotic normality is derived as follows:

	First, note that $I_{C_i = 0}$ are i.i.d. within a given stratum given $A^{(n)}$ and $B^{(n)}$. In fact, 
\[
	P(C_i = 0 \vert A^{(n)}, B^{(n)}) = P(C_i = 0 \vert B^{(n)}) = P(C_i = 0 \vert B_i),
\]
\[
	P(C_i = a, C_j = a^{\prime} \vert A^{(n)}, B^{(n)}) = P(C_i = a, C_j = a^{\prime} \vert B^{(n)}) = P(C_i = a, C_j = a^{\prime} \vert B_i, B_j).
\]
\noindent where $a, a^{\prime} \in \{0, 1\}$. The first expression shows that $I_{C_i = 0}$ are identically distributed within a given stratum given $A^{(n)}$ and $B^{(n)}$ because $(B_i, C_i)$ are identically distributed. The reason for independence is as follows:
\begin{align}
P(C_i = a, C_j = a^{\prime} \vert B_i, B_j) &=  P(C_i = a \vert B_i, B_j, C_j = a^{\prime})P(C_j = a^{\prime} \vert B_i, B_j)  \notag \\
&= P(C_i = a \vert B_i)P(C_j = a^{\prime} \vert B_j)  \notag \\
&= P(C_i = a \vert A^{(n)}, B^{(n)})P(C_j = a^{\prime} \vert A^{(n)}, B^{(n)}). \notag
\end{align}

Therefore, by central limit theorem, given $A^{(n)}$ and $B^{(n)}$, we have 
\begin{align*}
	& \frac{1-\pi}{n_{[k]1}}\sum_{i \in [k], A_i = 1}\{I_{C_i = 0} - P(C_i = 0 \vert B_i = k)\} \\
     \xrightarrow{d} & N\bigl(0, (1-\pi)^2P(C_i = 0 \vert B_i = k)\{1-P(C_i = 0 \vert B_i = k)\}\bigr),
\end{align*}
\begin{align*}
	& \frac{\pi}{n_{[k]0}}\sum_{i \in [k], A_i = 0}\{I_{C_i = 0} - P(C_i = 0 \vert B_i = k)\} \\
    \xrightarrow{d} & N\bigl(0, \pi^2P(C_i = 0 \vert B_i = k)\{1-P(C_i = 0 \vert B_i = k)\}\bigr).
\end{align*}

Note that these two terms are conditionally independent. Therefore, the asymptotic normality of $n^{-1/2}I_2$ is derived by $n_{[k]1}/n = (n_{[k]1}/n_{[k]})(n_{[k]}/n)\xrightarrow{P}\pi p_{[k]}$, $n_{[k]0}/n = (n_{[k]0}/n_{[k]})(n_{[k]}/n)$
 
    \noindent $\xrightarrow{P}(1-\pi) p_{[k]}$ and Slutsky's theorem.

Note that the asymptotic normality is irrelevant to $A^{(n)}$. Therefore, according to the bounded convergence theorem (BCT), given $B^{(n)}$, we have 
\[
	n^{-1/2}I_2 \xrightarrow{d} N\bigl(0, \pi(1-\pi)p_{[k]}P(C_i = 0 \vert B_i = k)\{1-P(C_i = 0 \vert B_i = k)\}\bigr).
\]

	By Assumption \ref{as3}, given $B^{(n)}$, we have
\[
	n^{-1/2}I_1 \xrightarrow{d} N\bigl(0, p_{[k]}\qk P(C_i = 0 \vert B_i = k)^2\bigr).
\]

	To derive the asymptotic normality of $n^{-1/2}I_1 + n^{-1/2}I_2$ given $B^{(n)}$, we first derive the asymptotic normality of $(n^{-1/2}I_1, n^{-1/2}I_2)$ given $B^{(n)}$. In fact,
\begin{align}
&P\{n^{-1/2}I_1 \leq x, n^{-1/2}I_2 \leq y \vert B^{(n)}\} \notag \\
=& E\{I_{n^{-1/2}I_1 \leq x}I_{n^{-1/2}I_2 \leq y} \vert B^{(n)} \} \notag \\
=& E\bigl[E\{I_{n^{-1/2}I_1 \leq x}I_{n^{-1/2}I_2 \leq y} \vert A^{(n)}\} \vert B^{(n)} \bigr]  \notag \\
=& E\bigl[E\{I_{n^{-1/2}I_1 \leq x}I_{n^{-1/2}I_2 \leq y} \vert A^{(n)}, B^{(n)}\} \vert B^{(n)} \bigr]  \notag \\
=& E\bigl[I_{n^{-1/2}I_1 \leq x}E\{I_{n^{-1/2}I_2 \leq y} \vert A^{(n)}, B^{(n)}\} \vert B^{(n)} \bigr]  \notag \\
=& E\biggl[I_{n^{-1/2}I_1 \leq x}\bigl[E\{I_{n^{-1/2}I_2 \leq y}\vert A^{(n)}, B^{(n)}\} - \phi_2(y)\bigr] \vert B^{(n)} \biggr]  \notag + E\bigl\{I_{n^{-1/2}I_1 \leq x} \vert B^{(n)}\bigr\}\phi_2(y)  \notag \\
\xrightarrow{} & \phi_1(x) \phi_2(y), \notag
\end{align}
\noindent where the final convergence is due to BCT. Therefore, according to the asymptotic independence of $n^{-1/2}I_1$ and $n^{-1/2}I_2$, we have
\[
	n^{-1/2}I_1 + n^{-1/2}I_2 \xrightarrow{d} N(0, p_{[k]}\tilde{q}_{[k]}),
\]
\noindent where $\tilde{q}_{[k]} = \pi(1-\pi)P(C_i = 0 \vert B_i = k)\{1-P(C_i = 0 \vert B_i = k)\} + \qk P(C_i = 0 \vert B_i = k)^2$. Denote the number of fully observed observations as $n^{\prime}$ and the number of fully observed observations within stratum $k$ as $n_{[k]}^{\prime}$. Then, we have
\[
	\frac{n^{\prime}}{n} = \frac{\sum_{k=1}^K n_{[k]}^{\prime}}{n} \xrightarrow{P} \sum_{k=1}^K p_{[k]}P(C_i = 0 \vert B_i = k),
\]
\noindent where the convergence is because of $n_{[k]}^{\prime}/n_{[k]} \xrightarrow{P} P(C_i = 0 \vert B_i = k)$ according to the law of large numbers and $n_{[k]}/n \xrightarrow{P} p_{[k]}$. Therefore, we have 
\[
	(n^{\prime})^{-1/2}I_1 + (n^{\prime})^{-1/2}I_2 \xrightarrow{d} N(0, p_{[k]}{q}_{[k]}^{\prime}).
\]
Therefore, $D_{n[k]}^{\prime} = O_p(\sqrt{n'})$.

\end{proof}

\section{Additional simulation results}
\label{secII}

	This part gives additional simulation results under simple randomization and minimization in \ref{tab5}-\ref{tab12}. All results are similar to those under stratified block randomization.

\newpage

\begin{table}[H]
\tiny
\centering
\caption{Simulation results under simple randomization, $p = 5$, $\pi = 1/2$}
\label{tab5}
\begin{threeparttable}
\begin{tabular}{cccccccccccccccc}
\toprule
\multicolumn{3}{c}{~} & \multicolumn{4}{c}{Model 1} & \multicolumn{4}{c}{Model 2} & \multicolumn{4}{c}{Model 3} \\ \hline
estimator	&	bias	&	SD	&	SE	&	RMSE	&	CP	&	bias	&	SD	&	SE	&	RMSE	&	CP	&	bias	&	SD	&	SE	&	RMSE	&	CP	\\ \midrule
$\hat{\tau}_{B}$	&	-0.03	&	1.29	&	1.27	&	1.29	&	0.95	&	-0.01	&	0.84	&	0.83	&	0.84	&	0.95	&	0.00&	0.45	&	0.45	&	0.45	&	0.95 \\
$\hat{\tau}_{\adj,\ccov}$	&	-0.05	&	1.28	&	1.27	&	1.29	&	0.95	&	-0.01	&	0.79	&	0.79	&	0.79	&	0.95	&	-0.02	&	0.40&	0.40&	0.40&	0.95	\\
$\hat{\tau}_{\int, \ccov}$	&	-0.03	&	1.28	&	1.29	&	1.28	&	0.95	&	-0.01	&	0.79	&	0.80&	0.79	&	0.95	&	-0.01	&	0.40&	0.41	&	0.40&	0.95	\\
$\hat{\tau}_{\tom, \ccov}$	&	-0.05	&	1.28	&	1.29	&	1.29	&	0.95	&	-0.01	&	0.79	&	0.80&	0.79	&	0.95	&	-0.02	&	0.40&	0.41	&	0.40&	0.95	\\
$\hat{\tau}_{\adj, \ccov, \ss}$	&	-0.12	&	1.29	&	1.28	&	1.30&	0.95	&	0.00&	0.80&	0.80&	0.80&	0.95	&	-0.04	&	0.39	&	0.39	&	0.40&	0.94	\\
$\hat{\tau}_{\int, \ccov, \ss}$	&	-0.04	&	1.30&	1.29	&	1.30&	0.95	&	-0.02	&	0.80&	0.80&	0.80&	0.95	&	-0.02	&	0.39	&	0.39	&	0.39	&	0.95	\\
$\hat{\tau}_{\tom, \ccov, \ss}$	&	-0.12	&	1.29	&	1.28	&	1.30&	0.95	&	0.00&	0.80&	0.80&	0.80&	0.95	&	-0.04	&	0.39	&	0.39	&	0.40&	0.94	\\
$\hat{\tau}_{\adj, \imp}$	&	-0.04	&	1.19	&	1.18	&	1.19	&	0.95	&	-0.01	&	0.61	&	0.61	&	0.61	&	0.95	&	-0.02	&	0.39	&	0.39	&	0.39	&	0.95	\\
$\hat{\tau}_{\int, \imp}$	&	-0.02	&	1.20&	1.24	&	1.20&	0.95	&	-0.01	&	0.61	&	0.64	&	0.61	&	0.96	&	-0.01	&	0.39	&	0.41	&	0.39	&	0.96	\\
$\hat{\tau}_{\tom, \imp}$	&	-0.04	&	1.20&	1.24	&	1.20&	0.95	&	-0.01	&	0.61	&	0.64	&	0.61	&	0.96	&	-0.02	&	0.39	&	0.41	&	0.39	&	0.96	\\
$\hat{\tau}_{\adj, \imp, \ss}$	&	-0.10&	1.01	&	1.00&	1.01	&	0.94	&	0.01	&	0.62	&	0.62	&	0.62	&	0.95	&	-0.05	&	0.39	&	0.38	&	0.39	&	0.94	\\
$\hat{\tau}_{\int, \imp, \ss}$	&	-0.01	&	1.02	&	1.01	&	1.02	&	0.95	&	-0.01	&	0.63	&	0.63	&	0.63	&	0.95	&	-0.01	&	0.39	&	0.39	&	0.39	&	0.95	\\
$\hat{\tau}_{\tom, \imp, \ss}$	&	-0.10&	1.01	&	1.00&	1.01	&	0.94	&	0.01	&	0.62	&	0.62	&	0.63	&	0.95	&	-0.05	&	0.39	&	0.38	&	0.39	&	0.94	\\
$\hat{\tau}_{\adj, \mim}$	&	-0.06	&	1.18	&	1.17	&	1.18	&	0.94	&	-0.02	&	0.59	&	0.59	&	0.59	&	0.95	&	-0.04	&	0.38	&	0.37	&	0.38	&	0.94	\\
$\hat{\tau}_{\int, \mim}$	&	-0.02	&	1.19	&	1.28	&	1.19	&	0.96	&	-0.01	&	0.59	&	0.64	&	0.59	&	0.97	&	-0.01	&	0.38	&	0.41	&	0.38	&	0.96	\\
$\hat{\tau}_{\tom, \mim}$	&	-0.06	&	1.19	&	1.28	&	1.19	&	0.96	&	-0.02	&	0.59	&	0.64	&	0.59	&	0.97	&	-0.04	&	0.38	&	0.41	&	0.38	&	0.96	\\
$\hat{\tau}_{\adj, \mim, \ss}$	&	-0.18	&	1.00&	0.99	&	1.02	&	0.94	&	-0.08	&	0.61	&	0.61	&	0.62	&	0.95	&	-0.14	&	0.38	&	0.37	&	0.41	&	0.92	\\
$\hat{\tau}_{\int, \mim, \ss}$	&	-0.04	&	1.04	&	1.02	&	1.04	&	0.94	&	-0.03	&	0.63	&	0.62	&	0.63	&	0.95	&	-0.04	&	0.39	&	0.39	&	0.39	&	0.95	\\
$\hat{\tau}_{\tom, \mim, \ss}$	&	-0.18	&	1.00&	1.00&	1.02	&	0.94	&	-0.08	&	0.61	&	0.61	&	0.62	&	0.95	&	-0.14	&	0.38	&	0.37	&	0.41	&	0.92	\\ \bottomrule
\end{tabular}

\end{threeparttable}
\end{table}

\begin{table}[H]
\tiny
\centering
\caption{Simulation results under simple randomization, $p = 5$, $\pi = 2/3$}
\label{tab6}
\begin{threeparttable}
\begin{tabular}{cccccccccccccccc}
\toprule
\multicolumn{3}{c}{~} & \multicolumn{4}{c}{Model 1} & \multicolumn{4}{c}{Model 2} & \multicolumn{4}{c}{Model 3} \\ \hline
estimator	&	bias	&	SD	&	SE	&	RMSE	&	CP	&	bias	&	SD	&	SE	&	RMSE	&	CP	&	bias	&	SD	&	SE	&	RMSE	&	CP	\\ \midrule
$\hat{\tau}_{B}$	&	0.04	&	1.33	&	1.30&	1.33	&	0.95	&	0.02	&	0.88	&	0.88	&	0.88	&	0.95	&	0.01	&	0.43	&	0.42	&	0.43	&	0.95	\\
$\hat{\tau}_{\adj,\ccov}$	&	0.02	&	1.33	&	1.39	&	1.33	&	0.96	&	0.02	&	0.84	&	0.84	&	0.84	&	0.95	&	0.00&	0.39	&	0.47	&	0.39	&	0.98	\\
$\hat{\tau}_{\int, \ccov}$	&	0.04	&	1.33	&	1.33	&	1.33	&	0.95	&	0.02	&	0.84	&	0.85	&	0.84	&	0.95	&	0.00&	0.38	&	0.39	&	0.38	&	0.95	\\
$\hat{\tau}_{\tom, \ccov}$	&	0.02	&	1.33	&	1.33	&	1.33	&	0.95	&	0.03	&	0.84	&	0.85	&	0.84	&	0.95	&	0.00&	0.38	&	0.39	&	0.38	&	0.95	\\
$\hat{\tau}_{\adj, \ccov, \ss}$	&	-0.05	&	1.35	&	1.41	&	1.35	&	0.96	&	0.03	&	0.85	&	0.85	&	0.85	&	0.95	&	-0.02	&	0.38	&	0.45	&	0.38	&	0.98	\\
$\hat{\tau}_{\int, \ccov, \ss}$	&	0.04	&	1.36	&	1.32	&	1.36	&	0.95	&	0.02	&	0.86	&	0.85	&	0.86	&	0.94	&	0.00&	0.38	&	0.37	&	0.38	&	0.94	\\
$\hat{\tau}_{\tom, \ccov, \ss}$	&	-0.05	&	1.35	&	1.32	&	1.35	&	0.95	&	0.04	&	0.85	&	0.84	&	0.85	&	0.95	&	-0.03	&	0.37	&	0.37	&	0.38	&	0.94	\\
$\hat{\tau}_{\adj, \imp}$	&	0.01	&	1.25	&	1.27	&	1.25	&	0.95	&	0.01	&	0.65	&	0.63	&	0.65	&	0.94	&	-0.01	&	0.38	&	0.45	&	0.38	&	0.98	\\
$\hat{\tau}_{\int, \imp}$	&	0.03	&	1.25	&	1.31	&	1.25	&	0.96	&	0.01	&	0.66	&	0.70&	0.66	&	0.96	&	0.00&	0.37	&	0.39	&	0.37	&	0.96	\\
$\hat{\tau}_{\tom, \imp}$	&	0.01	&	1.25	&	1.30&	1.25	&	0.96	&	0.02	&	0.66	&	0.70&	0.66	&	0.96	&	-0.01	&	0.37	&	0.39	&	0.37	&	0.96	\\
$\hat{\tau}_{\adj, \imp, \ss}$	&	-0.06	&	1.07	&	1.07	&	1.07	&	0.95	&	0.03	&	0.67	&	0.65	&	0.67	&	0.94	&	-0.03	&	0.38	&	0.44	&	0.38	&	0.97	\\
$\hat{\tau}_{\int, \imp, \ss}$	&	0.01	&	1.11	&	1.07	&	1.11	&	0.94	&	0.01	&	0.71	&	0.68	&	0.71	&	0.94	&	-0.01	&	0.38	&	0.37	&	0.38	&	0.94	\\
$\hat{\tau}_{\tom, \imp, \ss}$	&	-0.08	&	1.08	&	1.05	&	1.08	&	0.94	&	0.02	&	0.69	&	0.67	&	0.69	&	0.94	&	-0.04	&	0.37	&	0.37	&	0.37	&	0.94	\\
$\hat{\tau}_{\adj, \mim}$	&	-0.01	&	1.24	&	1.25	&	1.24	&	0.95	&	-0.01	&	0.64	&	0.60&	0.64	&	0.93	&	-0.03	&	0.39	&	0.42	&	0.39	&	0.97	\\
$\hat{\tau}_{\int, \mim}$	&	0.03	&	1.25	&	1.39	&	1.26	&	0.97	&	0.01	&	0.65	&	0.73	&	0.65	&	0.97	&	0.00&	0.37	&	0.41	&	0.37	&	0.97	\\
$\hat{\tau}_{\tom, \mim}$	&	-0.01	&	1.25	&	1.38	&	1.25	&	0.97	&	-0.01	&	0.65	&	0.73	&	0.65	&	0.97	&	-0.03	&	0.37	&	0.40&	0.37	&	0.97	\\
$\hat{\tau}_{\adj, \mim, \ss}$	&	-0.15	&	1.07	&	1.05	&	1.08	&	0.94	&	-0.06	&	0.67	&	0.62	&	0.67	&	0.93	&	-0.13	&	0.39	&	0.41	&	0.41	&	0.94	\\
$\hat{\tau}_{\int, \mim, \ss}$	&	-0.03	&	1.16	&	1.09	&	1.16	&	0.93	&	-0.03	&	0.75	&	0.69	&	0.75	&	0.93	&	-0.06	&	0.40&	0.38	&	0.40&	0.94	\\
$\hat{\tau}_{\tom, \mim, \ss}$	&	-0.18	&	1.09	&	1.06	&	1.10&	0.94	&	-0.08	&	0.69	&	0.67	&	0.70&	0.94	&	-0.15	&	0.38	&	0.37	&	0.41	&	0.91	\\ \bottomrule
\end{tabular}

\end{threeparttable}
\end{table}

\begin{table}[H]
\tiny
\centering
\caption{Simulation results under simple randomization, $p = 7$, $\pi = 1/2$}
\label{tab7}
\begin{threeparttable}
\begin{tabular}{cccccccccccccccc}
\toprule
\multicolumn{3}{c}{~} & \multicolumn{4}{c}{Model 1} & \multicolumn{4}{c}{Model 2} & \multicolumn{4}{c}{Model 3} \\ \hline
estimator	&	bias	&	SD	&	SE	&	RMSE	&	CP	&	bias	&	SD	&	SE	&	RMSE	&	CP	&	bias	&	SD	&	SE	&	RMSE	&	CP	\\ \midrule
$\hat{\tau}_{B}$	&	0.02	&	1.29	&	1.27	&	1.29	&	0.95	&	0.02	&	0.84	&	0.83	&	0.84	&	0.95	&	0.00&	0.45	&	0.44	&	0.45	&	0.95	\\
$\hat{\tau}_{\adj,\ccov}$	&	-0.01	&	1.28	&	1.27	&	1.28	&	0.95	&	0.02	&	0.79	&	0.79	&	0.79	&	0.95	&	-0.01	&	0.41	&	0.40&	0.41	&	0.95	\\
$\hat{\tau}_{\int, \ccov}$	&	0.02	&	1.28	&	1.29	&	1.28	&	0.95	&	0.01	&	0.79	&	0.80&	0.79	&	0.95	&	0.00&	0.41	&	0.41	&	0.41	&	0.95	\\
$\hat{\tau}_{\tom, \ccov}$	&	-0.01	&	1.28	&	1.29	&	1.28	&	0.95	&	0.02	&	0.79	&	0.80&	0.79	&	0.95	&	-0.01	&	0.41	&	0.41	&	0.41	&	0.95	\\
$\hat{\tau}_{\adj, \ccov, \ss}$	&	-0.07	&	1.29	&	1.28	&	1.30&	0.94	&	0.03	&	0.80&	0.80&	0.80&	0.95	&	-0.03	&	0.40&	0.38	&	0.40&	0.94	\\
$\hat{\tau}_{\int, \ccov, \ss}$	&	0.01	&	1.30&	1.28	&	1.30&	0.95	&	0.01	&	0.80&	0.80&	0.80&	0.95	&	-0.01	&	0.40&	0.39	&	0.40&	0.95	\\
$\hat{\tau}_{\tom, \ccov, \ss}$	&	-0.07	&	1.29	&	1.28	&	1.30&	0.95	&	0.03	&	0.80&	0.79	&	0.80&	0.95	&	-0.03	&	0.40&	0.39	&	0.40&	0.94	\\
$\hat{\tau}_{\adj, \imp}$	&	-0.01	&	1.20&	1.18	&	1.20&	0.95	&	0.02	&	0.61	&	0.61	&	0.61	&	0.95	&	0.00&	0.40&	0.39	&	0.40&	0.94	\\
$\hat{\tau}_{\int, \imp}$	&	0.02	&	1.20&	1.30&	1.20&	0.97	&	0.02	&	0.61	&	0.67	&	0.61	&	0.97	&	0.01	&	0.40&	0.43	&	0.40&	0.96	\\
$\hat{\tau}_{\tom, \imp}$	&	0.00&	1.20&	1.30&	1.20&	0.97	&	0.02	&	0.61	&	0.67	&	0.61	&	0.97	&	0.00&	0.40&	0.43	&	0.40&	0.96	\\
$\hat{\tau}_{\adj, \imp, \ss}$	&	-0.03	&	1.03	&	1.02	&	1.03	&	0.95	&	0.06	&	0.64	&	0.64	&	0.64	&	0.95	&	-0.01	&	0.40&	0.39	&	0.40&	0.94	\\
$\hat{\tau}_{\int, \imp, \ss}$	&	0.06	&	1.05	&	1.05	&	1.05	&	0.95	&	0.04	&	0.65	&	0.65	&	0.66	&	0.95	&	0.03	&	0.41	&	0.40&	0.41	&	0.95	\\
$\hat{\tau}_{\tom, \imp, \ss}$	&	-0.03	&	1.03	&	1.03	&	1.03	&	0.95	&	0.06	&	0.64	&	0.64	&	0.65	&	0.95	&	-0.01	&	0.40&	0.40&	0.40&	0.94	\\
$\hat{\tau}_{\adj, \mim}$	&	-0.05	&	1.19	&	1.18	&	1.20&	0.95	&	-0.02	&	0.59	&	0.59	&	0.59	&	0.95	&	-0.04	&	0.38	&	0.37	&	0.38	&	0.94	\\
$\hat{\tau}_{\int, \mim}$	&	0.01	&	1.20&	1.45	&	1.20&	0.98	&	0.01	&	0.59	&	0.72	&	0.59	&	0.98	&	0.00&	0.38	&	0.46	&	0.38	&	0.98	\\
$\hat{\tau}_{\tom, \mim}$	&	-0.05	&	1.20&	1.43	&	1.20&	0.98	&	-0.02	&	0.59	&	0.71	&	0.59	&	0.98	&	-0.05	&	0.38	&	0.45	&	0.38	&	0.98	\\
$\hat{\tau}_{\adj, \mim, \ss}$	&	-0.21	&	1.05	&	1.03	&	1.07	&	0.93	&	-0.12	&	0.64	&	0.63	&	0.65	&	0.94	&	-0.20&	0.40&	0.38	&	0.44	&	0.89	\\
$\hat{\tau}_{\int, \mim, \ss}$	&	0.00&	1.13	&	1.18	&	1.13	&	0.96	&	-0.01	&	0.69	&	0.72	&	0.69	&	0.96	&	-0.03	&	0.42	&	0.44	&	0.42	&	0.96	\\
$\hat{\tau}_{\tom, \mim, \ss}$	&	-0.21	&	1.06	&	1.06	&	1.08	&	0.94	&	-0.12	&	0.64	&	0.65	&	0.66	&	0.94	&	-0.20&	0.40&	0.39	&	0.45	&	0.91	\\ \bottomrule
\end{tabular}

\end{threeparttable}
\end{table}

\begin{table}[H]
\tiny
\centering
\caption{Simulation results under simple randomization, $p = 7$, $\pi = 2/3$}
\label{tab8}
\begin{threeparttable}
\begin{tabular}{cccccccccccccccc}
\toprule
\multicolumn{3}{c}{~} & \multicolumn{4}{c}{Model 1} & \multicolumn{4}{c}{Model 2} & \multicolumn{4}{c}{Model 3} \\ \hline
estimator	&	bias	&	SD	&	SE	&	RMSE	&	CP	&	bias	&	SD	&	SE	&	RMSE	&	CP	&	bias	&	SD	&	SE	&	RMSE	&	CP	\\ \midrule
$\hat{\tau}_{B}$	&	0.03	&	1.30&	1.30&	1.30&	0.95	&	0.03	&	0.89	&	0.88	&	0.89	&	0.94	&	0.01	&	0.43	&	0.42	&	0.43	&	0.94	\\
$\hat{\tau}_{\adj,\ccov}$	&	0.01	&	1.30&	1.38	&	1.30&	0.96	&	0.03	&	0.85	&	0.83	&	0.85	&	0.94	&	0.00&	0.39	&	0.46	&	0.39	&	0.98	\\
$\hat{\tau}_{\int, \ccov}$	&	0.03	&	1.30&	1.33	&	1.30&	0.95	&	0.03	&	0.85	&	0.85	&	0.85	&	0.95	&	0.01	&	0.39	&	0.39	&	0.39	&	0.95	\\
$\hat{\tau}_{\tom, \ccov}$	&	0.01	&	1.30&	1.33	&	1.30&	0.95	&	0.03	&	0.85	&	0.85	&	0.85	&	0.95	&	0.00&	0.39	&	0.39	&	0.39	&	0.95	\\
$\hat{\tau}_{\adj, \ccov, \ss}$	&	-0.05	&	1.31	&	1.39	&	1.31	&	0.96	&	0.04	&	0.85	&	0.84	&	0.86	&	0.94	&	-0.02	&	0.38	&	0.45	&	0.38	&	0.97	\\
$\hat{\tau}_{\int, \ccov, \ss}$	&	0.03	&	1.33	&	1.32	&	1.33	&	0.95	&	0.03	&	0.87	&	0.85	&	0.87	&	0.94	&	0.00&	0.38	&	0.37	&	0.38	&	0.94	\\
$\hat{\tau}_{\tom, \ccov, \ss}$	&	-0.05	&	1.31	&	1.32	&	1.32	&	0.95	&	0.04	&	0.86	&	0.84	&	0.86	&	0.94	&	-0.02	&	0.38	&	0.37	&	0.38	&	0.94	\\
$\hat{\tau}_{\adj, \imp}$	&	0.01	&	1.22	&	1.27	&	1.22	&	0.96	&	0.03	&	0.66	&	0.63	&	0.66	&	0.94	&	0.00&	0.39	&	0.45	&	0.39	&	0.97	\\
$\hat{\tau}_{\int, \imp}$	&	0.03	&	1.22	&	1.39	&	1.22	&	0.97	&	0.02	&	0.66	&	0.74	&	0.66	&	0.97	&	0.01	&	0.38	&	0.41	&	0.38	&	0.96	\\
$\hat{\tau}_{\tom, \imp}$	&	0.00&	1.22	&	1.38	&	1.22	&	0.97	&	0.02	&	0.66	&	0.74	&	0.66	&	0.97	&	0.00&	0.38	&	0.41	&	0.38	&	0.96	\\
$\hat{\tau}_{\adj, \imp, \ss}$	&	-0.03	&	1.08	&	1.09	&	1.08	&	0.95	&	0.06	&	0.69	&	0.67	&	0.69	&	0.94	&	-0.01	&	0.40&	0.44	&	0.40&	0.97	\\
$\hat{\tau}_{\int, \imp, \ss}$	&	0.02	&	1.15	&	1.11	&	1.15	&	0.94	&	0.02	&	0.75	&	0.71	&	0.75	&	0.93	&	0.01	&	0.40&	0.39	&	0.40&	0.94	\\
$\hat{\tau}_{\tom, \imp, \ss}$	&	-0.07	&	1.10&	1.09	&	1.10&	0.95	&	0.04	&	0.71	&	0.70&	0.71	&	0.94	&	-0.03	&	0.39	&	0.38	&	0.39	&	0.94	\\
$\hat{\tau}_{\adj, \mim}$	&	-0.06	&	1.23	&	1.25	&	1.23	&	0.95	&	-0.04	&	0.65	&	0.59	&	0.65	&	0.93	&	-0.06	&	0.39	&	0.41	&	0.39	&	0.96	\\
$\hat{\tau}_{\int, \mim}$	&	0.00&	1.24	&	1.70&	1.24	&	0.99	&	0.00&	0.66	&	0.91	&	0.66	&	0.99	&	-0.01	&	0.37	&	0.48	&	0.37	&	0.99	\\
$\hat{\tau}_{\tom, \mim}$	&	-0.06	&	1.23	&	1.68	&	1.23	&	0.99	&	-0.04	&	0.66	&	0.90&	0.66	&	0.99	&	-0.06	&	0.37	&	0.47	&	0.38	&	0.98	\\
$\hat{\tau}_{\adj, \mim, \ss}$	&	-0.24	&	1.11	&	1.08	&	1.14	&	0.93	&	-0.15	&	0.69	&	0.64	&	0.71	&	0.92	&	-0.21	&	0.41	&	0.41	&	0.46	&	0.91	\\
$\hat{\tau}_{\int, \mim, \ss}$	&	-0.07	&	1.36	&	1.31	&	1.37	&	0.95	&	-0.06	&	0.90&	0.84	&	0.90&	0.94	&	-0.08	&	0.45	&	0.45	&	0.46	&	0.95	\\
$\hat{\tau}_{\tom, \mim, \ss}$	&	-0.29	&	1.14	&	1.19	&	1.17	&	0.95	&	-0.18	&	0.73	&	0.76	&	0.75	&	0.95	&	-0.25	&	0.41	&	0.41	&	0.48	&	0.90\\ \bottomrule
\end{tabular}

\end{threeparttable}
\end{table}

\begin{table}[H]
\tiny
\centering
\caption{Simulation results under minimization, $p = 5$, $\pi = 1/2$}
\label{tab9}
\begin{threeparttable}
\begin{tabular}{cccccccccccccccc}
\toprule
\multicolumn{3}{c}{~} & \multicolumn{4}{c}{Model 1} & \multicolumn{4}{c}{Model 2} & \multicolumn{4}{c}{Model 3} \\ \hline
estimator	&	bias	&	SD	&	SE	&	RMSE	&	CP	&	bias	&	SD	&	SE	&	RMSE	&	CP	&	bias	&	SD	&	SE	&	RMSE	&	CP	\\ \midrule
$\hat{\tau}_{B}$	&	0.03	&	1.27	&	1.27	&	1.27	&	0.95	&	0.01	&	0.84	&	0.83	&	0.84	&	0.95	&	0.01	&	0.44	&	0.44	&	0.44	&	0.95	\\
$\hat{\tau}_{\adj,\ccov}$	&	0.00&	1.27	&	1.26	&	1.27	&	0.95	&	0.02	&	0.79	&	0.78	&	0.79	&	0.95	&	0.00&	0.40&	0.40&	0.40&	0.95	\\
$\hat{\tau}_{\int, \ccov}$	&	0.02	&	1.27	&	1.29	&	1.27	&	0.95	&	0.01	&	0.79	&	0.80&	0.79	&	0.95	&	0.01	&	0.40&	0.41	&	0.40&	0.95	\\
$\hat{\tau}_{\tom, \ccov}$	&	0.00&	1.27	&	1.29	&	1.27	&	0.95	&	0.02	&	0.79	&	0.80&	0.79	&	0.95	&	0.00&	0.40&	0.41	&	0.40&	0.95	\\
$\hat{\tau}_{\adj, \ccov, \ss}$	&	-0.06	&	1.27	&	1.27	&	1.28	&	0.94	&	0.02	&	0.80&	0.79	&	0.80&	0.95	&	-0.02	&	0.39	&	0.38	&	0.39	&	0.94	\\
$\hat{\tau}_{\int, \ccov, \ss}$	&	0.02	&	1.28	&	1.28	&	1.28	&	0.95	&	0.01	&	0.80&	0.80&	0.80&	0.95	&	0.00&	0.39	&	0.39	&	0.39	&	0.95	\\
$\hat{\tau}_{\tom, \ccov, \ss}$	&	-0.06	&	1.27	&	1.28	&	1.28	&	0.95	&	0.02	&	0.80&	0.79	&	0.80&	0.95	&	-0.02	&	0.39	&	0.39	&	0.39	&	0.95	\\
$\hat{\tau}_{\adj, \imp}$	&	0.00&	1.18	&	1.17	&	1.18	&	0.95	&	0.02	&	0.61	&	0.60&	0.61	&	0.95	&	0.00&	0.39	&	0.39	&	0.39	&	0.95	\\
$\hat{\tau}_{\int, \imp}$	&	0.03	&	1.18	&	1.24	&	1.18	&	0.96	&	0.01	&	0.61	&	0.64	&	0.61	&	0.96	&	0.01	&	0.39	&	0.41	&	0.39	&	0.96	\\
$\hat{\tau}_{\tom, \imp}$	&	0.00&	1.18	&	1.24	&	1.18	&	0.96	&	0.02	&	0.61	&	0.64	&	0.61	&	0.96	&	0.00&	0.39	&	0.41	&	0.39	&	0.96	\\
$\hat{\tau}_{\adj, \imp, \ss}$	&	-0.06	&	1.00&	0.99	&	1.01	&	0.95	&	0.03	&	0.62	&	0.95	&	0.62	&	0.95	&	-0.03	&	0.38	&	0.37	&	0.38	&	0.94	\\
$\hat{\tau}_{\int, \imp, \ss}$	&	0.03	&	1.01	&	1.01	&	1.01	&	0.95	&	0.01	&	0.62	&	0.63	&	0.62	&	0.95	&	0.01	&	0.38	&	0.39	&	0.38	&	0.95	\\
$\hat{\tau}_{\tom, \imp, \ss}$	&	-0.06	&	1.00&	1.00&	1.01	&	0.95	&	0.03	&	0.62	&	0.62	&	0.62	&	0.95	&	-0.03	&	0.38	&	0.38	&	0.38	&	0.95	\\
$\hat{\tau}_{\adj, \mim}$	&	-0.01	&	1.16	&	1.16	&	1.16	&	0.95	&	0.00&	0.59	&	0.58	&	0.59	&	0.95	&	-0.02	&	0.38	&	0.37	&	0.38	&	0.94	\\
$\hat{\tau}_{\int, \mim}$	&	0.03	&	1.17	&	1.28	&	1.17	&	0.97	&	0.01	&	0.59	&	0.64	&	0.59	&	0.97	&	0.00&	0.38	&	0.41	&	0.38	&	0.97	\\
$\hat{\tau}_{\tom, \mim}$	&	-0.01	&	1.17	&	1.27	&	1.17	&	0.97	&	0.00&	0.59	&	0.64	&	0.59	&	0.97	&	-0.02	&	0.38	&	0.41	&	0.38	&	0.97	\\
$\hat{\tau}_{\adj, \mim, \ss}$	&	-0.14	&	0.99	&	0.98	&	1.00&	0.94	&	-0.06	&	0.61	&	0.60&	0.61	&	0.95	&	-0.12	&	0.38	&	0.36	&	0.40&	0.92	\\
$\hat{\tau}_{\int, \mim, \ss}$	&	0.00&	1.02	&	1.02	&	1.02	&	0.95	&	-0.01	&	0.62	&	0.62	&	0.62	&	0.95	&	-0.02	&	0.38	&	0.38	&	0.38	&	0.95	\\
$\hat{\tau}_{\tom, \mim, \ss}$	&	-0.14	&	0.99	&	0.99	&	1.00&	0.95	&	-0.06	&	0.61	&	0.61	&	0.61	&	0.95	&	-0.12	&	0.38	&	0.37	&	0.40&	0.93	\\ \bottomrule
\end{tabular}

\end{threeparttable}
\end{table}

\begin{table}[H]
\tiny
\centering
\caption{Simulation results under minimization, $p = 5$, $\pi = 2/3$}
\label{tab10}
\begin{threeparttable}
\begin{tabular}{cccccccccccccccc}
\toprule
\multicolumn{3}{c}{~} & \multicolumn{4}{c}{Model 1} & \multicolumn{4}{c}{Model 2} & \multicolumn{4}{c}{Model 3} \\ \hline
estimator	&	bias	&	SD	&	SE	&	RMSE	&	CP	&	bias	&	SD	&	SE	&	RMSE	&	CP	&	bias	&	SD	&	SE	&	RMSE	&	CP	\\ \midrule
$\hat{\tau}_{B}$	&	0.01	&	1.30&	1.31	&	1.30&	0.95	&	0.00&	0.89	&	0.88	&	0.89	&	0.94	&	0.00&	0.42	&	0.43	&	0.42	&	0.95	\\
$\hat{\tau}_{\adj,\ccov}$	&	-0.02	&	1.30&	1.39	&	1.30&	0.96	&	0.01	&	0.85	&	0.84	&	0.85	&	0.94	&	-0.01	&	0.38	&	0.47	&	0.38	&	0.98	\\
$\hat{\tau}_{\int, \ccov}$	&	0.01	&	1.30&	1.33	&	1.30&	0.96	&	0.00&	0.85	&	0.85	&	0.85	&	0.95	&	0.00&	0.37	&	0.39	&	0.37	&	0.95	\\
$\hat{\tau}_{\tom, \ccov}$	&	-0.02	&	1.30&	1.33	&	1.30&	0.96	&	0.01	&	0.85	&	0.85	&	0.85	&	0.95	&	-0.01	&	0.37	&	0.39	&	0.37	&	0.95	\\
$\hat{\tau}_{\adj, \ccov, \ss}$	&	-0.08	&	1.32	&	1.41	&	1.32	&	0.96	&	0.01	&	0.86	&	0.84	&	0.86	&	0.94	&	-0.03	&	0.37	&	0.45	&	0.37	&	0.98	\\
$\hat{\tau}_{\int, \ccov, \ss}$	&	0.01	&	1.32	&	1.32	&	1.32	&	0.95	&	0.01	&	0.87	&	0.85	&	0.87	&	0.94	&	-0.01	&	0.37	&	0.37	&	0.37	&	0.95	\\
$\hat{\tau}_{\tom, \ccov, \ss}$	&	-0.08	&	1.31	&	1.32	&	1.32	&	0.95	&	0.02	&	0.86	&	0.84	&	0.86	&	0.94	&	-0.03	&	0.37	&	0.37	&	0.37	&	0.95	\\
$\hat{\tau}_{\adj, \imp}$	&	-0.01	&	1.23	&	1.28	&	1.23	&	0.96	&	0.01	&	0.65	&	0.64	&	0.65	&	0.94	&	-0.01	&	0.37	&	0.45	&	0.37	&	0.98	\\
$\hat{\tau}_{\int, \imp}$	&	0.01	&	1.22	&	1.31	&	1.22	&	0.96	&	0.01	&	0.66	&	0.70&	0.66	&	0.96	&	0.00&	0.37	&	0.39	&	0.37	&	0.96	\\
$\hat{\tau}_{\tom, \imp}$	&	-0.02	&	1.22	&	1.31	&	1.22	&	0.96	&	0.01	&	0.66	&	0.70&	0.66	&	0.96	&	-0.01	&	0.36	&	0.39	&	0.37	&	0.96	\\
$\hat{\tau}_{\adj, \imp, \ss}$	&	-0.07	&	1.07	&	1.07	&	1.07	&	0.95	&	0.03	&	0.67	&	0.65	&	0.67	&	0.94	&	-0.04	&	0.37	&	0.44	&	0.38	&	0.97	\\
$\hat{\tau}_{\int, \imp, \ss}$	&	0.01	&	1.09	&	1.07	&	1.09	&	0.94	&	0.01	&	0.70&	0.68	&	0.70&	0.94	&	-0.01	&	0.37	&	0.37	&	0.37	&	0.95	\\
$\hat{\tau}_{\tom, \imp, \ss}$	&	-0.08	&	1.07	&	1.05	&	1.07	&	0.94	&	0.03	&	0.69	&	0.67	&	0.69	&	0.94	&	-0.05	&	0.36	&	0.37	&	0.37	&	0.95	\\
$\hat{\tau}_{\adj, \mim}$	&	-0.04	&	1.23	&	1.26	&	1.23	&	0.95	&	-0.01	&	0.64	&	0.60&	0.64	&	0.93	&	-0.03	&	0.38	&	0.42	&	0.38	&	0.97	\\
$\hat{\tau}_{\int, \mim}$	&	0.01	&	1.22	&	1.38	&	1.22	&	0.97	&	0.01	&	0.65	&	0.73	&	0.65	&	0.97	&	0.00&	0.36	&	0.40&	0.36	&	0.97	\\
$\hat{\tau}_{\tom, \mim}$	&	-0.04	&	1.22	&	1.38	&	1.22	&	0.97	&	-0.01	&	0.65	&	0.73	&	0.65	&	0.97	&	-0.03	&	0.36	&	0.40&	0.36	&	0.97	\\
$\hat{\tau}_{\adj, \mim, \ss}$	&	-0.16	&	1.07	&	1.05	&	1.08	&	0.94	&	-0.06	&	0.67	&	0.62	&	0.67	&	0.93	&	-0.13	&	0.39	&	0.41	&	0.41	&	0.95	\\
$\hat{\tau}_{\int, \mim, \ss}$	&	-0.05	&	1.14	&	1.09	&	1.14	&	0.94	&	-0.04	&	0.73	&	0.69	&	0.73	&	0.93	&	-0.06	&	0.38	&	0.38	&	0.39	&	0.94	\\
$\hat{\tau}_{\tom, \mim, \ss}$	&	-0.19	&	1.08	&	1.06	&	1.09	&	0.94	&	-0.08	&	0.69	&	0.67	&	0.69	&	0.94	&	-0.16	&	0.37	&	0.37	&	0.40&	0.92	\\ \bottomrule
\end{tabular}

\end{threeparttable}
\end{table}

\begin{table}[H]
\tiny
\centering
\caption{Simulation results under minimization, $p = 7$, $\pi = 1/2$}
\label{tab11}
\begin{threeparttable}
\begin{tabular}{cccccccccccccccc}
\toprule
\multicolumn{3}{c}{~} & \multicolumn{4}{c}{Model 1} & \multicolumn{4}{c}{Model 2} & \multicolumn{4}{c}{Model 3} \\ \hline
estimator	&	bias	&	SD	&	SE	&	RMSE	&	CP	&	bias	&	SD	&	SE	&	RMSE	&	CP	&	bias	&	SD	&	SE	&	RMSE	&	CP	\\ \midrule
$\hat{\tau}_{B}$	&	0.01	&	1.26	&	1.27	&	1.26	&	0.95	&	-0.01	&	0.83	&	0.83	&	0.83	&	0.95	&	0.00&	0.44	&	0.44	&	0.44	&	0.95	\\
$\hat{\tau}_{\adj,\ccov}$	&	-0.02	&	1.26	&	1.26	&	1.26	&	0.95	&	0.00&	0.78	&	0.78	&	0.78	&	0.95	&	-0.01	&	0.40&	0.40&	0.40&	0.95	\\
$\hat{\tau}_{\int, \ccov}$	&	0.00&	1.26	&	1.29	&	1.26	&	0.95	&	0.00&	0.78	&	0.80&	0.78	&	0.96	&	0.00&	0.40&	0.41	&	0.40&	0.95	\\
$\hat{\tau}_{\tom, \ccov}$	&	-0.02	&	1.26	&	1.29	&	1.26	&	0.95	&	0.00&	0.78	&	0.80&	0.78	&	0.96	&	-0.01	&	0.40&	0.41	&	0.40&	0.95	\\
$\hat{\tau}_{\adj, \ccov, \ss}$	&	-0.08	&	1.27	&	1.27	&	1.28	&	0.95	&	0.01	&	0.79	&	0.79	&	0.79	&	0.95	&	-0.03	&	0.39	&	0.38	&	0.39	&	0.94	\\
$\hat{\tau}_{\int, \ccov, \ss}$	&	0.00&	1.28	&	1.28	&	1.28	&	0.95	&	-0.01	&	0.79	&	0.80&	0.79	&	0.95	&	-0.01	&	0.39	&	0.39	&	0.39	&	0.95	\\
$\hat{\tau}_{\tom, \ccov, \ss}$	&	-0.09	&	1.27	&	1.28	&	1.27	&	0.95	&	0.01	&	0.79	&	0.79	&	0.79	&	0.95	&	-0.03	&	0.39	&	0.39	&	0.39	&	0.94	\\
$\hat{\tau}_{\adj, \imp}$	&	-0.01	&	1.18	&	1.17	&	1.18	&	0.95	&	0.01	&	0.61	&	0.60&	0.61	&	0.95	&	0.00&	0.39	&	0.39	&	0.39	&	0.95	\\
$\hat{\tau}_{\int, \imp}$	&	0.01	&	1.18	&	1.30&	1.18	&	0.97	&	0.01	&	0.61	&	0.67	&	0.61	&	0.97	&	0.01	&	0.39	&	0.43	&	0.39	&	0.96	\\
$\hat{\tau}_{\tom, \imp}$	&	-0.01	&	1.18	&	1.29	&	1.18	&	0.97	&	0.01	&	0.61	&	0.67	&	0.62	&	0.97	&	0.00&	0.39	&	0.43	&	0.39	&	0.96	\\
$\hat{\tau}_{\adj, \imp, \ss}$	&	-0.05	&	1.03	&	1.01	&	1.03	&	0.95	&	0.04	&	0.64	&	0.63	&	0.65	&	0.94	&	-0.02	&	0.40&	0.38	&	0.40&	0.94	\\
$\hat{\tau}_{\int, \imp, \ss}$	&	0.04	&	1.04	&	1.05	&	1.04	&	0.95	&	0.03	&	0.65	&	0.65	&	0.65	&	0.95	&	0.02	&	0.40&	0.40&	0.40&	0.95	\\
$\hat{\tau}_{\tom, \imp, \ss}$	&	-0.05	&	1.03	&	1.03	&	1.03	&	0.95	&	0.04	&	0.64	&	0.64	&	0.65	&	0.95	&	-0.02	&	0.40&	0.40&	0.40&	0.94	\\
$\hat{\tau}_{\adj, \mim}$	&	-0.05	&	1.17	&	1.17	&	1.17	&	0.95	&	-0.03	&	0.59	&	0.58	&	0.59	&	0.95	&	-0.05	&	0.38	&	0.37	&	0.38	&	0.94	\\
$\hat{\tau}_{\int, \mim}$	&	0.00&	1.18	&	1.43	&	1.18	&	0.98	&	0.00&	0.59	&	0.71	&	0.59	&	0.98	&	0.00&	0.37	&	0.46	&	0.37	&	0.98	\\
$\hat{\tau}_{\tom, \mim}$	&	-0.05	&	1.17	&	1.42	&	1.17	&	0.98	&	-0.03	&	0.59	&	0.71	&	0.59	&	0.98	&	-0.05	&	0.38	&	0.45	&	0.38	&	0.98	\\
$\hat{\tau}_{\adj, \mim, \ss}$	&	-0.23	&	1.04	&	1.02	&	1.07	&	0.94	&	-0.13	&	0.64	&	0.62	&	0.65	&	0.94	&	-0.20&	0.39	&	0.37	&	0.44	&	0.90\\
$\hat{\tau}_{\int, \mim, \ss}$	&	-0.02	&	1.10&	1.16	&	1.10&	0.96	&	-0.02	&	0.68	&	0.71	&	0.68	&	0.96	&	-0.03	&	0.41	&	0.43	&	0.41	&	0.96	\\
$\hat{\tau}_{\tom, \mim, \ss}$	&	-0.23	&	1.04	&	1.06	&	1.07	&	0.95	&	-0.14	&	0.64	&	0.65	&	0.65	&	0.95	&	-0.20&	0.39	&	0.39	&	0.44	&	0.91	\\ \bottomrule
\end{tabular}

\end{threeparttable}
\end{table}

\begin{table}[H]
\tiny
\centering
\caption{Simulation results under minimization, $p = 7$, $\pi = 2/3$}
\label{tab12}
\begin{threeparttable}
\begin{tabular}{cccccccccccccccc}
\toprule
\multicolumn{3}{c}{~} & \multicolumn{4}{c}{Model 1} & \multicolumn{4}{c}{Model 2} & \multicolumn{4}{c}{Model 3} \\ \hline
estimator	&	bias	&	SD	&	SE	&	RMSE	&	CP	&	bias	&	SD	&	SE	&	RMSE	&	CP	&	bias	&	SD	&	SE	&	RMSE	&	CP	\\ \midrule
$\hat{\tau}_{B}$	&	0.03	&	1.31	&	1.30&	1.31	&	0.95	&	0.03	&	0.87	&	0.88	&	0.87	&	0.95	&	0.00&	0.42	&	0.43	&	0.42	&	0.95	\\
$\hat{\tau}_{\adj,\ccov}$	&	0.01	&	1.31	&	1.39	&	1.31	&	0.96	&	0.02	&	0.83	&	0.84	&	0.83	&	0.95	&	0.00&	0.38	&	0.47	&	0.38	&	0.99	\\
$\hat{\tau}_{\int, \ccov}$	&	0.03	&	1.31	&	1.33	&	1.31	&	0.95	&	0.02	&	0.83	&	0.85	&	0.83	&	0.95	&	0.01	&	0.38	&	0.39	&	0.38	&	0.95	\\
$\hat{\tau}_{\tom, \ccov}$	&	0.01	&	1.31	&	1.33	&	1.31	&	0.95	&	0.02	&	0.83	&	0.85	&	0.83	&	0.95	&	0.00&	0.38	&	0.39	&	0.38	&	0.95	\\
$\hat{\tau}_{\adj, \ccov, \ss}$	&	-0.06	&	1.32	&	1.40&	1.32	&	0.96	&	0.03	&	0.84	&	0.84	&	0.84	&	0.95	&	-0.02	&	0.37	&	0.45	&	0.37	&	0.98	\\
$\hat{\tau}_{\int, \ccov, \ss}$	&	0.03	&	1.33	&	1.32	&	1.33	&	0.95	&	0.02	&	0.85	&	0.85	&	0.85	&	0.95	&	0.00&	0.37	&	0.37	&	0.37	&	0.95	\\
$\hat{\tau}_{\tom, \ccov, \ss}$	&	-0.06	&	1.32	&	1.31	&	1.32	&	0.95	&	0.04	&	0.84	&	0.84	&	0.84	&	0.95	&	-0.02	&	0.37	&	0.37	&	0.37	&	0.95	\\
$\hat{\tau}_{\adj, \imp}$	&	0.00&	1.24	&	1.28	&	1.24	&	0.96	&	0.01	&	0.66	&	0.64	&	0.66	&	0.94	&	0.00&	0.38	&	0.45	&	0.38	&	0.98	\\
$\hat{\tau}_{\int, \imp}$	&	0.02	&	1.24	&	1.39	&	1.24	&	0.97	&	0.01	&	0.66	&	0.75	&	0.66	&	0.97	&	0.00&	0.37	&	0.41	&	0.37	&	0.97	\\
$\hat{\tau}_{\tom, \imp}$	&	-0.01	&	1.24	&	1.38	&	1.24	&	0.97	&	0.01	&	0.66	&	0.75	&	0.66	&	0.97	&	-0.01	&	0.37	&	0.41	&	0.37	&	0.97	\\
$\hat{\tau}_{\adj, \imp, \ss}$	&	-0.04	&	1.10&	1.09	&	1.10&	0.95	&	0.05	&	0.68	&	0.67	&	0.69	&	0.94	&	-0.01	&	0.39	&	0.45	&	0.39	&	0.97	\\
$\hat{\tau}_{\int, \imp, \ss}$	&	0.00&	1.17	&	1.11	&	1.17	&	0.94	&	0.00&	0.75	&	0.71	&	0.75	&	0.94	&	0.00&	0.39	&	0.39	&	0.39	&	0.95	\\
$\hat{\tau}_{\tom, \imp, \ss}$	&	-0.08	&	1.12	&	1.09	&	1.12	&	0.94	&	0.03	&	0.71	&	0.70&	0.71	&	0.94	&	-0.03	&	0.38	&	0.38	&	0.38	&	0.95	\\
$\hat{\tau}_{\adj, \mim}$	&	-0.06	&	1.24	&	1.26	&	1.25	&	0.95	&	-0.05	&	0.65	&	0.60&	0.65	&	0.93	&	-0.06	&	0.38	&	0.42	&	0.39	&	0.96	\\
$\hat{\tau}_{\int, \mim}$	&	0.00&	1.26	&	1.69	&	1.26	&	0.99	&	-0.01	&	0.66	&	0.91	&	0.66	&	0.99	&	-0.01	&	0.37	&	0.47	&	0.37	&	0.99	\\
$\hat{\tau}_{\tom, \mim}$	&	-0.06	&	1.25	&	1.68	&	1.25	&	0.99	&	-0.04	&	0.65	&	0.90&	0.66	&	0.99	&	-0.06	&	0.37	&	0.47	&	0.37	&	0.99	\\
$\hat{\tau}_{\adj, \mim, \ss}$	&	-0.24	&	1.12	&	1.08	&	1.15	&	0.94	&	-0.15	&	0.69	&	0.64	&	0.71	&	0.92	&	-0.21	&	0.41	&	0.42	&	0.46	&	0.91	\\
$\hat{\tau}_{\int, \mim, \ss}$	&	-0.08	&	1.44	&	1.31	&	1.44	&	0.94	&	-0.06	&	0.88	&	0.83	&	0.89	&	0.94	&	-0.08	&	0.44	&	0.45	&	0.45	&	0.94	\\
$\hat{\tau}_{\tom, \mim, \ss}$	&	-0.30&	1.15	&	1.19	&	1.18	&	0.95	&	-0.18	&	0.73	&	0.76	&	0.75	&	0.94	&	-0.25	&	0.41	&	0.41	&	0.47	&	0.90\\ \bottomrule
\end{tabular}

\end{threeparttable}
\end{table}

\end{document}